\documentclass[12pt]{article}
\usepackage{xcolor}
\usepackage{amsmath,amssymb,amsthm,bbm}
\usepackage[colorlinks,citecolor=blue,urlcolor=blue,linkcolor=blue]{hyperref}
\usepackage{fullpage}
\usepackage{natbib}
\usepackage[onehalfspacing]{setspace}
\usepackage{dsfont}
\usepackage{tikz}
\usepackage{pgfplots}
\numberwithin{equation}{section}
\AtBeginDocument{%
   \setlength\abovedisplayskip{0.05in}
   \setlength\belowdisplayskip{0.08in}}

\newtheorem{assumption}{Assumption}
\newtheorem{definition}{Definition}
\newtheorem{lemma}{Lemma}
\newtheorem{proposition}{Proposition}
\newtheorem{corollary}{Corollary}
\newtheorem{theorem}{Theorem}

\newtheorem{example}{Example}
\newtheorem{obs}{Observation}
\newcommand{\E}{\mathbb{E}}
\newcommand{\R}{\mathbb{R}}

\newcommand{\m}[1]{\ensuremath{\boldsymbol {#1}}}

\title{Complexity and Misspecification\thanks{We thank Simone Cerreia-Vioglio, Piotr Dworczak, Giacomo Lanzani, Massimo Marinacci, Lasse Mononen, Arthur Robson, Larry Samuelson, Jakub Steiner,  and Tomasz Strzalecki for helpful comments, and NSF grant SES-2417162 for financial support.}}
\author{Drew Fudenberg\thanks{Department of Economics, MIT, drew.fudenberg@gmail.com} 
        \and Florian Mudekereza\thanks{Department of Economics, MIT, florianm@mit.edu}}
\date{\today}
\begin{document}

\maketitle
\thispagestyle{empty}

\begin{abstract}

We propose a tractable model of repeated decision problems that combines concern about model misspecification, as in robust control, with a complexity cost, such as Shannon entropy, that makes pessimistic beliefs trade off statistical plausibility against simplicity. In a static setting, stronger complexity aversion selects more concentrated worst-case beliefs and tilts choice toward actions whose adverse scenarios are harder to summarize with a simple narrative. In a dynamic learning environment,  complexity aversion can eliminate the endogenous cycles generated by misspecification concerns alone. We use the model to explain  scale heterogeneity in discrete choice, probability neglect, and home bias.

\end{abstract}

\newpage
\setcounter{page}{1}

\section{Introduction}

Economic agents routinely rely on simplified models to make decisions in complex
environments such as financial markets, insurance, and macroeconomic forecasting.
These models distill reality into a small number of parameters and probabilistic
relationships, but agents may worry that such structured representations are
only approximations of the true data-generating process. A large literature on robust
control, beginning with \citet{hansen01,hansen08}, addresses this concern by modeling
decision makers who choose actions to perform well under worst-case distortions of a
reference model. In  this framework,  concerns about misspecification are captured by allowing Nature to distort the reference model, with larger departures penalized according to relative entropy.  But these penalties depend only on statistical distance: they do not distinguish between a diffuse worst-case belief and a simple, concentrated adverse narrative.

In contrast, a
long tradition in statistics and learning theory dating back to Occam’s razor emphasizes a preference
for simpler explanations over complex ones \citep{box76,mdl78}. This raises a natural
question: how should an agent trade off concern about misspecification against a
preference for simple, plausible worst-case narratives? We answer this question by adding a  penalty on the entropy of the worst-case distortion. In the baseline model, we measure complexity using Shannon entropy, because it preserves the convex-duality structure underlying robust control and yields a closed-form Gibbs distortion.  Importantly, we penalize the entropy of worst-case distortions, not the richness of the model class itself, so simplicity disciplines the construction of adversarial explanations rather than restricting which structured models the agent considers.

Our analysis builds on \citeauthor{miss25}’s (\citeyear{miss25})  framework for decision making under model misspecification, which distinguishes uncertainty about models from concerns that all of the models being considered may be wrong. \citet{lanzani25} studies the dynamic implications of this approach when misspecification concerns adjust endogenously as  model fit changes. We extend this dynamic framework by adding a preference for simple worst-case beliefs: we penalize the entropy of the distortion, so among all equally pessimistic worst-case scenarios, Nature prefers the one that is most concentrated.  In the main text we suppose that the complexity penalty is proportional to the entropy of the distribution, so that the  worst-case beliefs admit a closed-form characterization with
a Gibbs structure, combining an exponential tilt in utilities with a power
transformation of the structured model.  However, this functional form is not essential; as we explain in Online Appendix \ref{app:general_complexity}, it suffices to use the more general additive perturbations of \citet{fudenberg2015stochastic}.

We first study this criterion in a static environment.  The key result in this setting is that increasing complexity aversion strictly reduces the entropy of worst-case distortions (except in degenerate cases). This creates a notable asymmetry: risky actions can be made pessimistic via concentrated, low-entropy distortions---such as overweighting a severe market crash---whereas safe actions cannot easily be made pessimistic this way; their adverse scenarios must remain diffuse. As a result, complexity aversion tilts choice toward actions that are ``safe'' in a precise sense: actions whose worst-case narratives are harder to describe simply. 
 
\par We then embed the static criterion into the dynamic learning environment of
\citet{lanzani25}. The agent updates beliefs about structured models using Bayes rule,
and their robustness parameter evolves with relative model fit. We show
that Lanzani’s dynamic selection result carries over to our setting: long-run empirical
frequencies of actions correspond to mixed robust equilibria under  the modified
criterion. Complexity aversion reshapes this equilibrium set by penalizing actions
whose worst-case explanations are simple and sharply concentrated.  In a two-action environment with a
safe and a risky arm, we show that sufficiently strong complexity aversion eliminates
all mixed robust equilibria. Thus, the endogenous \textit{cycles} between safe and risky actions that
arise under robustness alone disappear, and long-run behavior converges to a pure
action.
\par To illustrate, consider a central bank in a small open economy choosing between a hard currency peg (safe) and a managed float (risky). The float offers better stabilization when the bank’s structured model is reliable, but leaves the economy exposed to unexpected shocks when the model is misspecified. When the bank's model has a forecast failure, its  misspecification concern rises, forcing a retreat to the peg. This leads to a calm period in which realized outcomes are more in line with the bank’s model, reducing misspecification concerns and tempting a return to the float, which generates a cycle between   risky and safe regimes.
Such policy cycling is costly, as it undermines credibility and creates instability.  We show how a preference for simple worst-case narratives eliminates these cycles, ensuring long-run policy stability, and that  this stabilization can improve long-run realized payoffs.  
In our open-economy example, this implies the central bank can  enhance  welfare by prioritizing simpler economic narratives.  More generally, when multiple risky actions are available, complexity aversion
acts as a selection device: actions whose worst-case distortions have lower entropy are
eliminated first as the preference for simplicity increases.

\par In our framework, the simplicity motive does not play a uniform role across parameter regimes. When the simplicity motive is moderate relative to misspecification concerns, the preference remains within the \textit{average robust control} (ARC) family axiomatized by \citet{lanzani25w}: it is observationally equivalent to ARC, but with a transformed benchmark  within each model. However, when the simplicity motive is sufficiently strong, the criterion no longer has an ARC structure: evaluation collapses to the most adverse outcomes within each structured model, and further increases in the simplicity motive do not affect choices. When misspecification concerns vanish, the preference reduces to a subjective version of the \textit{entropy-modified expected utility} criterion axiomatized by \citet{mononen25}.

\par In addition to  stabilizing learning dynamics, our framework generates new explanations for several well-documented empirical phenomena.  First, we apply our framework to discrete choice problems by embedding complexity-augmented robustness into the rational inattention framework of \citet{matvejka15}. When the agent's worst-case narrative concentrates on specific adverse states, choice becomes less stochastic in  those states. This provides a microfoundation for scale heterogeneity in discrete choice \citep{swait93,fiebig10} and probability neglect in behavioral economics \citep{sunstein02,sunstein03}. 
Second, we apply our framework to the stochastic-growth setting of \citet{growth22}. Standard explanations of equity home bias are invariant to relabeling of loss states and therefore cannot distinguish assets that share the same reference payoff distribution. A preference for simple worst-case narratives breaks this invariance: a foreign asset whose losses are concentrated in a small number of focal scenarios is perceived as more fragile, generating home bias even without expected-return differences, hedging motives, or market segmentation. The same degree of misspecification can generate very different growth losses depending on how returns load on states, which matches  the empirical observation that large uncertainty shocks can persist with much smaller growth losses than standard models would imply  \citep[e.g.,][]{shock14,shock17}.

\subsection{Related Work}

Our work sits at the intersection of robust control, information theory, and 
learning theory. We highlight some key connections:

\paragraph{Robust control and misspecification.}

Our framework is grounded in the literature on variational preferences for decision-making under ambiguity \citep{mmr06}. \citet{hansen01,hansen08} introduced and developed the multiplier criterion, where agents form worst-case beliefs subject to a relative entropic penalty to account for potential misspecification of a reference model. \citet{strz11} provided the axiomatic foundations for this criterion within the class of variational preferences. Subsequently, \citet{miss25} provided a formal separation between  ambiguity aversion and fear of misspecification in a maxmin  setting, and also outlined a Bayesian version of their model. 
 \citet{lanzani25}  axiomatized a special case of this Bayesian robust approach for the multiplier criterion and developed its learning dynamics, where the concern for misspecification evolves endogenously with model fit. We extend Lanzani’s dynamic framework by incorporating a preference for simple worst-case beliefs: Instead of fearing  any statistically plausible distortion, our agent specifically fears distortions that are both plausible and parsimonious. Specifically, simplicity is penalized using  Shannon entropy in our baseline analysis. This is consistent with  \citeauthor{tomasz24}'s (\citeyear{tomasz24}) use of an entropic penalty in variational Bayes settings to microfound non-Bayesian updating rules such as the ``exponential'' updating rule documented in \citet{ben19} and \citet{rabin21}. More broadly, the use of worst-case beliefs as a model of ambiguity aversion originates with
\citet{gilboa89}. Our approach departs from multiple-prior models by endogenizing the
selection of pessimistic beliefs through relative entropy and complexity penalties, which yields
sharp comparative statics in both static and dynamic environments.

\paragraph{Expected utility and Complexity Aversion} An adjacent literature studies complexity aversion without model uncertainty. On the experimental side, \citet{bernheim20} shows that valuations of lotteries are sensitive to event-splitting manipulations,  \citet{enke23} shows that a measure of cognitive uncertainty  predicts systematic biases in risky choice,  \citet{enke25} extends this perspective to intertemporal choice, and  \citet{ortoleva25} provides evidence that  choice patterns are shaped by complexity considerations.

\par The closest theoretical paper is \citet{gabaix25}, which introduces ``first-order'' complexity aversion through imperfect signals about the optimal action.  We show that the high-simplicity version of our framework generates the same reduced-form objective as first-order complexity aversion in the paper's leading examples. \citet{mononen25} provides an axiomatic foundation for the entropy-modified expected utility using event-splitting considerations;  our framework nests a subjective version of this model. \citet{puri25} studies simplicity in risky choice through the support size of lotteries,  a different measure of complexity than entropy. A bit further afield, \citet{oprea20} studies the complexity of decision rules using automata theory,  and \citet{lipman95} models bounded rationality using coarse representations. 

\paragraph{Narratives and coarse reasoning.}
\citet{shiller17} defends the inclusion of ``narratives'' in formal economic analysis, defining them as simple, contagious explanations of events that drive decision-making. This focus on simplicity aligns with theories of coarse reasoning, such as \citet{mullainathan2008coarse}, where agents act on limited representations by grouping distinct states into broad categories. Our framework operationalizes this preference for simple narratives through an entropic penalty on worst-case beliefs. By favoring low-entropy distortions, our model captures an agent who is most concerned by simple, concentrated explanations of failure,  effectively ``coarsening'' the worst-case scenario.

\paragraph{Information-theoretic preferences.}
The entropic penalty on worst-case beliefs relates to rational inattention 
\citep{sims03} and information-theoretic models of choice \citep{matvejka15}. 
In those frameworks, agents incur information processing costs proportional to 
entropy reduction. Here, we interpret the complexity parameter as a shadow price on the entropy of 
worst-case scenarios, linking preferences for simple explanations 
to information capacity constraints.\footnote{Formally, since our baseline criterion adds entropy linearly to Nature’s minimization problem, it is the Lagrangian form of an equivalent problem in which worst-case beliefs are chosen subject to an upper bound on their entropy. Hence, our complexity parameter is naturally interpreted as the shadow price of one extra unit of information-processing capacity, available to construct the adversarial belief.}

\paragraph{Model selection and learning.}
While minimum description length (MDL; \citeauthor{mdl78}, \citeyear{mdl78}) 
is a popular complexity measure for model selection, it is designed for settings 
with fixed exogenous data. Our framework differs in two ways: (i) data is 
endogenously generated by the agent's repeated action choices, and (ii) complexity 
operates through the agent's construction of worst-case beliefs within a robust 
criterion, not through model selection itself. \citet[][Section 6.3]{gr19} surveys some challenges of using MDL under misspecification. More generally, the idea  
that simpler models tend to be more robust motivates our approach. For example, \citet{blumer87} 
shows that hypotheses with lower VC dimension incur smaller generalization error 
in learning problems. We operationalize this principle through entropy penalties 
on worst-case beliefs: agents prefer concentrated explanations for observed 
data, which embodies the Occam's razor intuition that simpler models are more 
defensible under misspecification.

\section{Static Environment}
\label{sec:static}
This section  introduces the static complexity-augmented criterion, characterizes the 
worst-case beliefs in closed form, and shows  how complexity aversion 
shapes pessimistic distortions.

\subsection{Primitives and running example}

Let $A$ be a finite set of actions and $Y$ a finite set of outcomes. For each
$a\in A$ and $y\in Y$, the agent's per-period utility is $u(a,y)\in\R$.

Following \citet{hansen01} and \citet{lanzani25}, a \emph{structured model} specifies action-dependent outcome distributions
$\m{q}:=(q_a)_{a\in A}\in\Delta(Y)^A$.\footnote{Throughout, we use bold notation for distributions in $\Delta(Y)^A$ to distinguish them from those in $\Delta(Y)$.} The agent entertains a finite set of such models
$Q\subset\Delta(Y)^A$ and holds a posterior belief $\pi\in\Delta(Q)$.

For any distributions $p,q\in\Delta(Y)$ with $p\ll q$, let
$$
R(p\Vert q)
:=\sum_{y\in Y} p(y)\log\frac{p(y)}{q(y)},
\qquad
H(p):=-\sum_{y\in Y} p(y)\log p(y)
$$
denote relative entropy or Kullback-Leibler (KL) divergence and Shannon entropy, respectively.

\begin{assumption}[Full support]\label{ass:full_support}
For each $\m{q}\in Q$, each $a\in A$, and each $y\in Y$, $q_a(y)>0$.
\end{assumption}

Throughout, we will use the following running  example of an agent who must choose between a constant ``safe'' payoff and a state-dependent ``risky''
payoff under misspecified models. The example is rich enough to illustrate how worst-case beliefs
are distorted and how the new parameter $\mu$ can flip the ranking between safe and risky,
yet simple enough that all expressions can be written in closed form and reused in the
dynamic analysis.

\begin{example}[Safe vs risky arm]\label{ex:running}
Let $A=\{r,s\}$, where $r$ is a \emph{risky} action and $s$ is a \emph{safe} action.
Let $Y=\{g,b\}$ with $g$ a ``good'' outcome and $b$ a ``bad'' outcome. Utilities are
$$
  u(r,g)=1,\quad u(r,b)=0,\qquad
  u(s,g)=u(s,b)=\bar u\in(0,1).
$$
Thus, $s$ delivers a constant payoff $\bar u$ and $r$ delivers a high payoff in $g$ and a
low payoff in $b$.

The set of structured models is $Q=\{\m{q}^H,\m{q}^L\}$, where for each $\m{q}\in Q$,
$$
  q_r^H(g)=p_H,\quad q_r^H(b)=1-p_H,\qquad
  q_r^L(g)=p_L,\quad q_r^L(b)=1-p_L,
$$
with $p_H>p_L$ so that $\m{q}^H$ is ``optimistic'' and $\m{q}^L$ is ``pessimistic'' about the risky
arm. Under the safe arm both models agree:
$
  q_s^H = q_s^L = \tfrac{1}{2}\delta_g+\tfrac{1}{2}\delta_b.$
We take an arbitrary full-support prior $\pi_0\in\Delta(Q)$, for instance
$\pi_0(\m{q}^H)=\pi_0(\m{q}^L)=1/2$.
\end{example}

\subsection{Robust criterion with complexity aversion}
We now introduce the complexity-augmented criterion. For each action $a$ and model $q$, Nature chooses a distorted distribution $p$ to minimize payoffs, subject to a KL penalty (robustness) and an entropic penalty (complexity).  Fix parameters $\lambda>0$ and $\mu\in\mathbb{R}$. For each action $a\in A$ and
model $\m{q}\in Q$, the agent evaluates $a$ by solving
\begin{equation}
\label{eq:static_criterion}
v_{\lambda,\mu}(a;\m{q})
:=
\min_{p_a\in\Delta(Y)}
\left\{
\E_{p_a}[u(a,y)]
+ \frac{1}{\lambda}R(p_a\Vert q_a)
 + \mu H(p_a)
\right\}.
\end{equation}
The term $R(p_a\Vert q_a)/\lambda$ in \eqref{eq:static_criterion} is the standard entropic penalty that discourages
Nature from choosing models $p_a$ that are too far from the structured model $q_a$, while the term $\mu H(p_a)$
penalizes or rewards the entropy of the distortion itself depending on the sign of $\mu$. When $\mu=0$, (\ref{eq:static_criterion}) becomes \citeauthor{hansen01}'s (\citeyear{hansen01}) multiplier criterion. For $\mu>0$, the adversary favors low-entropy (more
concentrated) distortions; for $\mu<0$, the adversary prefers high-entropy (more diffused)
distortions. Notice that the penalty on the entropy is consistent with the penalty in \citet[][eq. (6)]{tomasz24}.\footnote{\citet[][p. 6]{tomasz24}  also interprets $\mu>0$ as capturing an agent who prefers to have ``simple theories of the world.''} When $\lambda\rightarrow0$ and $\mu\neq0$, (\ref{eq:static_criterion}) becomes $\E_{q_a}[u(a,y)]+\mu H(q_a)$, which is a subjective version of  \citeauthor{mononen25}'s (\citeyear{mononen25}) entropy-modified expected utility preference that  captures an attitude toward complexity when the agent has no misspecification concerns. When $\lambda\rightarrow\infty$ and $\mu=0$, (\ref{eq:static_criterion}) becomes \citeauthor{gilboa89}'s (\citeyear{gilboa89}) maxmin criterion where the set of models is the entire simplex $\Delta(Y)$, which reflects an agent who is excessively concerned about misspecification. Section \ref{sec:role_mu} discusses the interpretation of $\mu$ and Online Appendix \ref{app:general_complexity} shows that all our main results extend to a more general class of complexity functionals that nests Shannon entropy $H(\cdot)$.

Given a posterior $\pi\in\Delta(Q)$, the overall value of action $a$ is
\begin{equation}\label{eq:posterior_value}
  V_{\lambda,\mu}(a;\pi)
  := \sum_{\m{q}\in Q} v_{\lambda,\mu}(a;\m{q}) \pi(\m{q}).
\end{equation}
The associated static choice rule chooses any action in
$$
  \operatorname{BR}^{\mu}_\lambda(\pi)
  :=
  \arg\max_{a\in A} V_{\lambda,\mu}(a;\pi).
$$

The inner minimization \eqref{eq:static_criterion}
represents an adversarial process: given a model $\m{q}$ and action $a$, Nature chooses a
distorted distribution $p_a$ that makes payoffs low, but pays a penalty for
moving away from $q_a$ and for using a complicated (high-entropy) explanation. The 
maximization over $a$ in $\operatorname{BR}^{\mu}_\lambda(\pi)$ captures the agent's robust choice
 given their posterior $\pi$ over $Q$. When $\mu=0$ (no complexity penalty), \eqref{eq:posterior_value} coincides with the \emph{average robust control} (ARC) criterion in  \citet[][]{lanzani25}, which is a special case of the  ``smooth Bayesian'' criterion in  \citet[][]{miss25}.

Thus, the new parameter $\mu$ trades off misspecification concerns
against complexity concerns: a higher $\mu$ means that among all ways of being pessimistic,
the agent prefers those with simpler probability distributions. Our complexity-augmented criterion in \eqref{eq:posterior_value} therefore provides a tractable and parsimonious unification of many existing decision criteria in the literature including \citet{gilboa89}, \citet{hansen01}, \citet{lanzani25}, and \citet{mononen25}. Section \ref{subsec:gabaix} shows how our framework relates to \citet{gabaix25}.

The next lemma characterizes the solution to \eqref{eq:static_criterion} in closed form.

\begin{lemma}\label{lem:closed_form}
Suppose Assumption \ref{ass:full_support} holds and define
$$
  \kappa := \frac{1}{\lambda}-\mu,\qquad
  \beta := \frac{1}{1-\lambda\mu}.
$$
If $\kappa>0$ (i.e., $\mu<1/\lambda$), then for every $(a,\m{q})\in A\times Q$, \eqref{eq:static_criterion} has a unique minimizer
\begin{equation}\label{eq:p_star_closed_form}
  \hat{p}_{\lambda,\mu}(a;\m{q})(y)
  = \frac{\exp\{-u(a,y)/\kappa\} q_a(y)^{\beta}}
         {\sum_{z\in Y}\exp\{-u(a,z)/\kappa\} q_a(z)^{\beta}}
  \qquad (y\in Y).
\end{equation}
Moreover,
$v_{\lambda,\mu}(a;\m{q})$ and $\hat{p}_{\lambda,\mu}(a;\m{q})$ are continuous in
$(\lambda,\mu,\m{q})$ on $\{\kappa>0\}$, and the following  envelope identities hold for all $(\lambda,\mu,\m{q})$ with $\kappa>0$:
\begin{equation}\label{eq:envelope_mu}
  \frac{\partial}{\partial\mu}v_{\lambda,\mu}(a;\m{q})
  = H\big(\hat{p}_{\lambda,\mu}(a;\m{q})\big),
  \qquad
  \frac{\partial}{\partial\lambda}v_{\lambda,\mu}(a;\m{q})
  = -\frac{1}{\lambda^2} 
    R\big(\hat{p}_{\lambda,\mu}(a;\m{q})\Vert q_a\big).
\end{equation}
\end{lemma}

The worst–case distortion $\hat p_{\lambda,\mu}(a;\m{q})$ has a Gibbs form: it reweights
model $q_a$ by a payoff tilt $e^{-u(a,y)/\kappa}$ and a power–prior term $q_a(y)^\beta$.
When $\mu=0$, $\beta=1$, so we recover the standard entropic distortion
$\hat p_{\lambda,0}\propto q_a e^{-\lambda u}$. When $\mu>0$, the coefficient $\kappa$
shrinks and $\beta>1$, so the distortion places more weight on low–utility outcomes and
leans more heavily on model $q_a$, resulting in a worst–case belief that is more concentrated and typically also more
pessimistic.\footnote{Higher $\mu$ does not always increase pessimism: if $q_a$ is already concentrated on high-utility states, a preference for simple narratives can amplify this prior concentration, which can make the worst-case belief less pessimistic.}

The envelope identities  show that $\mu$ and $\lambda$ operate through different channels.
The marginal effect of $\mu$ on the worst–case value equals the entropy of the distortion,
while the marginal effect of $\lambda$ is proportional to minus the KL divergence from $q_a$.
Thus, misspecification allows the worst–case belief to move farther
from the reference model,  complexity concerns restrict how diffuse
that belief can be.

The next result shows that increasing $\mu$ strictly reduces the entropy of the worst-case
distortion except in the knife–edge case where $u(a,\cdot)$ is affine in $\log q_a(\cdot)$.
\begin{proposition}[Entropy strictly decreasing in complexity aversion]
\label{prop:entropy_decreasing_mu}
Fix $\lambda > 0$. For each $\mu < 1/\lambda$ and each $(a,\m{q})$, let 
$\hat{p}_{\lambda,\mu}(a;\m{q})$ be the worst-case distortion from 
Lemma \ref{lem:closed_form}. Assume that for every $(a,\m{q})$ there are no constants $b_1, b_2$ such that  payoffs satisfy: 
$u(a,y) = b_1 + b_2 \log q_a(y)$ for all $y \in Y$. Then, $\mu \mapsto H(\hat{p}_{\lambda,\mu}(a;\m{q}))$ is strictly decreasing.
\end{proposition}

  Proposition \ref{prop:entropy_decreasing_mu} shows
that $\mu$ really is a complexity parameter: as the agent becomes more complexity averse
(higher $\mu$), the ``worst-case story'' about outcomes becomes simpler and more
concentrated. The condition that $u(a,\cdot)$ and $q_a$ are nondegenerate rules out knife-edge
cases where the distortion cannot move probability mass in a meaningful way. In those
generic cases, the adversary optimally sacrifices entropy to make low-payoff states more
likely when the agent puts more weight on simplicity.

\subsection{Static illustration in the running example}

This section illustrates the static effect of $\mu$ using the safe-versus-risky environment from
Example \ref{ex:running}. For concreteness, we fix some specific values for our primitives: $\lambda=1$, $p_H=\frac{7}{10}$, $p_L=\frac{3}{10}$,
$\bar u=\frac{7}{12}$, and consider a single model $\m{q}$ for the risky arm with $q_r(g)=\frac{7}{10},q_r(b)=\frac{3}{10}$, so $\pi$ is degenerate at $\m{q}$.

\begin{example}\label{ex:static_running}
\noindent
\emph{--- Risky arm.} For the risky action $r$, the worst-case distortion
$\hat{p}_{\lambda,\mu}(r;\m{q})$ is given by \eqref{eq:p_star_closed_form}. A simple numerical
calculation yields:
\begin{itemize}
  \item $\mu=0$ (benchmark). Here $\kappa=1,\beta=1$ and
  $$
    \hat{p}_{\lambda,\mu}(r;\m{q})(g)\approx 0.462,\qquad \hat{p}_{\lambda,\mu}(r;\m{q})(b)\approx0.538.
  $$
  \item $\mu=\frac{4}{10}$ (complexity aversion). Here $\kappa=0.6,\beta\approx1.667$ and
  $$
    \hat{p}_{\lambda,\mu}(r;\m{q})(g)\approx 0.437,\qquad \hat{p}_{\lambda,\mu}(r;\m{q})(b)\approx0.563.
  $$
  \item $\mu=-\frac{4}{10}$ (complexity love). Here $\kappa=1.4,\beta\approx0.714$ and
  $$
    \hat{p}_{\lambda,\mu}(r;\m{q})(g)\approx 0.473,\qquad \hat{p}_{\lambda,\mu}(r;\m{q})(b)\approx0.527.
  $$
\end{itemize}
As $\mu$ increases, the distortion becomes more concentrated and more pessimistic.

\medskip
\noindent
\emph{--- Safe arm.} For the safe action $s$, the payoff $u(s,y)=\bar u$ is constant across
states. The Gibbs form \eqref{eq:p_star_closed_form} then implies that the payoff tilt
cancels and $\hat{p}_{\lambda,\mu}(s;\m{q})$ is uniform on $\{g,b\}$ for all $\mu$. Its entropy
is therefore maximal and independent of $\mu$, whereas for the risky arm the worst-case
distribution becomes less entropic as $\mu$ rises.

\medskip
\noindent
\emph{--- Ranking of actions.}  $V_{\lambda,\mu}(r;\pi)$ and
$V_{\lambda,\mu}(s;\pi)$ are computed by plugging $\hat{p}_{\lambda,\mu}$ into
\eqref{eq:static_criterion} and averaging with respect to $\pi$. For the parameter values
above, we can verify numerically that:
\begin{itemize}
  \item for $\mu=0$ the risky arm is strictly preferred to the safe arm;
  \item for $\mu>0$ sufficiently large, the safe arm becomes strictly preferred because the
  worst-case distortion against $r$ becomes very concentrated on the bad outcome;
  \item for all $\mu<0$, the risky arm is more attractive because
  worst-case beliefs are forced to be diffuse, so there is a unique cutoff $\mu^*\in(0,\frac{4}{10})$ at which the ranking flips. 
\end{itemize}
\end{example}
This example illustrates that $\mu$ reweights the risky–safe tradeoff: 
complexity aversion allows simple (concentrated) worst-case narratives 
for risky outcomes, pushing toward safety.

\section{Dynamic Environment and Long-Run Behavior}
\label{sec:dynamic}

We embed our static criterion into the dynamic learning environment of  \citet{lanzani25}: the agent chooses actions over time,
updates beliefs over a finite set of structured models by Bayes rule, and adjusts their
misspecification concern based on relative model fit. 
\subsection{Dynamic environment}

As in the static environment, there is a finite set of actions $A$ and a finite outcome space $Y$.
The agent’s per-period payoff is given by a utility function
$u:A\times Y\to\mathbb{R}$.
The true data-generating process (DGP) is
$\m{p}^\star=(p^\star_a)_{a\in A}\in\Delta(Y)^A$. The agent entertains a finite set of structured models
$Q\subset\Delta(Y)^A$ and begins with a full-support prior $\pi_0\in\Delta(Q)$.
At each time $t \in \mathbb{N}$, the agent observes outcomes, updates the posterior $\pi_t$ via Bayes rule, and updates the misspecification concern $\lambda_t$ based on how much the best model in $Q$ underperforms relative to an unrestricted class of models. 
For $t\ge0$, let
$
  h_t=(a_1,y_1,\dots,a_t,y_t)\in (A\times Y)^t
$
denote the history up to time $t$, with $h_0=\varnothing$.
Let $H_t=(A\times Y)^t$ be the set of histories of length $t$,
$\mathcal{F}_t$ the sigma-field generated by $H_t$,
and $\mathcal{H}:=\bigcup_{t\ge0}H_t$ the set of all finite histories.

Following \citet{lanzani25}, the agent evaluates the adequacy of the structured model
class by comparing its likelihood to that of the unrestricted benchmark $\Delta(Y)^A$.\footnote{Since $Y$ is finite and all structured models have full support, we follow \citeauthor{lanzani25}'s (\citeyear[][Section 2.2]{lanzani25}) suggestion by choosing the entire simplex $\Delta(Y)^A$ as the set of ``alternative unstructured'' models.}
The agent's misspecification concern is allowed to depend on how well their models fit the data.  
Specifically, at each history $h_t$, the agent's  misspecification concern depends on the  log-likelihood ratio (LLR) between the structured model class 
$Q$ and the unrestricted model class $\Delta(Y)^A$:
$$
  \operatorname{LLR}(h_t; Q)
  := 
  -\log \frac{\max_{\m{q} \in Q} \prod_{\tau=1}^t q_{a_\tau}(y_\tau)}
             {\max_{\m{p} \in \Delta(Y)^A} \prod_{\tau=1}^t p_{a_\tau}(y_\tau)},
$$
for all $t$ and all $h_t\in H_t$. The numerator is the maximum likelihood of data under the structured models 
in $Q$, whereas the denominator is the maximum likelihood under unrestricted models. 
A large LLR indicates a poor fit of $Q$ to data, which suggests a high risk of misspecification. The agent adjusts this over  time according to
\begin{equation}\label{eq:rho_def}
      \lambda_t := \frac{1}{ct} \operatorname{LLR}(h_t; Q),\qquad (t>0); 
      \end{equation}
   define this updating rule as $\Lambda(h_t):=\frac{1}{ct}\operatorname{LLR}(h_t;Q)$, where $c>0$ is a constant.  Normalizing the LLR by time makes misspecification concern only respond to persistent per-observation model failures, rather than mechanically rising just because more data have accumulated. \citet{lanzani25} provides normative justifications for  this rule by showing that it is a rationality benchmark for repeated decision problems. 
   
      At each $t$, given misspecification concern $\lambda_t$, the agent evaluates actions using
the criterion of Section \ref{sec:static}, with  complexity parameter set to
\begin{equation*}
  \mu_t := \bar{\mu}\hspace{0.02in}\lambda_t, 
\end{equation*}
where $\bar{\mu}\ge 0$ is a fixed preference parameter.  The period-$t$ evaluation for action $a$ given model $\m{q}\in Q$ becomes 
\begin{equation}\label{eq:dynamic} 
v_{\lambda_t,\bar{\mu} \lambda_t}(a;\m{q}) = \min_{p_a \in \Delta(Y)} \left\{ \E_{p_a}[u(a,y)]  +  \frac{1}{\lambda_t} R(p_a \Vert q_a)  +  \bar{\mu} \lambda_t H(p_a) \right\}  \end{equation}
which is the dynamic analogue of \eqref{eq:static_criterion}. Scaling the complexity parameter proportionally to $\lambda_t$ ensures that complexity
aversion is active only when misspecification concerns are relevant.  In contrast, a fixed $\mu_t\equiv\mu > 0$ would prevent convergence to standard Bayesian choice when the agent is correctly specified, which can reduce welfare in the long run (Online Appendix \ref{sec:welfare}). Thus, setting $\mu_t = \bar{\mu}\lambda_t$ should be viewed as a reduced-form representation of the idea that complexity aversion is \textit{countercyclical}: it tightens when model fit deteriorates.

The entropic penalty affects behavior only through the shape of the worst-case beliefs. When $\lambda_t$ is small, the entropic penalty is negligible and the agent behaves approximately as a Bayesian, with worst-case beliefs approaching the structured model. When $\lambda_t$ is large, worst-case distortions are both more pessimistic and more concentrated, reflecting a preference for simple adversarial explanations. As $\lambda_t\to0$, both penalties vanish and the evaluation converges to Bayesian expected utility under $\pi_t$ (Proposition \ref{prop:correct_spec_limit}). Since the static minimization problem preserves its convex structure, the worst-case belief retains the same closed-form characterization as in Lemma \ref{lem:closed_form}, with $\mu$ replaced by $\mu_t$. Thus, all envelope and comparative-statics arguments carry over to the dynamic setting.

\begin{assumption}[Uniform positivity of $\kappa_t$]\label{ass:kappa_pos}
There is $\bar{\lambda}<\infty$ such that $0<\lambda_t\le \bar{\lambda}$ for all $t$
a.s. Fix $\bar{\mu} \geq 0$ such that $\bar{\mu} \bar{\lambda}^2 < 1$, so
that $\kappa_t := \frac{1}{\lambda_t} - \bar{\mu} \lambda_t > 0$ for all $t$.
\end{assumption}

Given a posterior $\pi_t\in\Delta(Q)$, the overall value of action $a$ at $(\lambda_t,\bar{\mu},\pi_t)$ is
$$
  V_{\lambda_t,\bar{\mu} \lambda_t}(a;\pi_t)
  := \sum_{\m{q}\in Q} v_{\lambda_t,\bar{\mu} \lambda_t}(a;\m{q}) \pi_t(\m{q}),
$$
where $v_{\lambda_t,\bar{\mu} \lambda_t}$ is defined by \eqref{eq:dynamic}.
For every $t$, the associated best-reply correspondence is
$$
  \operatorname{BR}^{\bar{\mu}}_{\lambda_t}(\pi_t)
  :=
  \arg\max_{a\in A} V_{\lambda_t,\bar{\mu} \lambda_t}(a;\pi_t).
$$
A (pure) \textit{policy} is a measurable $\sigma:\mathcal{H}\rightarrow A$ that specifies an
action for each history. Given a policy $\sigma$ and the true DGP $\m{p}^\star$, the induced
probability measure on histories is denoted $\mathbb{P}_{\sigma}$.

\begin{definition}[$\bar{\mu}$-optimal policies]\label{def:mu_optimal}
Fix $\bar{\mu} \geq 0$ satisfying $\bar{\mu} \bar{\lambda}^2 < 1$ and the updating rule
\eqref{eq:rho_def}. A policy $\sigma$ is \emph{$\bar{\mu}$-optimal} if, for every $t$ and
every history $h_t$,
$$
  \sigma_t(h_t)\in\operatorname{BR}^{\bar{\mu}}_{\lambda_t}(\pi_t),
$$
where $\lambda_t=\Lambda(h_t)$, $\mu_t = \bar{\mu} \lambda_t$, and $\pi_t$ is the posterior
induced by history $h_t$.
\end{definition}

  A $\bar{\mu}$-optimal policy is simply one that
behaves myopically optimally in every period according to the $(\lambda_t,\mu_t)$-criterion
with $\mu_t = \bar{\mu} \lambda_t$. The agent updates beliefs and $\lambda_t$ as in
\citet{lanzani25}, but when choosing actions they  now care about the complexity of worst-case
explanations, not only about their fit. This makes the dynamic impact of $\bar{\mu}$
transparent: any differences in long-run behavior compared to the baseline model are due to
this changed static objective, not to any new learning assumption.

\begin{definition}[$\Lambda$-limit frequencies]\label{def:limit_freq}
Let $\sigma$ be a $\bar{\mu}$-optimal policy and let $\alpha_t(h_t)\in\Delta(A)$ denote the
empirical frequency of actions up to time $t$ on history $h_t$:
$$
  \alpha_t(h_t)(a)
  :=
  \frac{1}{t}\sum_{\tau=1}^t \mathbbm{1}\{a_\tau=a\},
  \qquad (a\in A).
$$
A random mixed action $\alpha^\Lambda\in\Delta(A)$ is a \emph{$\Lambda$-limit frequency} if there
exists a $\bar{\mu}$-optimal policy $\sigma$ such that, for every $\varepsilon>0$,
$$
  \mathbb{P}_{\sigma}\Big(
    \limsup_{t\to\infty}
    \big\|\alpha_t(h_t)-\alpha^\Lambda\big\|_1 \le \varepsilon
  \Big)
  > 0.
$$
Equivalently, with positive probability the empirical frequencies get and remain arbitrarily
close to $\alpha^\Lambda$ along the realized history.
\end{definition}

\subsection{Mixed $c$-robust equilibria for the $\bar{\mu}$-criterion}

We adapt Lanzani's static concept of mixed $c$-robust equilibrium to the
$\bar{\mu}$-criterion. Fix the true DGP $\m{p}^\star$ and the set of structured models $Q$.
For any mixed action $\alpha\in\Delta(A)$, let
$$
  D(\m{q};\alpha)
  :=
  \sum_{a\in A}\alpha(a) R\big(p^\star_a\Vert q_a\big)
$$
denote the average model misfit of $\m{q}$ under $\alpha$, and define the set of best-fit
models
$$
  Q(\alpha)
  :=
  \arg\min_{\m{q}\in Q} D(\m{q};\alpha).
$$

\begin{definition}[Mixed $c$-robust equilibrium]
\label{def:mixed_equilibrium_mu}
A triple $(\alpha,\eta,\tau)\in\Delta(A)\times\Delta(Q)\times\R_+$ is a
\emph{mixed $c$-robust equilibrium for the $\bar{\mu}$-criterion} if
\begin{enumerate}
  \item $\text{\normalfont supp}(\eta)\subset Q(\alpha)$;
  \item for every $a\in A$ with $\alpha(a)>0$,
  $a\in\operatorname{BR}^{\bar{\mu}}_\tau(\eta)$, where the best reply is evaluated at
  $(\lambda,\mu) = (\tau, \bar{\mu} \tau)$;
  \item the misspecification concern $\tau$ equals the average misfit per unit $c>0$:
  $$
    \tau
    =
    \frac{1}{c}\min_{\m{q}\in Q}D(\m{q};\alpha)
    =
    \frac{1}{c}D(\m{q};\alpha)
    \quad\text{for all }\m{q}\in\text{\normalfont supp}(\eta).
  $$
\end{enumerate}
\end{definition}

 A mixed $c$-robust equilibrium is a self-consistent
long-run configuration. The mixed action $\alpha$ describes the frequencies with which the
agent plays each action; $Q(\alpha)$ is the set of models that best fit the data generated
by $\alpha$; the belief $\eta$ puts weight only on those well-fitting models; and $\tau$ is
the strength of misspecification concern suggested by the data. 
 Relative to
Lanzani's framework, the only change is that best replies are evaluated using the
complexity-augmented criterion.

The next result extends \citet[][Theorem 3]{lanzani25} to the $\bar{\mu}$-criterion.

\begin{proposition}
\label{prop:dynamic_selection}
Fix $\bar{\mu} \geq 0$ satisfying $\bar{\mu} \bar{\lambda}^2 < 1$ and a
$\bar{\mu}$-optimal policy $\sigma$. Let $\alpha^\Lambda$ be any $\Lambda$-limit frequency of
$\sigma$. Then, there exist
$\eta^{\bar{\mu}}\in\Delta(Q)$ and $\tau^{\bar{\mu}}\ge0$ such that
$(\alpha^\Lambda,\eta^{\bar{\mu}},\tau^{\bar{\mu}})$ is a mixed $c$-robust equilibrium for the
$\bar{\mu}$-criterion.
\end{proposition}

Proposition \ref{prop:dynamic_selection} says that the empirical behavior of a $\bar{\mu}$-optimal agent over the long run is described by a
mixed $c$-robust equilibrium, so only the static notion of ``best reply'' has changed. Thus, any new dynamic
phenomena we derive will be due to how $\bar{\mu}$ reshapes the set of mixed $c$-robust
equilibria.

\subsection{Correct specification and the Bayesian limit}
\label{subsec:correct_specification}

We now describe what the dynamic framework implies when the agent is \emph{correctly
specified}. In this case, the structured family $Q$ contains the true model $\m{p}^\star$, so the LLR statistic that determines
$\lambda_t$ eventually stops detecting misspecification.

\begin{assumption}[Correct specification]\label{ass:correct_spec}
There exist $\m{q}^\star\in Q$ such that $q^\star_a = p^\star_a$,
  for all $a\in A.$
\end{assumption}

Assumption \ref{ass:correct_spec} says that the agent's structured models are rich enough
to include a model $\m{q}^\star$ that coincides with the true DGP. Recall that statistic $\lambda_t=\Lambda(h_t)$ measures
how much better the best model in $Q$ fits the data than the best model in $\Delta(Y)^A$. Under
correct specification, $Q$ does \emph{not} systematically underperform relative to $\Delta(Y)^A$,
so the LLR statistic has no reason to grow, and therefore misspecification concerns will vanish in the long run.

\begin{proposition}
\label{prop:correct_spec_limit}
Suppose Assumptions \ref{ass:full_support}, \ref{ass:kappa_pos}, and
\ref{ass:correct_spec} hold. Fix $\bar{\mu} \geq 0$ satisfying
$\bar{\mu} \bar{\lambda}^2 < 1$. Let $\sigma$ be any $\bar{\mu}$-optimal policy and
let $\alpha^\Lambda$ be any $\Lambda$-limit frequency of $\sigma$. Then, on the event
$\{\alpha_t\to \alpha^\Lambda\}$, for every action $a\in A$ with $\alpha^\Lambda(a)>0$, $V_{\lambda_t,\bar{\mu}\lambda_t}(a;\pi_t)
  \longrightarrow
  \E_{p^\star_a}[u(a,y)]$ $\mathbb{P}_\sigma\text{-a.s.}$
\end{proposition}

The mechanism has two parts. First, by Proposition \ref{prop:dynamic_selection}, the
$\Lambda$-limit frequency $\alpha^\Lambda$ corresponds to a mixed $c$-robust equilibrium
$(\alpha^\Lambda,\eta^{\bar{\mu}},\tau^{\bar{\mu}})$, and the long-run posterior $\pi_t$
concentrates on the set of best-fitting models $Q(\alpha^\Lambda)$. For actions played
infinitely often, the true model $\m{q}^\star$ must belong to $Q(\alpha^\Lambda)$ under correct
specification, so the posterior eventually assigns all weight to models that coincide with
$\m{p}^\star$ on the played actions. Second, because $\mu_t = \bar{\mu} \lambda_t \to 0$ as
$\lambda_t \to 0$, both the robustness penalty and the complexity penalty vanish.
Consequently, the worst-case distortion $\hat{p}_{\lambda_t,\mu_t}(a;\m{q})$ converges to
$q_a = p^\star_a$ for any $\m{q} \in Q(\alpha^\Lambda)$, and the agent's evaluation coincides
with Bayesian expected utility.  
 The specification
$\mu_t = \bar{\mu} \lambda_t$ thus serves as ``cognitive scaffolding'': complexity
constraints are instrumental tools for learning that are removed once the true model is
identified on path.

\subsection{Complexity aversion and elimination of $\lambda$-cycles}

We specialize to the two-action case $A=\{r,s\}$ of the running example and show that
complexity aversion can eliminate $\lambda$-driven cycles. The key object is the value
difference between $r$ and $s$. Fix $(\lambda,\mu)$ with $\mu = \bar{\mu} \lambda$ and $\kappa=1/\lambda-\mu>0$, a posterior
$\pi\in\Delta(Q)$, and distinguished actions $r,s\in A$. Define the value difference
\begin{equation}\label{eq:Delta_def}
  \Delta(\lambda,\bar{\mu},\pi)
  :=
  V_{\lambda,\bar{\mu}\lambda}(r;\pi) - V_{\lambda,\bar{\mu}\lambda}(s;\pi).
\end{equation}
For $a\in\{r,s\}$, define the average entropy of the worst-case
distortions as
\begin{align}
  H_a(\lambda,\bar{\mu},\pi)
    &:= \sum_{\m{q}\in Q}  \pi(\m{q}) H\big(\hat{p}_{\lambda,\bar{\mu}\lambda}(a;\m{q})\big).
    \label{eq:Ha_def}
\end{align}
Define the \emph{switching surface} as
\begin{align}\label{eq:switch}
  S
  :=
  \big\{(\lambda,\bar{\mu},\pi): \Delta(\lambda,\bar{\mu},\pi)=0\big\},
\end{align}
i.e., the set of parameters where the agent is indifferent between $r$ and $s$.

\begin{lemma}\label{lem:Delta_derivatives}
Suppose Assumption \ref{ass:full_support} holds. Then:
\begin{enumerate}
  \item For each $(\lambda,\bar{\mu},\pi)$ with $\kappa=1/\lambda-\bar{\mu}\lambda>0$, the
  partial derivative $\partial\Delta/\partial\bar{\mu}$ exists and is given by
  $\frac{\partial}{\partial\bar{\mu}}\Delta(\lambda,\bar{\mu},\pi)
      = \lambda \big( H_r(\lambda,\bar{\mu},\pi) - H_s(\lambda,\bar{\mu},\pi) \big).$
  \item In particular, on the switching surface $S$ in eq. (\ref{eq:switch}), the sign of
  $\partial\Delta/\partial\bar{\mu}$ equals the sign of the entropy difference
  $H_r(\lambda,\bar{\mu},\pi)-H_s(\lambda,\bar{\mu},\pi)$.
\end{enumerate}
\end{lemma}

  Lemma \ref{lem:Delta_derivatives} shows that the
marginal effect of complexity aversion on the risky-safe tradeoff is governed entirely by
the difference in worst-case entropies. If the worst-case explanation of the risky arm is
more concentrated (lower entropy) than that of the safe arm, then increasing $\bar{\mu}$
tilts the value difference in favor of the safe arm. The factor $\lambda$ reflects the fact
that the effective complexity penalty is $\mu = \bar{\mu} \lambda$: at higher levels of
misspecification concern, the same increase in $\bar{\mu}$ has a larger effect on the
value difference. More formally, the dynamic evolution $\mu_t=\bar{\mu}\lambda_t$ and the $\lambda$-factor in Lemma \ref{lem:Delta_derivatives} formalize the idea that complexity aversion is endogenously countercyclical: when data fit is poor and $\lambda_t$ is high, the simplicity constraint is tighter. To obtain a clean dynamic selection result, we impose two assumptions.

\begin{assumption}[Equilibrium range of $\lambda$]\label{ass:U}
There exists a compact interval $[\underline{\lambda},\bar{\lambda}] \subset (0,\infty)$
with the following property: for every parameter $\bar{\mu} \geq 0$
satisfying $\bar{\mu} \bar{\lambda}^2 < 1$ and every mixed $c$-robust equilibrium
$(\alpha,\eta,\tau)$ for the $\bar{\mu}$-criterion in the sense of
Definition \ref{def:mixed_equilibrium_mu} such that
$\alpha(r) > 0$, 
the associated misspecification concern $\tau$ belongs to
$[\underline{\lambda},\bar{\lambda}]$.
\end{assumption}

Assumption \ref{ass:U} restricts attention to
environments in which, whenever both 
 $r$ and $s$ are in the support of the equilibrium
randomization, the associated misspecification concern $\tau$ is neither
arbitrarily small nor arbitrarily large. Instead, $\tau$ always lies in a fixed interval $[\underline{\lambda},\bar{\lambda}]$. It holds whenever the true DGP lies in the interior of the convex hull of $Q$, so that no structured model can fit the long-run data arbitrarily well or arbitrarily poorly. 

\begin{assumption}[Uniform entropy gap]\label{ass:uniform_entropy_gap}
Let $[\underline{\lambda},\bar{\lambda}]$ be as in Assumption \ref{ass:U}.
There exist constants $\bar{\mu}_0 \ge 0$ and $H^{\ast} > 0$ such that for all
$(\lambda,\pi) \in [\underline{\lambda},\bar{\lambda}] \times \Delta(Q)$ and all
$\bar{\mu} \ge \bar{\mu}_0$ satisfying $\bar{\mu}\lambda^2<1$, the following implication holds: $\Delta(\lambda,\bar{\mu},\pi)\ge 0
 \Longrightarrow 
  H_s(\lambda,\bar{\mu},\pi) - H_r(\lambda,\bar{\mu},\pi) \ge H^{\ast}.$
\end{assumption}

By Lemma \ref{lem:Delta_derivatives}, Assumption \ref{ass:uniform_entropy_gap} implies
$\frac{\partial}{\partial \bar{\mu}}\Delta(\lambda,\bar{\mu},\pi)\le -\lambda H^\ast$  when $\Delta(\lambda,\bar{\mu},\pi)\ge0$. It imposes a uniform lower bound on the entropy gap
$H_s(\lambda,\bar{\mu},\pi)-H_r(\lambda,\bar{\mu},\pi)$ over the equilibrium range
$\lambda\in[\underline{\lambda},\bar{\lambda}]$ and all admissible $(\bar{\mu},\pi)$. In particular,
whenever the agent is indifferent between risky and safe, the worst-case explanation of the safe arm is
strictly more complex than that of the risky arm, and the difference is uniformly bounded.
 In the running example, this is natural: the safe arm has
state-independent payoffs, so pessimism about $s$ cannot be generated by concentrating on
any particular state; pessimistic narratives about $s$ must remain relatively diffuse. By
contrast, pessimism about $r$ can be implemented by piling probability on the bad state,
yielding a low-entropy worst-case distribution. Thus, at any point of indifference, the
safe arm is supported by a ``richer'' worst-case story.

\begin{theorem}[Complexity aversion eliminates $\lambda$-cycles]
\label{thm:mu_kills_cycles}
Let the
misspecification concern evolve according to the normalization rule
$\lambda_t = \Lambda(h_t)$ in \eqref{eq:rho_def}. Suppose $A=\{r,s\}$ and
Assumptions \ref{ass:U} and \ref{ass:uniform_entropy_gap} hold. Then, there exists
$\bar{\mu}^{\ast}>0$ such that the following holds for every
$\bar{\mu}\in[\bar{\mu}^{\ast},1/\bar{\lambda}^{2})$:
\begin{enumerate}
  \item[\textnormal{(i)}] Every mixed $c$-robust equilibrium
  $(\alpha^{\bar{\mu}},\eta^{\bar{\mu}},\tau^{\bar{\mu}})$ for the $\bar{\mu}$-criterion
  satisfies $\alpha^{\bar{\mu}}(s)=1$.
  \item[\textnormal{(ii)}] Thus, for
  all $\bar{\mu}\in[\bar{\mu}^{\ast},1/\bar{\lambda}^{2})$ every $\Lambda$-limit frequency $\alpha^\Lambda$ of any
  $\bar{\mu}$-optimal policy is pure and equals the safe action. In particular, the
  $\lambda$-driven cycles between $r$ and $s$ that arise in \citeauthor{lanzani25}'s (\citeyear{lanzani25}) $\bar{\mu}=0$
  benchmark cannot occur when $\bar{\mu}$ is sufficiently large.
\end{enumerate}
\end{theorem}

  When $\bar{\mu}=0$, \citet[][Section 4]{lanzani25} shows that the dynamic
adjustment of $\lambda_t$ can generate mixed $c$-robust equilibria in which the agent
endogenously cycles between risky and safe: bad model fit after $r$ raises $\lambda_t$ and
push toward $s$, while good model fit after $s$ lowers $\lambda_t$ and favors 
$r$. With $\bar{\mu}>0$, the static comparison at any point of indifference is tilted
toward the safe arm because its worst-case explanation is more complex.
Theorem \ref{thm:mu_kills_cycles} shows that, once $\bar{\mu}$ passes a threshold,
these small tilts accumulate: mixed equilibria with both actions in support disappear, and
only pure equilibria survive. The long-run empirical frequencies therefore put all mass on the safe action.
This  shows that a preference for simple worst-case narratives can
eliminate dynamic instability driven by misspecification concerns alone.

If Assumption \ref{ass:uniform_entropy_gap} is dropped, the complexity penalty need not tilt
incentives toward safety: the safe arm may admit an equally simple (low-entropy) worst-case narrative,
or the entropy ranking may reverse, so increasing $\bar{\mu}$ need not shrink (and may even expand)
the region of indifference and mixing. In that case, the endogenous adjustment of $\lambda_t$ can
continue to sustain mixed play and $\lambda$-cycles even for large $\bar{\mu}$.

Figure \ref{fig:mu_kills_cycles_illustration}  illustrates the conclusion of  Theorem \ref{thm:mu_kills_cycles} in the environment of Example \ref{ex:running}. We consider a posterior
$\pi_\theta$ with $\pi_\theta(\m{q}^H)=\theta$ and $\pi_\theta(\m{q}^L)=1-\theta$, and plot the value
difference $\Delta(\lambda,\bar{\mu},\theta):=\Delta(\lambda,\bar{\mu},\pi_{\theta})$ in \eqref{eq:Delta_def} as a function of $\lambda$.  

\begin{figure}[hbt!]
    \centering
    \includegraphics[width=1\textwidth]{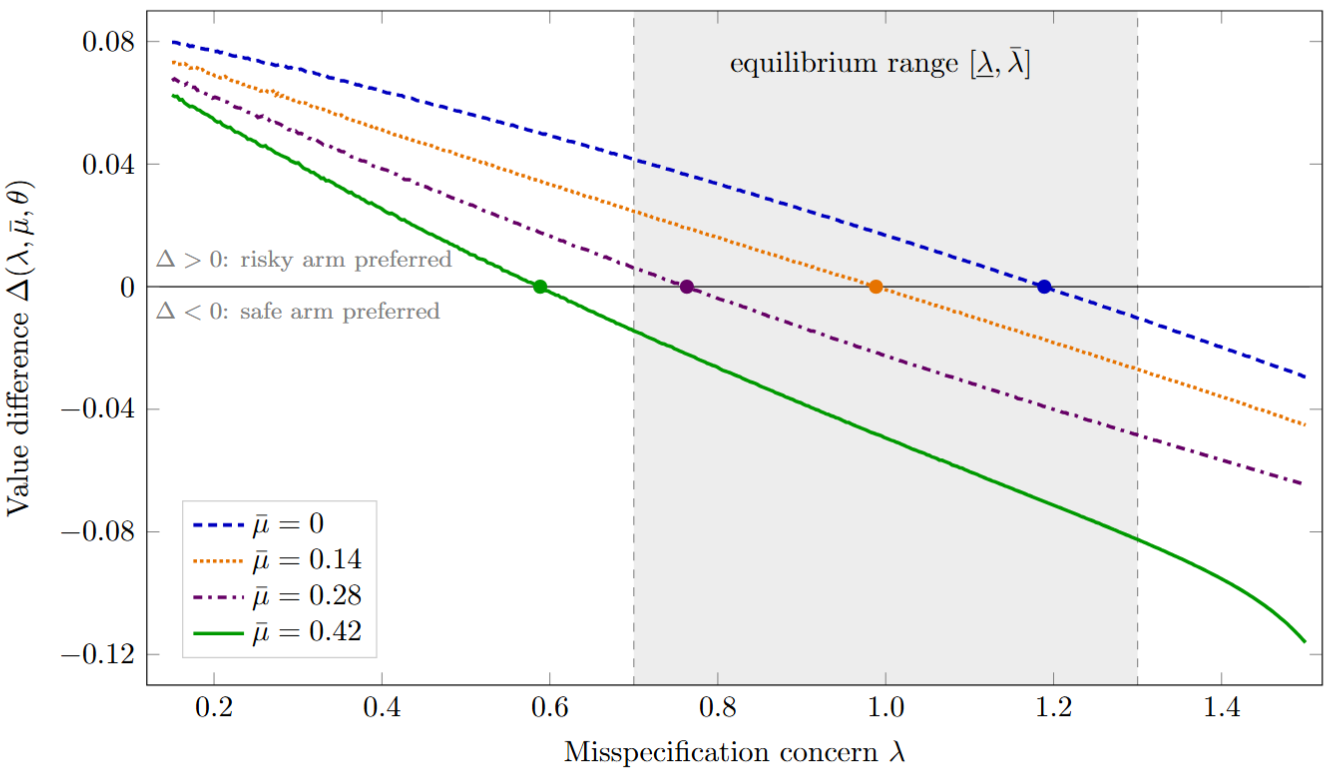} 
    \caption{Stronger complexity aversion ($\bar{\mu}$) pushes the indifference 
point $\Delta=0$ outside the equilibrium range (shaded), eliminating 
$\lambda$-cycles. Parameters: $p_H=0.85$, $p_L=0.05$, $\theta=0.90$, 
$\bar{u}=0.68$, $[\underline{\lambda},\bar{\lambda}]=[0.70,1.30]$.}
\label{fig:mu_kills_cycles_illustration}
\end{figure}
 Above the horizontal axis,  $\Delta(\lambda,\bar{\mu},\theta)>0$, so the risky arm is preferred; below it,   $\Delta(\lambda,\bar{\mu},\theta)<0$, so the safe arm is preferred; and where the curves cross the axis, the agent is indifferent between the two arms. The curve with $\bar{\mu}=0$ is \citeauthor{lanzani25}'s  (\citeyear{lanzani25}) benchmark with no complexity aversion. Since this benchmark curve crosses the horizontal axis within the equilibrium range of $\lambda$, there are equilibrium levels of misspecification concern at which both arms can be supported, so the endogenous cycles between safe and risk arms due to robustness can arise. The next two curves, corresponding to larger but still moderate values of $\bar{\mu}$, shift the indifference point to the left, but not enough to rule out cycles. In contrast, at $\bar{\mu}=0.42$, the indifference point lies strictly to the left of the equilibrium range, so the curve is already below the horizontal axis when it enters that range, which means that only the safe arm can be supported in equilibrium.
 \subsubsection{Economic Interpretation of Assumption \ref{ass:uniform_entropy_gap}}

Online Appendices \ref{sec:lanzani25_taxation} and \ref{sec:chamb20_entropy_gap} illustrate why Assumption \ref{ass:uniform_entropy_gap} is natural in the taxation example of \citeauthor{lanzani25} (\citeyear{lanzani25}) and the portfolio-choice environment of \citeauthor{chamb20} (\citeyear{chamb20}). In the taxation example, the agent chooses between not working and earning a positive target gross income under a progressive tax schedule. Working increases income, but it also increases exposure to mistakes about the tax schedule. Assumption \ref{ass:uniform_entropy_gap} captures this asymmetry: tax misspecification is more of a concern if  the agent works. In the portfolio-choice example, the investor chooses between a risk-free asset and a  risky asset, so terminal wealth is either fixed or fully exposed to the realized return path. The risky portfolio offers upside, but it also magnifies losses after sequences of poor returns. Assumption \ref{ass:uniform_entropy_gap} captures that caution is most relevant for the portfolio exposed to the risk of bad returns. 

\subsubsection{Extreme environments}

Theorem \ref{thm:mu_kills_cycles} may hold vacuously in extreme environments where
the risky arm is so attractive that even the largest admissible degree of complexity aversion
cannot overturn its advantage over the safe arm. In those cases, the interval
$[\bar{\mu}^\ast,1/\bar{\lambda}^2)$ in Theorem \ref{thm:mu_kills_cycles} may be empty: the admissibility restriction
$\bar{\mu}\bar{\lambda}^2<1$ binds before complexity aversion has enough capacity to affect the
risky-safe arm comparison. 
\par A natural way to rule out these extreme environments is to bound the risky arm's
\emph{initial} advantage at the baseline $\bar{\mu}_0$ from Assumption
\ref{ass:uniform_entropy_gap}. Let
$$
K_0
:=
\sup_{(\lambda,\pi)\in[\underline{\lambda},\bar{\lambda}]\times\Delta(Q)}
\max\{\Delta(\lambda,\bar{\mu}_0,\pi),0\}.
$$
By Lemma \ref{lem:Delta_derivatives} and Assumption \ref{ass:uniform_entropy_gap},
whenever $\Delta(\lambda,\bar{\mu},\pi)\geq 0$, increasing $\bar{\mu}$ lowers the risky arm's advantage $\Delta$ at
rate at least $\underline{\lambda}H^\ast$. Thus, before admissibility binds, the minimum guaranteed total tilt that complexity aversion can generate is
$\underline{\lambda}H^\ast(1/\bar{\lambda}^2-\bar{\mu}_0)$.

\begin{assumption}[Initial risky advantage not too large]
\label{ass:initial_risky_advantage}
$K_0
<
\underline{\lambda}H^\ast(1/\bar{\lambda}^2-\bar{\mu}_0).$
\end{assumption}

Assumption \ref{ass:initial_risky_advantage} says that the risky arm's initial advantage at $\bar{\mu}_0$ is small enough that the simplicity motive
can affect behavior before reaching the admissibility limit $1/\bar{\lambda}^2$.

\begin{corollary}\label{cor:nonvacuity_theorem1}
In Theorem \ref{thm:mu_kills_cycles}, suppose Assumption
\ref{ass:initial_risky_advantage} also holds. Then, the threshold
$\bar{\mu}^\ast$ can be chosen in $(\bar{\mu}_0,1/\bar{\lambda}^2)$. Hence, the interval
$[\bar{\mu}^\ast,1/\bar{\lambda}^2)$ is nonempty.
\end{corollary}

\subsection{Complexity aversion and switching arms}
Having established the stabilizing role of complexity aversion in the binary safe-versus-risky case, we now turn to richer choice sets. When several risky arms are available, with or without a safe benchmark, complexity aversion discourages switching and helps select among the available alternatives.
\subsubsection{Safe arm and multiple risky arms}

We extend the running example (Example \ref{ex:running}) to an environment with one safe arm $s$ and 
$r_1,\dots,r_K$ risky arms, for $K\geq2$. We continue to assume that $Q=\{\m{q}^H,\m{q}^L\}$ and parameterize posterior beliefs by the weight
$\theta\in[0,1]$ on the optimistic model: $\pi_\theta(\m{q}^H)=\theta$ and
$\pi_\theta(\m{q}^L)=1-\theta$. For each risky arm $r_i$, define its risky-safe value difference
\begin{equation}\label{eq:Delta_i_def_main}
  \Delta_i(\lambda,\bar{\mu},\theta)
  :=
  V_{\lambda,\bar{\mu}\lambda}(r_i;\pi_\theta)
  -
  V_{\lambda,\bar{\mu}\lambda}(s;\pi_\theta).
\end{equation}

As in the two-arm case, the effect of complexity aversion on each risky-safe comparison 
is governed by the entropy gap between their worst-case distortions. When beliefs are parameterized by the posterior weight on an optimistic model, 
each risky arm admits a one-dimensional comparison against the safe arm. 
The slope of this comparison with respect to optimistic beliefs is constant, 
so each risky arm is characterized by a single belief threshold.

\begin{assumption}[Uniform entropy gap for each risky arm]\label{ass:uniform_entropy_gap_multi}
Let $[\underline{\lambda},\bar{\lambda}]$ be as in Assumption \ref{ass:U}. 
There exist constants $H_1^\ast,\dots,H_K^\ast>0$ such that for all
$(\lambda,\pi)\in[\underline{\lambda},\bar{\lambda}]\times\Delta(Q)$, all
$\bar{\mu} \geq 0$ with $\bar{\mu} \lambda^2 < 1$, and for each risky arm $r_i$, $H_s(\lambda,\bar{\mu},\pi)-H_{r_i}(\lambda,\bar{\mu},\pi) \geq  H_i^\ast.$
\end{assumption}

Assumption \ref{ass:uniform_entropy_gap_multi} requires that the safe arm’s 
worst-case distortion be uniformly more entropic than that of each risky arm 
over the relevant parameter range. This ensures that increasing $\bar{\mu}$ 
tilts every risky-safe comparison in favor of the safe arm.

\begin{assumption}[Monotone effect of the optimistic model]\label{ass:monotone_theta_multi}
For every $\lambda\in[\underline{\lambda},\bar{\lambda}]$, every $\bar{\mu} \geq 0$ 
with $\bar{\mu} \lambda^2 < 1$, and every risky arm $r_i$, 
shifting posterior weight toward the optimistic model strictly raises the 
risky-safe value difference.
\end{assumption}

Assumption \ref{ass:monotone_theta_multi} guarantees that learning in favor 
of the optimistic model raises the relative value of each risky arm, 
so each arm admits a well-defined belief threshold.
\begin{proposition}\label{prop:multi_thresholds}
Suppose Assumptions \ref{ass:full_support}, \ref{ass:U},
\ref{ass:uniform_entropy_gap_multi}, and \ref{ass:monotone_theta_multi} hold. 
Fix any $\lambda\in[\underline{\lambda},\bar{\lambda}]$ and any 
$\bar{\mu}$ satisfying $\bar{\mu}\lambda^2<1$. Assume further that for each risky arm $r_i$, $\Delta_i(\lambda,\bar{\mu},0)<0<\Delta_i(\lambda,\bar{\mu},1).$ Then, for each risky arm $r_i$:

\begin{enumerate}
  \item[(i)] There exists a unique belief threshold $\theta_i^\ast(\bar{\mu})\in(0,1)$ 
  such that $r_i$ is (weakly) optimal relative to the safe arm 
  if and only if the posterior weight on the optimistic model 
  exceeds $\theta_i^\ast(\bar{\mu})$.

  \item[(ii)]  If $\bar\mu_2>\bar\mu_1$ are such that
$\Delta_i(\lambda,\bar\mu_j,0)<0<\Delta_i(\lambda,\bar\mu_j,1),$ for $j=1,2,$
then the corresponding thresholds satisfy
$\theta_i^\ast(\bar\mu_2)>\theta_i^\ast(\bar\mu_1).$
\end{enumerate}
\end{proposition}

Proposition \ref{prop:multi_thresholds} shows that each risky arm admits a belief threshold 
$\theta_i^\ast(\bar{\mu})$, and stronger complexity aversion raises these thresholds. 
Thus, as $\bar{\mu}$ increases, experimentation with risky arms requires increasingly 
optimistic beliefs. Moreover, the magnitude of the entropy gap governs how rapidly thresholds shift. 
Risky arms whose worst-case distortions are especially concentrated 
(low entropy) experience larger downward shifts in relative value as $\bar{\mu}$ rises, 
and are therefore eliminated first. Complexity aversion acts as a cross-sectional 
selection device among risky alternatives. In the dynamic model, these shifting thresholds imply that risky arms disappear sequentially as $\bar{\mu}$ increases.
\subsubsection{Only risky arms}

We now consider the case in which all arms are risky, so there is no safe benchmark. 
In this environment, we show that complexity aversion acts as a cross-sectional selection device among risky alternatives without necessarily inducing conservatism.

For any two actions $i,j \in A$, define the pairwise value difference 
$$
  \Delta_{ij}(\lambda,\bar{\mu},\pi)
  :=
  V_{\lambda,\bar{\mu}\lambda}(i;\pi)
  -
  V_{\lambda,\bar{\mu}\lambda}(j;\pi).
$$
As in the safe-risky case, the effect of complexity aversion on this comparison 
is governed by the entropy gap between the corresponding worst-case distortions.

\begin{lemma}\label{lem:pairwise_Delta}
Suppose Assumption \ref{ass:full_support} holds. Then, for any 
$i,j\in A$ and any $(\lambda,\bar{\mu},\pi)$ with 
$\kappa=1/\lambda-\bar{\mu}\lambda>0$, $\frac{\partial}{\partial\bar{\mu}}
\Delta_{ij}(\lambda,\bar{\mu},\pi)
=
\lambda \big( H_i(\lambda,\bar{\mu},\pi) - H_j(\lambda,\bar{\mu},\pi) \big).$
\end{lemma}

Lemma \ref{lem:pairwise_Delta} implies that increasing $\bar{\mu}$ tilts comparisons 
in favor of actions whose worst-case distortions are more entropic, and against those 
whose pessimistic narratives are simpler (more concentrated).

\paragraph{Pairwise entropy dominance.}
To obtain a clear elimination result, we require that one risky arm has a uniform 
entropy advantage over another across the relevant parameter range.

\begin{assumption}[Pairwise entropy dominance]\label{ass:pairwise_entropy}
There exist distinct actions $i,j\in A$, a constant $H_{ij}^\ast>0$, 
and a baseline $\bar{\mu}_0 \geq 0$ such that:

\begin{enumerate}
  \item[(i)] For all $(\lambda,\bar{\mu},\pi)
  \in
  [\underline{\lambda},\bar{\lambda}]
  \times[\bar{\mu}_0,1/\bar{\lambda}^2)
  \times\Delta(Q),$ let $H_j(\lambda,\bar{\mu},\pi) - H_i(\lambda,\bar{\mu},\pi)
  \ge
  H_{ij}^\ast.$

  \item[(ii)] The initial value advantage of $i$ at $\bar{\mu}_0$ is not too large:
  \begin{align}\label{eq:Mij_def}
  K_{ij}
  :=
  \sup_{(\lambda,\pi)\in[\underline{\lambda},\bar{\lambda}]\times\Delta(Q)}
  \max\{\Delta_{ij}(\lambda,\bar{\mu}_0,\pi),0\}
  <
  \underline{\lambda} H_{ij}^\ast\big(1/\bar{\lambda}^2-\bar{\mu}_0\big).
  \end{align}
\end{enumerate}
\end{assumption}

Assumption \ref{ass:pairwise_entropy} is the multi-arm analogue of Assumptions \ref{ass:uniform_entropy_gap} and \ref{ass:initial_risky_advantage}. It requires that arm $j$’s worst-case distortions 
be uniformly more entropic than arm $i$’s over the relevant parameter range, and that 
this entropy advantage be quantitatively strong enough to overturn any baseline 
value advantage of $i$ as $\bar{\mu}$ increases.

\begin{proposition}[Elimination without a safe arm]
\label{prop:all_risky_elimination}
Suppose Assumptions \ref{ass:full_support} and 
\ref{ass:pairwise_entropy} hold. Then, there exists 
$\bar{\mu}^{\ast}\in(\bar{\mu}_0,1/\bar{\lambda}^2)$ such that for all
$\bar{\mu}\in[\bar{\mu}^{\ast},1/\bar{\lambda}^2)$,
$\lambda\in[\underline{\lambda},\bar{\lambda}]$, and $\pi\in\Delta(Q)$, $V_{\lambda,\bar{\mu}\lambda}(j;\pi)
>
V_{\lambda,\bar{\mu}\lambda}(i;\pi).$
In particular, for all such $\bar{\mu}$, action $i$ is never a best reply.
\end{proposition}

Proposition \ref{prop:all_risky_elimination} shows that sufficiently strong 
complexity aversion can eliminate a risky arm even when no safe alternative exists. 
In this case, complexity aversion does not simply induce conservatism, 
but instead selects among risky alternatives by favoring those whose 
worst-case narratives are more diffuse. 
It is useful to contrast this result with the safe–risky environment. 
With a safe benchmark, Assumption \ref{ass:uniform_entropy_gap_multi} 
ensures that the safe arm enjoys a uniform entropy advantage over all risky arms, 
so increasing $\bar{\mu}$ systematically pushes the agent toward safety. 
In the present setting, where all arms are risky, there is no canonical benchmark. 
Complexity aversion does not necessarily make the agent more conservative in any absolute sense; 
rather, it reshapes the relative ranking of risky actions.

Absent a global entropy ordering as in Assumption \ref{ass:pairwise_entropy}, 
Lemma \ref{lem:pairwise_Delta} still implies that increasing $\bar{\mu}$ locally tilts 
pairwise comparisons toward arms with more entropic worst-case distortions. 
However, no arm need become uniformly dominated across the relevant parameter range. 
In such environments, the dynamic evolution of $\lambda_t$ can continue to sustain 
switching among risky arms even at high levels of complexity aversion.

\subsection{Complexity aversion and long-run payoffs}\label{sec:wel}

We now isolate a class of environments in which increasing complexity aversion is welfare-improving: it raises the agent's long-run realized payoff under the true DGP, even though it makes the agent more conservative over time. The key here is that the true DGP actually favors the safe arm, so eliminating safe–risky cycles prevents the agent from over-experimenting with the risky action. When the safe action is ex-post optimal but the $\bar{\mu}=0$ benchmark exhibits endogenous safe–risky cycles, sufficiently strong complexity aversion eliminates experimentation with the risky arm and raises long-run realized payoffs. In such environments, simplicity stabilizes behavior and prevents overreaction to temporarily favorable but weakly supported outcomes. Online Appendix \ref{sec:welfare} provides the details and a formal welfare comparison, and it  also shows settings where complexity aversion alone can create welfare losses.

\section{Applications}\label{sec:app}
Sections \ref{sec:cri_ri} and \ref{sec:home_bias} explore applications in discrete-choice analysis and stochastic growth, respectively. Online Appendices \ref{sec:lanzani25_taxation} and \ref{sec:chamb20_entropy_gap} apply our framework, respectively, to taxation and portfolio-choice environments.
\subsection{Rational Inattention}\label{sec:cri_ri}

This application shows how our complexity-augmented criterion in
Section \ref{sec:static} delivers new, empirically relevant predictions when embedded
in \citeauthor{matvejka15}'s (\citeyear{matvejka15}) rational inattention (RI) framework for discrete-choice analysis. The key is that complexity aversion induces
\emph{state-dependent choice sensitivity} 
even when the underlying information-cost technology is the  Shannon entropy. 

\subsubsection{Setup and definitions}

We follow  the framework in \citet{matvejka15}. Let $A=\{1,\dots,n\}$ be a finite set of actions and $\Omega$ be a finite set of payoff states.
In state $\omega\in\Omega$, action $a\in A$ yields payoff $v(a,\omega)\in\R$.
The agent's \emph{structured} or \emph{reference} prior over $\Omega$ is $g\in\Delta(\Omega)$.

An RI strategy is a stochastic choice rule $\psi \in \Delta(A)^\Omega$ satisfying
$$ \psi(a\mid\omega)\in[0,1],\quad\text{and}\quad \sum_{a\in A}\psi(a\mid\omega)=1,
$$
for every state $\omega\in\Omega$. Let the (reference) unconditional choice probabilities be
$$\bar{\psi}(a):=\sum_{\omega\in\Omega} g(\omega)\psi(a\mid\omega)\in\Delta(A),$$
and for a cost parameter $\xi>0$  define 
\begin{equation}\label{eq:cri_ri_cost}
  \mathcal{C}(\psi;g)
  :=\xi \sum_{\omega\in\Omega} g(\omega)\sum_{a\in A}\psi(a\mid\omega)
  \log\frac{\psi(a\mid\omega)}{\bar{\psi}(a)}.
\end{equation}
$\mathcal{C}$ is $\xi$ times the mutual information between $\omega$ and $a$ under the reference prior $g$, which coincides with \citet[][eq. (5)]{matvejka15}. We now introduce robustness to misspecification of the prior over payoff states with complexity aversion over worst-case scenarios.
Given a strategy $\psi$, let the induced state-contingent payoff be
$U_\psi(\omega) := \sum_{a\in A} v(a,\omega)\psi(a\mid\omega).$
Then, Nature chooses a distorted distribution $m\in\Delta(\Omega)$ to make payoffs low, paying
a KL penalty relative to $g$ and an entropic penalty  as in \eqref{eq:static_criterion}.
Our agent anticipates this distortion and therefore chooses $\psi$ according to our complexity-augmented criterion:
\begin{equation}\label{eq:cri_ri_problem}
  \max_{\psi\in\Delta(A)^\Omega}\;
  \min_{m\in\Delta(\Omega)}
  \Bigg\{
    \sum_{\omega\in\Omega} m(\omega)\hspace{0.03in}U_\psi(\omega)
    +\frac{1}{\lambda}R(m\Vert g) + \mu H(m)
    \;-\; \mathcal{C}(\psi;g)
  \Bigg\},
\end{equation}
where $\lambda>0$ is the misspecification concern and $\mu\geq0$ is the complexity parameter. When the agent has no concern for misspecification and no complexity aversion, \eqref{eq:cri_ri_problem} coincides with \citet[][eq. (10)]{matvejka15}. In the standard RI framework, the prior over payoff states is fixed at $g$ and
the optimal $\psi$ has a generalized logit form with a \emph{single} global scale parameter \citep[][eq. (15)]{matvejka15}.
In \eqref{eq:cri_ri_problem}, the state distribution entering payoffs is endogenously distorted by robustness
and disciplined by complexity. This interaction then implies that the \emph{effective} choice sensitivity becomes
state-dependent even though the information-cost technology \eqref{eq:cri_ri_cost} is unchanged.

\subsubsection{Worst-case payoff states}

Fix a stochastic rule $\psi$. The inner problem in \eqref{eq:cri_ri_problem} is the static criterion
\eqref{eq:static_criterion} applied to the outcome space $\Omega$ with reference $g$ and utility $U_\psi$.
Thus, the closed-form characterization in Lemma \ref{lem:closed_form} applies directly.

\begin{obs}\label{prop:cri_ri_mstar}
Suppose the prior $g$ has full support on $\Omega$ and let $\kappa:=\frac{1}{\lambda}-\mu$ and
$\beta:=\frac{1}{1-\lambda\mu}$ as in Lemma \ref{lem:closed_form}. If $\kappa>0$,
then for every $\psi$ the inner minimization over $m$ in \eqref{eq:cri_ri_problem} has a unique minimizer
$m^\star_{\lambda,\mu}(\cdot;\psi)\in\Delta(\Omega)$ given by
\begin{equation}\label{eq:cri_ri_mstar}
  m^\star_{\lambda,\mu}(\omega;\psi)
  = \frac{\exp\{-U_\psi(\omega)/\kappa\}\hspace{0.03in}g(\omega)^{\beta}}
         {\sum_{\omega'\in\Omega}\exp\{-U_\psi(\omega')/\kappa\}\hspace{0.03in}g(\omega')^{\beta}}
  \qquad (\omega\in\Omega).
\end{equation}
For fixed $\lambda$, increasing $\mu$ reduces the entropy of the worst-case distribution (Proposition \ref{prop:entropy_decreasing_mu}).
\end{obs}

Observation \ref{prop:cri_ri_mstar} formalizes the central ``narrative'' effect:
robustness shifts weight toward low $U_\psi(\omega)$ states, while complexity aversion amplifies concentration by
penalizing diffused distortions.  Note that $m^\star_{\lambda,\mu}(\cdot;\psi)$ depends on $\psi$ through $U_\psi$, so the worst-case distribution and the induced effective scale are determined jointly with the agent’s choice rule.

\subsubsection{Endogenous state-dependent logit scale}

We now state the main new implication for observables in discrete-choice applications:
the optimal stochastic choice rule inherits \emph{state-dependent} sensitivity.

Define the state-specific \emph{effective scale} induced by $(g,m^\star)$:
\begin{equation}\label{eq:cri_ri_scale}
  \xi_\omega(\psi)
  := \xi\hspace{0.03in}\frac{g(\omega)}{m^\star_{\lambda,\mu}(\omega;\psi)} \in (0,\infty),
  \qquad (\omega\in\Omega).
\end{equation}
The effective scale $\xi_\omega(\psi)$ captures how pessimistic reweighting of state
$\omega$ alters choice sensitivity. When the worst-case distribution overweights a state
relative to the reference prior ($m^\star_{\lambda,\mu}(\omega;\psi)>g(\omega)$), the
effective scale satisfies $\xi_\omega(\psi)<\xi$, so the agent behaves as if information
were cheaper and choices were more deterministic. Conversely, states that
are underweighted by the worst-case narrative exhibit higher effective noise.

\begin{proposition}\label{prop:cri_ri_logit}
Assume $\mu<1/\lambda$ so that Observation \ref{prop:cri_ri_mstar} applies, and consider any solution
$\psi^\star$ of \eqref{eq:cri_ri_problem} with $\bar{\psi}^\star(a)>0$ for all $a\in A$. 
Then, for each $\omega\in\Omega$ and $a\in A$,
\begin{equation}\label{eq:cri_ri_logit}
  \psi^\star(a\mid\omega)
  =
  \frac{\bar{\psi}^\star(a)\exp\big\{\frac{v(a,\omega)}{\xi_\omega(\psi^\star)}\big\}}
       {\sum_{b\in A}\bar{\psi}^\star(b)\exp\big\{\frac{v(b,\omega)}{\xi_\omega(\psi^\star)}\big\}},
\end{equation}
where $\xi_\omega(\psi^\star)$ is defined by \eqref{eq:cri_ri_scale}.
\end{proposition}
The proof of this and the other propositions in this section are in Online  Appendix \ref{ola:omittedproofs}. 
\begin{corollary}\label{cor:cri_ri_logodds}
Under the conditions of Proposition \ref{prop:cri_ri_logit}, for any $a,b\in \text{\normalfont supp}\hspace{0.02in}\bar{\psi}^\star$ and $\omega\in\Omega$,
\begin{equation*}
  \log\frac{\psi^\star(a\mid\omega)}{\psi^\star(b\mid\omega)}
  =
  \log\frac{\bar{\psi}^\star(a)}{\bar{\psi}^\star(b)}
  + \xi_\omega(\psi^\star)^{-1}\big(v(a,\omega)-v(b,\omega)\big).
\end{equation*}
Hence, the slope of log choice odds with respect to payoff differences equals 
$\xi_\omega(\psi^\star)^{-1} = \xi^{-1}\hspace{0.03in}m^\star_{\lambda,\mu}(\omega;\psi^\star)/g(\omega)$.
\end{corollary}

Corollary \ref{cor:cri_ri_logodds} provides a key empirical prediction of our framework:
choice sensitivity endogenously rises in states that the worst-case narrative overweights.
This result therefore provides a new foundation for the pervasive practice of allowing scale heterogeneity
in logit and mixed-logit estimation \citep[e.g.,][]{swait93,fiebig10}.

\subsubsection{Probability neglect}\label{sec:prob}

We next relate our framework to a phenomenon called ``probability neglect'': when emotions or salience push agents toward worst-case thinking,
responses often appear to depend on the severity of a bad outcome more than on its probability \citep[see,][]{sunstein02,sunstein03}.
In our framework, probability neglect arises when complexity aversion drives the worst-case belief $m^\star$ toward a concentrated
single scenario, which makes the scale $\xi_\omega$ small precisely there.

\begin{corollary}\label{cor:cri_ri_probneglect}
Fix $\lambda>0$ and $\xi>0$. Consider any $\psi$ and the associated $m^\star_{\lambda,\mu}(\cdot;\psi)$
from \eqref{eq:cri_ri_mstar}. As $\mu\uparrow 1/\lambda$,
$m^\star_{\lambda,\mu}(\cdot;\psi)$ concentrates on the set
$\arg\min_{\omega\in\Omega}\big\{U_\psi(\omega)-\frac{1}{\lambda}\log g(\omega)\big\}.$
If the minimizer is unique, say $\hat{\omega}$, then
$m^\star_{\lambda,\mu}(\hat{\omega};\psi)\to 1$ and $m^\star_{\lambda,\mu}(\omega;\psi)\to 0$
for all $\omega\neq\hat{\omega}$, and 
$$
  \xi_{\hat{\omega}}(\psi) = \xi\hspace{0.03in}\frac{g(\hat{\omega})}{m^\star_{\lambda,\mu}(\hat{\omega};\psi)} \to \xi\hspace{0.03in}g(\hat{\omega}),
  \qquad
  \xi_{\omega}(\psi)\to\infty\ \text{for }\omega\neq\hat{\omega}.
$$
Thus, the induced logit rule \eqref{eq:cri_ri_logit} becomes comparatively ``sharp'' (high payoff sensitivity) in the
selected worst-case state $\hat{\omega}$ and comparatively ``flat'' elsewhere.
\end{corollary}

Notice, for example, that when the prior $g$ is uniform, the minimizer in Corollary \ref{cor:cri_ri_probneglect} satisfies $\hat{\omega}=\arg\min_{\omega\in\Omega}\{U_\psi(\omega)\}$, i.e., $\hat{\omega}$ is the worst payoff state. Corollary \ref{cor:cri_ri_probneglect} therefore provides a new mechanism for probability neglect without appealing to bounded rationality:
a sophisticated agent overly concerned about complexity behaves as if a single adverse scenario dominates evaluation, and allocates decision precision accordingly.
The standard  RI framework assumes an exogenous scale parameter, so it cannot produce the endogenous state-dependent behavior needed to explain probability neglect.

\subsubsection{Discussion}

Empirical discrete-choice analysis typically allows the logit scale to vary across contexts or individuals
(i.e., scale heterogeneity) and develops flexible specifications to accommodate it \citep[e.g.,][]{swait93,fiebig10}.
Corollary \ref{cor:cri_ri_logodds} shows that scale heterogeneity can arise \emph{endogenously} even with a fixed
information-cost technology: it is driven by the likelihood ratio $m^\star/g$ generated by the robust-complexity
criterion.   Observable factors that increase pessimistic overweighting of particular states---such as adverse news, forecast failures, regime changes, or time pressure---raise
the likelihood ratio $m^\star_{\lambda,\mu}(\omega)/g(\omega)$ and therefore reduce the
effective scale $\xi_\omega$. The model predicts that such factors generate
state-specific increases in choice sensitivity, rather than a uniform reduction
in noise. This provides a structural interpretation of scale heterogeneity in
discrete-choice data based on robustness and complexity aversion, rather than
 preference shocks.

In summary, this application highlights what robustness and complexity aversion add to the RI literature:
it transforms a single global ``noise'' parameter into a structured, state-dependent object pinned down by primitives
(prior $g$, payoffs, and the $(\lambda,\mu)$ criterion). This delivers a tractable bridge between robustness,
Occam-style simplicity of narratives, and empirical evidence of context-dependent choice sensitivity.

\subsection{Stochastic growth}\label{sec:home_bias}
We introduce our criterion in the stochastic-growth analysis of \citet{growth22}.
\subsubsection{Home bias}
 A key empirical regularity in international finance is  \emph{equity home bias}: domestic assets are overweighted relative to diversification benchmarks. In our framework, home bias can arise even when domestic and foreign assets have the \emph{same} unconditional payoff distribution under the reference model. The mechanism is that, when $\mu>0$, the criterion \eqref{eq:static_criterion} penalizes diffuse worst-case distortions; this makes assets whose downside risk is concentrated in a small set of focal crisis scenarios (low-entropy narratives) endogenously more fragile. This channel is absent in entropy-free benchmarks (i.e., $\mu=0$), including the stochastic-growth formulation in \citet{growth22}, which depends only on a KL penalty and therefore cannot distinguish assets that share the same payoff distribution under the reference model.

\par There are two assets  $A=\{d,f\}$ and outcomes $Y=\{y^*,y_1,\dots,y_N\}$ with $N\ge2$. Utilities are binary and identical across assets: $u(d,y^*)=u(f,y^*)=1,$ and $
u(d,y)=u(f,y)=0\ \text{for all }y\in\{y_1,\dots,y_N\}.$ Hence, outcome $y^*$ is a favorable scenario, whereas the remaining $N$ outcomes are unfavorable scenarios that capture various economic downturns.
Let $Q=\{\m{q}\}$ be a structured model, with $q_d,q_f\in\Delta(Y)$.
Fix $\delta\in(0,1)$ and $\varepsilon\in(0,\delta/N)$ and define
$$
q_d(y^*)=q_f(y^*)=1-\delta,
\qquad
q_d(y)=\frac{\delta}{N}\quad \text{for }y=y_1,\dots,y_N,
$$
$$
q_f(y_1)=\delta-(N-1)\varepsilon,
\qquad
q_f(y)=\varepsilon\quad \text{for }y=y_2,\dots,y_N.
$$
Thus, under $\m{q}$ both assets deliver payoff $1$ with probability $1-\delta$ and payoff $0$ with probability $\delta$, but the foreign downside mass is concentrated on a single crisis state. Lemma \ref{lem:closed_form} implies the closed-form value for each asset $a\in\{d,f\}$ is $v_{\lambda,\mu}(a;\m{q}) = -\kappa\log Z_a,$ where $ 
Z_a := e^{-1/\kappa}q_a(y^*)^\beta+\sum_{i=1}^N q_a(y_i)^\beta,$
so the \emph{home-bias premium} is
$v_{\lambda,\mu}(d;\m{q})-v_{\lambda,\mu}(f;\m{q})=\kappa\log\frac{Z_f}{Z_d}.$

\begin{proposition}\label{prop:home_bias}
For every $\lambda>0$,
$$
v_{\lambda,0}(d;\m{q})=v_{\lambda,0}(f;\m{q})
\qquad\text{and}\qquad
v_{\lambda,\mu}(d;\m{q})>v_{\lambda,\mu}(f;\m{q}) \quad\text{ for all }\mu\in(0,1/\lambda).
$$
\end{proposition}

In \citet{growth22},  only the likelihood cost of a distorted model  and the associated payoff tilt matter. As a result, relabeling or reshuffling which particular states deliver a loss has no effect, because the benchmark criterion depends on \emph{how much} probability mass is shifted, not on \emph{how concentrated} the implied downside scenarios are. In contrast,  with our  complexity-augmented criterion a foreign asset whose losses are concentrated in a small number of focal states is strictly less attractive when $\mu>0$  holding fixed the unconditional payoff distribution, because the worst-case distortion can load on a simpler narrative. This delivers a parsimonious preference-based rationale for home bias that does not rely on expected-return differences, hedging motives, or segmentation: investors avoid foreign exposure precisely when the perceived foreign downside is concentrated around a small number of prominent crisis scenarios. Thus, Proposition \ref{prop:home_bias} links home bias to the \emph{structure} of tail risk, which generates sharp comparative statics that the stochastic-growth benchmark cannot deliver. Note also that since 
$v_{\lambda,\mu}(d;\m{q})\ge v_{\lambda,\mu}(f;\m{q}) \iff Z_d\le Z_f$ in our framework,
 home bias disappears (or reverses) whenever domestic tail risk is at least as concentrated as foreign tail risk in the power-sum sense $\sum_{i\ge1}q_a(y_i)^\beta$ induced by $\beta>1$.

\par The magnitude of home bias is potentially large even holding the probability of downside scenarios $\delta$ fixed: along the extreme-concentration limit $\varepsilon\downarrow 0$, $\sum_{i=1}^N q_f(y_i)^\beta \longrightarrow \delta^\beta$ and $\sum_{i=1}^N q_d(y_i)^\beta = \delta^\beta N^{1-\beta},$ so
$$
  v_{\lambda,\mu}(d;\m{q})-v_{\lambda,\mu}(f;\m{q})
  \longrightarrow
  \kappa\log \left(
    \frac{e^{-1/\kappa}(1-\delta)^\beta+\delta^\beta}
         {e^{-1/\kappa}(1-\delta)^\beta+\delta^\beta N^{1-\beta}}
  \right).
$$
This limit is strictly positive for every $\mu>0$, and for fixed $(\lambda,\mu,\delta)$, it is increasing in the number of distinct downside scenarios $N$. Thus, what matters in our analysis is not ``rarity'' per se but the low-entropy downside support: for $\mu>0$, any relabelling that concentrates probability within the loss set $\{y_1,\dots,y_N\}$ increases the foreign term $\sum_{i\ge1}q_f(y_i)^\beta$ and therefore decreases $v_{\lambda,\mu}(f;\m{q})$.
Rare crises amplify the mechanism when robustness is tight or payoffs are very low in those crisis states, because the Gibbs tilt in Lemma \ref{lem:closed_form} loads the worst-case distortion on low-utility states even if they are ex ante unlikely. Also note that whether induced home bias raises or lowers realized returns depends on the true joint distribution of domestic/foreign payoffs and the associated risk premia: home bias improves outcomes precisely when the concentrated foreign tail corresponds to genuine downside risk under the true DGP, and is costly when it induces under-diversification without compensating tail-risk reduction.

\subsubsection{Endogenous growth rates}

In \citet[][Proposition 4]{growth22}, the growth-rate loss from optimizing under misspecification equals the degree of misspecification, measured by the KL divergence between the structured model and true DGP, so long-run growth does not depend on the payoff environment  for a given degree of misspecification. We show that with  complexity aversion, 
misspecification losses become \emph{payoff-dependent} because the relevant  choice rule is now disciplined by
simplicity and therefore interacts with how returns load on states.

\paragraph{Complexity-augmented criterion.}
Fix $\mu\in[0,1)$. Following \citet{growth22}, let $A$ be a finite set of actions and let $Q=\{p'\}$ be a singleton structured model $p'\in\Delta(Y)$. For each mixed strategy  $\alpha\in\Delta(A)$ and state ${y}\in Y$, 
\begin{equation}\label{eq:complex_growth_aggregator}
G_\mu(\alpha,{y})=
-\min_{q\in\Delta(A)}
\Big\{
\E_q[-u(a,{y})]
+
R(q\Vert \alpha)
+
\mu H(q)
\Big\}
\end{equation}
is the analogue of our complexity-augmented criterion in stochastic growth, where $\lambda=1$, $-u$, and the negative sign in front are imposed to match \citet[][eq. (7)]{growth22} when $\mu=0$. Given a model $p\in\Delta({Y})$, the objective becomes $V_\mu(p):=\underset{\alpha\in\Delta(A)}{\max}\E_p\big[G_\mu(\alpha,{y})\big],$ where $\alpha_\mu^*(p)\in\underset{\alpha\in\Delta(A)}{\arg\max}\hspace{0.03in}\E_p\big[G_\mu(\alpha,{y})\big].$
Given models $p,p'\in\Delta({Y})$, the function
\begin{equation}\label{eq:complex_growth_loss}
L_\mu(p,p'):=\E_p\left[G_\mu(\alpha_\mu^*(p),{y})\right]-\E_p\left[G_\mu(\alpha_\mu^*(p'),{y})\right]
\end{equation} is the reduction of the growth rate caused by the optimization for the misspecified model $p'$ when the true model is $p$.
When $\mu=0$, $L_0(p,p')$ is the loss in \citet[][p. 371]{growth22}. 
In \citet[][Proposition 4]{growth22}, two environments with the same $R(p\Vert p')$ exhibit the same growth loss function under their Regularity Condition 2.\footnote{Regularity Condition 2 requires that the degree of misspecification is not too large (Assumption \ref{ass:rc2}).  }
\begin{proposition}\label{prop:growth_misspec_fails}
Under \eqref{eq:complex_growth_aggregator}, model-misspecification costs are not universal: even when the misspecified model satisfies \citet[][Regularity Condition 2]{growth22} and yields $L_0(p,p')=R(p\Vert p')$, the loss $L_\mu(p,p')$ in \eqref{eq:complex_growth_loss} need not equal $R(p\Vert p')$ for $\mu>0$.
\end{proposition}

In our framework, the same $R(p\Vert p')$ can generate very different
growth losses even when \citet[][Regularity Condition 2]{growth22} is satisfied, because complexity aversion penalizes diffuse sampled-choice rules, so the cost of a given belief error
depends on how much randomization is required to implement growth-optimal behavior in that payoff environment.
This captures a key empirical pattern in macro-finance and growth: belief disagreements and ``narrative shocks''
tend to have large effects when mixed strategies load on a small set of salient states (crisis
scenarios), while similarly sized misspecifications in calm environments have negligible effects on growth.
\citeauthor{growth22}'s (\citeyear[][Proposition 4]{growth22}) payoff-invariant loss rules out this heterogeneity, whereas our framework is able to capture it in a tractable way.\footnote{\citet{eden26} shows that in long-horizon stochastic-growth problems, a standard expected-utility maximizer is generically driven by atypical, rare-shock frequencies. We do not study how expected utility  amplifies rare events, but how adding complexity aversion changes which downside scenarios matter.}

\section{Representation}\label{sec:axiomatization}

This section provides a representation result for our criterion by  restating it in an 
Anscombe–-Aumann framework as in  \citet{lanzani25w}. The outcome space $Y$ introduced earlier now plays 
the role of a state space, an act is a mapping $f:Y\to X$ from states to 
consequences in a convex consequence space $X$, and $Q\subset \Delta(Y)$ is the set of structured models. This reformulation does not 
change the economics: in the static model, each action $a$ induces a payoff 
profile $y\mapsto u(a,y)$, which can be identified with an act over $Y$.

\par Fix an affine, nonconstant utility index $u:X\to\R$. For each $(f,q)\in X^{Y}\times Q$, define the analogues of (\ref{eq:static_criterion}) and (\ref{eq:posterior_value}) as
\begin{equation}\label{eq:axiom_v}
  v_{\lambda,\mu}(f;q)
  :=
  \min_{p\in\Delta(Y)}
  \left\{
    \E_p[u(f)]
    +\frac{1}{\lambda}R(p\Vert q)
    +\mu H(p)
  \right\},
\end{equation}
and
\begin{equation}\label{eq:axiom_V}
  V_{\lambda,\mu}(f;\pi)
  :=
  \sum_{q\in Q}\pi(q)v_{\lambda,\mu}(f;q).
\end{equation}
The induced preference, denoted $\succeq$, satisfies $f\succeq g$ if and only if  $V_{\lambda,\mu}(f;\pi)\ge V_{\lambda,\mu}(g;\pi)$. When $\mu=0$ (i.e., no complexity penalty), (\ref{eq:axiom_v}) becomes \citeauthor{hansen01}'s (\citeyear{hansen01}) standard multiplier criterion and (\ref{eq:axiom_V}) becomes the  ARC criterion in \citet[][eq. (6)]{lanzani25w}. This yields the following hierarchy of representations: the  multiplier criterion is a special case of the standard ARC criterion when $Q$ is a singleton, and the standard ARC criterion is a special case of our complexity-augmented criterion when $\mu=0$.

\subsection{ARC representation}

Given  $\beta=1/(1-\lambda\mu)\in\R$ and a model $q\in\Delta(Y)$, let $Z(q):=\sum_{y\in Y}q(y)^{\beta}$ be a normalizing constant  and define the \textit{power-transform} of $q$ as
\begin{equation}\label{eq:axiom_Zq}
\tilde q(\cdot):=\frac{q(\cdot)^{\beta}}{Z(q)}\in\Delta(Y).
\end{equation}
Let $\widetilde Q:=\{\tilde q: q\in Q\}\subset \Delta(Y)$ and let $\widetilde\pi\in\Delta(\widetilde Q)$ be the
pushforward of $\pi$ under $q\mapsto \tilde q$.\footnote{The power-transform distribution $\tilde q\propto q^{\beta}$ in \eqref{eq:axiom_Zq} is called an ``escort distribution'' in statistical physics  and information geometry \citep[e.g.,][]{esc93}.}

\begin{theorem}\label{prop:arc_representation}
Let $\mu\in\R$, finite $Q\subset\Delta(Y)$, and $\pi\in\Delta(Q)$. The following hold:
\begin{enumerate}
   \item[1.] Suppose $\lambda>0$, $\mu<1/\lambda$, and define $\tilde q$ by \eqref{eq:axiom_Zq}. Then, $\kappa=\frac{1}{\lambda}-\mu>0$ and there exists a constant
$\widetilde{C}(\lambda,\mu,\pi)\in\R$, independent of acts, such that for all acts $f$,
\begin{equation}\label{eq:arc_form}
  V_{\lambda,\mu}(f;\pi)
  =
  \widetilde{C}(\lambda,\mu,\pi)
  +\sum_{q\in Q}\pi(q)\,
    \min_{p\in\Delta(Y)}
    \Big\{
      \E_p[u(f)]
      +\kappa\,R(p\Vert\tilde q)
    \Big\}.
\end{equation}
Equivalently, $\succeq$ is ordinally equivalent to an ARC preference with parameter
$\lambda_{\mathrm{ARC}}:=1/\kappa$ and structured models $\widetilde Q$ endowed with belief $\widetilde\pi$.
\item[2.] Suppose $\lambda>0$ and $\mu\ge 1/\lambda$. Then, the preference induced by $V_{\lambda,\mu}(\cdot;\pi)$ is a variational preference but not an ARC preference for $|Y|\ge 2$. In particular, for all acts $f$, 
\begin{align*}
 V_{\lambda,\mu}(f;\pi)
=\sum_{q\in Q}\pi(q)\min_{y\in\text{\normalfont supp}(q)}\Big\{u(f(y)) +\frac{1}{\lambda}\log\frac{1}{q(y)}\Big\}.
\end{align*}
 which is independent of $\mu$ on the entire region $\mu\geq1/\lambda$.
\item[3.] Suppose $\lambda\rightarrow0$ and $\mu\in\R$ is fixed. Then, 
$\underset{\lambda\rightarrow0}{\text{\normalfont lim}}\hspace{0.02in}V_{\lambda,\mu}(\cdot;\pi)=\sum_{q\in Q}\pi(q)\big(\E_{q}[u(\cdot)]+\mu H(q)\big)$ induces a subjective expected utility preference. When $Q=\{q\}$, $\underset{\lambda\rightarrow0}{\text{\normalfont lim}}\hspace{0.02in}v_{\lambda,\mu}(\cdot;q)$ induces the subjective version of \citeauthor{mononen25}'s (\citeyear{mononen25}) entropy-modified expected utility preference.
\end{enumerate}
\end{theorem}

Theorem \ref{prop:arc_representation} classifies the induced preference into three regimes. In the low-simplicity region ($\mu< 1/\lambda$), the preference  is equivalent to an ARC criterion computed against a transformed within-model benchmark \eqref{eq:axiom_Zq}, so simplicity affects behavior only through a change in the effective reference distributions.

In the high-simplicity region ($\mu\ge 1/\lambda$), the criterion becomes extreme-state within each structured model and no longer induces an ARC preference; moreover, beyond the threshold, choices are unaffected by higher complexity aversion. Finally, in the vanishing-misspecification limit, the induced preference reduces to subjective expected utility preference with an act-independent additive entropy term. 

\subsection{Connection to \citet{gabaix25}}\label{subsec:gabaix}

Theorem \ref{prop:arc_representation} shows that in the high-simplicity regime ($\mu\ge 1/\lambda$), worst-case beliefs collapse to Dirac measures, so Nature constructs pessimism using a single salient adverse scenario within each structured model. For this reason, our high-simplicity model can generate the same reduced-form objective function, and hence the same choice behavior, as \citeauthor{gabaix25}'s (\citeyear{gabaix25}) first-order complexity aversion in the one-dimensional applications studied in Sections 3 and 4 of that paper. The equivalence is at the level of the objective function and its optimizer: our framework does not reproduce Gabaix's underlying signal structure, decomposition into primitive parameters, or comparative statics with respect to cognition costs and other terms that our model does not include.

The link is clearest in applications where the objective takes the form of a smooth term minus a first-order penalty proportional to the absolute size of the relevant choice variable. In those problems, our high-simplicity regime generates the same objective function, up to an additive constant. To see this, consider a generic scalar choice $x\in\mathbb R$ and suppose that Gabaix's objective can be written as
\begin{equation}\label{eq:gabaix_generic_obj}
V^{G}(x)=\widetilde V(x)-\Gamma |x|,
\end{equation}
where $\widetilde V$ is the smooth part of the problem and $\Gamma>0$ is the coefficient on first-order complexity aversion. In Gabaix's general formulation, this coefficient is a composite object that depends on several underlying terms, including the agent's complexity aversion, the curvature of utility, and the ``natural scale'' of the action, with the exact expression differing across applications.

To embed \eqref{eq:gabaix_generic_obj} into our framework, fix a binary narrative space $Y=\{0,1\}$, let $Q=\{q\}$ with
$
q(0)=q(1)=\frac12,$
and define an auxiliary act $f_x:Y\to\mathbb R$ by
$
f_x(0)=0$ and $f_x(1)=-\gamma |x|,$
where $\gamma>0$ is a coefficient to be matched to Gabaix's primitives. In the high-simplicity regime, our criterion becomes
$$
v_{\lambda,\mu}(f_x;q)
=
\min_{y\in\{0,1\}}
\Big\{f_x(y)+\frac{1}{\lambda}\log\frac{1}{q(y)}\Big\}
=
\frac{1}{\lambda}\log 2-\gamma |x|.
$$
Hence, the augmented objective
$$
V^{FM}(x):=\widetilde V(x)+v_{\lambda,\mu}(f_x;q)
=\frac{1}{\lambda}\log 2+\widetilde V(x)-\gamma |x|
$$
coincides with \eqref{eq:gabaix_generic_obj} up to the additive constant $\lambda^{-1}\log 2$ once we set $\gamma=\Gamma$. Therefore, whenever a Gabaix application reduces to a smooth objective plus a single first-order penalty $\Gamma |x|$, our model reproduces the same objective function in the high-simplicity regime.

In Gabaix's indexation application, the choice variable is the indexation coefficient $b$. The smooth part of the objective is
$
\widetilde V_I(b):=\bar v +(2\bar b b-b^2) -(1-m)b^2,
$
so the full objective is
\begin{equation*}
V_I^{G}(b)=\widetilde V_I(b)-2\xi\sigma_a(1-m)|b|.
\end{equation*}
Thus, the coefficient on the kinked term is $\Gamma_I=2\xi\sigma_a(1-m)$. Applying the generic construction above with $x=b$ and
$
\gamma=\Gamma_I=2\xi\sigma_a(1-m)
$
gives
$$
V_I^{FM}(b)=\frac{1}{\lambda}\log 2+\widetilde V_I(b)-2\xi\sigma_a(1-m)|b|.
$$
Thus, the optimizer satisfies the same threshold rule as in Gabaix. However, choice data from this problem identify only the composite coefficient on $|b|$. 
\par In summary, when $\mu \geq 1/\lambda$, Nature's worst-case distortion becomes a Dirac mass on a single adverse state. 
Economically, the agent evaluates actions \emph{as if} pessimism is driven by one salient 
scenario, not by a carefully balanced collection of bad outcomes. This generates literal 
probability neglect due to the symmetry of $q$: the criterion depends only on the worst payoff. For $\mu < 1/\lambda$, 
 the worst-case distortion $\hat{p}_{\lambda,\mu}(f_b;q)(1)$ in \eqref{eq:p_star_closed_form} remains interior: $\frac{d}{d|b|} v_{\lambda,\mu}(f_b;q) = -\gamma  \hat{p}_{\lambda,\mu}(1;b)$, 
where $ \hat p_{\lambda,\mu}(f_b;q)(1)$ increases continuously in $|b|$. This expression predicts smooth \textit{partial} simplicity 
(many small but nonzero choices) rather than sharp bunching at exactly zero. Empirical evidence favors partial dampening as the baseline, with full dampening emerging only under more extreme conditions.\footnote{ In insurance choice, \citet{barseghyan2013nature} finds probability distortions but only mild insensitivity to changes in probabilities, indicating attenuated rather than shut-down responses. In tax salience, \citet{chetty09} shows that consumers underreact to taxes that are not salient, again implying partial but nonzero adjustment. More broadly, \citet{handel18}  emphasizes that interventions often move choices only part-way toward frictionless benchmarks.}

As we show in Online Appendix \ref{app:gabaix}, the same reduced-form logic applies to Gabaix's pricing application. The taxation application in Gabaix's Section 5 requires a further extension, because the relevant reduced-form object is not a single absolute-value penalty but a sum of such penalties across tax features. A single binary narrative does not reproduce that objective exactly, but a collection of feature-specific binary narratives does.

\subsection{Discussion: interpretation of $\mu$}\label{sec:role_mu}
Theorem \ref{prop:arc_representation} shows that the role of $\mu$ depends critically on its magnitude relative to $\lambda$. When $\mu<1/\lambda$, the preference is ordinally equivalent to an ARC preference computed under the transformed reference models $\widetilde Q$ in \eqref{eq:axiom_Zq}. Since ARC does not identify $Q$ uniquely, $\mu$ should be understood as selecting a particular effective representation rather than as transforming  a uniquely defined set of models.

When $\mu\ge 1/\lambda$, the entropic penalty forces the criterion to collapse to the extreme-state form in Theorem \ref{prop:arc_representation}.2. In this region, the induced preference is variational but not ARC, and the criterion is independent of $\mu$. Thus, only the region $\mu<1/\lambda$ delivers a smooth comparative statics of complexity aversion within the ARC family, whereas beyond the threshold, the ARC structure collapses and additional preference for simplicity has no behavioral bite. 
 Notice also that this case  captures probability neglect as illustrated in Section \ref{sec:prob}.

This classification also sharpens our dynamic interpretation. In the ARC case, $\mu>0$  selects  representations that can eliminate  robustness-driven cycles. In the extreme-state case, stability is driven by the collapse of worst-case beliefs to the most adverse outcomes.
\section{Conclusion}

We developed  a unified framework that combines a concern for model misspecification with a preference for simplicity. By penalizing the entropy of worst-case robust distortions, our agent balances pessimism against the plausibility of adversarial narratives. This mechanism yields sharp dynamic predictions: it breaks the indifference in standard robust control by penalizing risky actions that admit simple disaster narratives, thereby eliminating cycles, favoring safety and long-run stability. Our framework also offers microfoundations for empirical phenomena such as scale heterogeneity, probability neglect, and home bias. By disciplining robustness with Occam’s razor, our framework establishes a formal link between statistical learning and behavioral regularities in decision-making under uncertainty. We focused on myopic decision making to avoid the complications arising from combining misspecification concerns and dynamic consistency  flagged by \citet{lanzani25};  we hope that  future research will  extend our framework to forward-looking agents.

\appendix

\section{Appendix: Proofs from  Section \ref{sec:static}}\label{app:proofs_static}

\subsection{Proof of Lemma \ref{lem:closed_form}}

Fix $(\lambda,\mu)$ with $\kappa := 1/\lambda - \mu > 0$, an action $a \in A$, and a model $\m{q} \in Q$.
Define
$$
  \Phi(\lambda,\mu,p_a;\m{q})
  := \sum_{y \in Y} u(a,y) p_a(y)
     + \frac{1}{\lambda} R(p_a \Vert q_a)
     + \mu H(p_a).
$$
We solve $\min_{p_a \in \Delta(Y)} \Phi(\lambda,\mu,p_a;\m{q})$.

\paragraph{Step 1: Existence and Uniqueness.}
Rewrite the objective as
$$
  \Phi(\lambda,\mu,p_a;\m{q})
  = \sum_{y \in Y} \left[ u(a,y) p_a(y) + \kappa \, p_a(y) \log p_a(y) \right]
    - \frac{1}{\lambda} \sum_{y \in Y} p_a(y) \log q_a(y),
$$
where we use $0 \log 0 := 0$. Since $\kappa > 0$, the function $x \mapsto x \log x$ (with 
$0 \log 0 = 0$) is strictly convex on $[0,1]$, making the objective strictly convex on $\Delta(Y)$. 
The simplex $\Delta(Y)$ is compact, so a unique global minimizer $p^* \in \Delta(Y)$ exists.

\paragraph{Step 2: Interiority of the Minimizer.}
Suppose $p^*(y_0) = 0$ for some $y_0 \in Y$. Then there exists $y_1 \in Y$ with $p^*(y_1) > 0$.
For $t \in (0, p^*(y_1))$, define the perturbation $p_t$ by
$$
  p_t(y) = \begin{cases}
    t & \text{if } y = y_0 \\
    p^*(y_1) - t & \text{if } y = y_1 \\
    p^*(y) & \text{otherwise.}
  \end{cases}
$$
Note that $p_t \in \Delta(Y)$ by construction.  Define
$\psi_y(x):=u(a,y)x+\kappa x\log x-\frac{1}{\lambda}x\log q_a(y),$ for $x\in[0,1],$
with the convention $0\log 0:=0$. Then,
$\Phi(\lambda,\mu,p_t;\m{q})-\Phi(\lambda,\mu,p^*;\m{q})
=
\big[\psi_{y_0}(t)-\psi_{y_0}(0)\big]
+
\bigl[\psi_{y_1}(p^*(y_1)-t)-\psi_{y_1}(p^*(y_1))\bigr].$ Dividing by $t$ and taking $t \downarrow 0^+$:

For the first term:
$$
\lim_{t\downarrow 0^+}\frac{\psi_{y_0}(t)-\psi_{y_0}(0)}{t}
=
\lim_{t\downarrow 0^+}\left[u(a,y_0)+\kappa\log t-\frac{1}{\lambda}\log q_a(y_0)\right]
=
-\infty,
$$
since $\kappa > 0$, $q_a(y_0)>0$ by Assumption \ref{ass:full_support}, and $\lim_{t \downarrow 0^+} \log t = -\infty$. For the second term: Since $p^*(y_1) > 0$, the function $\psi_{y_1}$ is differentiable in a 
neighborhood of $p^*(y_1)$ (away from 0), so
$$
\lim_{t\downarrow 0^+}
\frac{\psi_{y_1}(p^*(y_1)-t)-\psi_{y_1}(p^*(y_1))}{-t}
=
\psi_{y_1}'(p^*(y_1))\in\mathbb{R}.
$$
Hence,
$\lim_{t\downarrow 0^+}
\frac{1}{t}\big(\Phi(\lambda,\mu,p_t;\m{q})-\Phi(\lambda,\mu,p^*;\m{q})\big)
=
-\infty.$
Therefore, $\Phi(\lambda,\mu,p_t;\m{q})<\Phi(\lambda,\mu,p^*;\m{q})$ for all sufficiently small $t>0$, contradicting optimality of $p^*$. Thus, $p^*(y)>0$ for all $y\in Y$.

\paragraph{Step 3: First-Order Conditions and Closed Form.}
Since $p^*$ is interior, standard Lagrangian analysis applies. Let $\gamma$ be the multiplier  
for $\sum_y p_a(y) = 1$. The Lagrangian is
$$
  \mathcal{L}(p_a, \gamma)
  = \sum_{y \in Y} u(a,y) p_a(y)
    + \frac{1}{\lambda} \sum_{y \in Y} p_a(y) \log \frac{p_a(y)}{q_a(y)}
    + \mu \sum_{y \in Y} (-p_a(y) \log p_a(y))
    + \gamma \Big(\sum_{y \in Y} p_a(y) - 1\Big).
$$
Taking $\partial \mathcal{L} / \partial p_a(y) = 0$ for each $y$:
$$
  u(a,y) + \frac{1}{\lambda}(\log p_a(y) - \log q_a(y) + 1)
  + \mu(-\log p_a(y) - 1) + \gamma = 0.
$$
Rearranging: $\left(\frac{1}{\lambda} - \mu\right) \log p_a(y)
  = \frac{1}{\lambda} \log q_a(y) - u(a,y) + C,$
where $C$ is independent of $y$. Using $\kappa = 1/\lambda - \mu > 0$ and $\beta = 1/(\lambda\kappa)$: $\log p_a(y) = -\frac{u(a,y)}{\kappa} + \beta \log q_a(y) + \text{constant}.$
Exponentiating and normalizing yields: $\hat{p}_{\lambda,\mu}(a;\m{q})(y)
  = \frac{\exp\{-u(a,y)/\kappa\} q_a(y)^\beta}
         {\sum_{z \in Y} \exp\{-u(a,z)/\kappa\} q_a(z)^\beta}.$

\paragraph{Step 4: Continuity and Envelope Identities.}
Continuity of $\hat p_{\lambda,\mu}(a;\m{q})$ in
$(\lambda,\mu,\m{q})$ on $\{\kappa>0\}$ follows directly from the closed-form expression
\eqref{eq:p_star_closed_form}, since $q_a(y)>0$ for all $y$ by Assumption \ref{ass:full_support} and the
denominator in \eqref{eq:p_star_closed_form} is strictly positive. Hence,
$v_{\lambda,\mu}(a;\m{q})=\Phi(\lambda,\mu,\hat p_{\lambda,\mu}(a;\m{q});\m{q})$ is continuous in
$(\lambda,\mu,\m{q})$ on $\{\kappa>0\}$. For the envelope identities, since
$\hat p_{\lambda,\mu}(a;\m{q})$ is the unique minimizer and lies in the relative interior of
$\Delta(Y)$, we can eliminate one coordinate using $\sum_{y\in Y}p_a(y)=1$ and apply the
standard envelope theorem to the resulting unconstrained problem. Therefore,
$\frac{\partial v_{\lambda,\mu}}{\partial \mu}(a;\m{q})
  = \frac{\partial \Phi}{\partial \mu}\Big|_{p_a = \hat{p}_{\lambda,\mu}(a;\m{q})}
  = H(\hat{p}_{\lambda,\mu}(a;\m{q})),$
$\frac{\partial v_{\lambda,\mu}}{\partial \lambda}(a;\m{q})
  = \frac{\partial \Phi}{\partial \lambda}\Big|_{p_a = \hat{p}_{\lambda,\mu}(a;\m{q})}
  = -\frac{1}{\lambda^2} R(\hat{p}_{\lambda,\mu}(a;\m{q}) \Vert q_a).$
which coincide with \eqref{eq:envelope_mu}. \qed

\subsection{Proof of Proposition \ref{prop:entropy_decreasing_mu}}
\begin{lemma}[Entropy in a finite exponential family]\label{lem:entropy_exp_family}
Let $Y$ be a finite set and $\varphi:Y\to\mathbb{R}$ be nonconstant. For each $t>0$ define
$$
  Z(t) := \sum_{y\in Y} e^{t\varphi(y)},
  \qquad
  A(t) := \log Z(t),
  \qquad
  p_t(y) := \frac{e^{t\varphi(y)}}{Z(t)} \quad (y\in Y),
$$
and let $H(t) := -\sum_{y\in Y} p_t(y)\log p_t(y)$ be the Shannon entropy of $p_t$. Then:

\begin{enumerate}
  \item The function $A$ is twice continuously differentiable on $(0,\infty)$, and for all
  $t>0$,
  $$
    A'(t) = \sum_{y\in Y} \varphi(y)p_t(y),
    \qquad
    A''(t) = \sum_{y\in Y} \big(\varphi(y) - A'(t)\big)^2 p_t(y).
  $$
  \item For all $t>0$, $H(t) = A(t) - t A'(t).$
  \item For all $t>0$, $H'(t)
    = -t A''(t)
    = -t\sum_{y\in Y} \big(\varphi(y) - A'(t)\big)^2 p_t(y).$
   In particular, since $\varphi$ is nonconstant and $p_t(y)>0$ for all $y$, we have
  $A''(t)>0$ and hence $H'(t)<0$ for all $t>0$. Thus $H$ is strictly decreasing on
  $(0,\infty)$.
\end{enumerate}
\end{lemma}

\begin{proof}[Proof of Lemma \ref{lem:entropy_exp_family}]
\emph{Step 1: Formulas for $A'(t)$ and $A''(t)$.}
Since $Y$ is finite and each map $t\mapsto e^{t\varphi(y)}$ is smooth, the function
$Z(t)=\sum_y e^{t\varphi(y)}$ is smooth on $\mathbb{R}$ and in particular $C^2$ on
$(0,\infty)$. Thus $A(t)=\log Z(t)$ is also $C^2$ on $(0,\infty)$.

Differentiating $A(t)=\log Z(t)$ gives $A'(t) = \frac{Z'(t)}{Z(t)}.$
We have $Z'(t) = \sum_{y\in Y} \varphi(y)e^{t\varphi(y)},$
so
$$
  A'(t)
  = \frac{\sum_y \varphi(y)e^{t\varphi(y)}}{\sum_z e^{t\varphi(z)}}
  = \sum_{y\in Y} \varphi(y)\frac{e^{t\varphi(y)}}{Z(t)}
  = \sum_{y\in Y} \varphi(y)p_t(y).
$$

For $A''(t)$, differentiate $A'(t)=Z'(t)/Z(t)$: $A''(t)
  = \frac{Z''(t)Z(t) - [Z'(t)]^2}{[Z(t)]^2}.$
Here, $Z''(t) = \sum_{y\in Y} \varphi(y)^2 e^{t\varphi(y)}.$
Therefore,
\begin{align*}
  A''(t)
  = \frac{\big(\sum_y \varphi(y)^2 e^{t\varphi(y)}\big)\big(\sum_z e^{t\varphi(z)}\big) -
           \big(\sum_y \varphi(y)e^{t\varphi(y)}\big)^2}
          {\big(\sum_z e^{t\varphi(z)}\big)^2}&= \sum_{y\in Y} \varphi(y)^2 \frac{e^{t\varphi(y)}}{Z(t)}
     - \Big(\sum_{y\in Y} \varphi(y)\frac{e^{t\varphi(y)}}{Z(t)}\Big)^2\\
  &= \sum_{y\in Y} \varphi(y)^2 p_t(y)
     - \Big(\sum_{y\in Y} \varphi(y)p_t(y)\Big)^2.
\end{align*}
Using the expression for $A'(t)$, we can rewrite this as $A''(t)
  = \sum_{y\in Y} \big(\varphi(y)-A'(t)\big)^2 p_t(y).$
This proves part (1).

\emph{Step 2: Representation of $H(t)$.}
By definition, $\log p_t(y) = t\varphi(y) - \log Z(t) = t\varphi(y) - A(t),$
so
\begin{align*}
  H(t)= -\sum_{y\in Y} p_t(y)\log p_t(y)= -\sum_y p_t(y)\big(t\varphi(y) - A(t)\big)&= -t\sum_y \varphi(y)p_t(y) + A(t)\sum_y p_t(y)\\
  &= -t A'(t) + A(t),
\end{align*}
where we used the expression for $A'(t)$ from Step 1 and the fact that
$\sum_y p_t(y)=1$. This proves part (2).

\emph{Step 3: Derivative of $H(t)$ and its sign.}
Differentiating $H(t)=A(t)-tA'(t)$, $H'(t)
  = A'(t) - A'(t) - tA''(t)
   = -t A''(t).$
Substituting the expression for $A''(t)$ from part (1) gives
$H'(t)
  = -t\sum_{y\in Y} \big(\varphi(y)-A'(t)\big)^2 p_t(y).$
For each $t>0$, $p_t(y)>0$ for all $y\in Y$. Since $\varphi$ is nonconstant,
$\varphi(y)$ is not a.s. constant under $p_t$, and thus
$\sum_y(\varphi(y)-A'(t))^2 p_t(y)>0$. Thus, $A''(t)>0$ and $H'(t)<0$ for all $t>0$,
showing that $H$ is strictly decreasing on $(0,\infty)$.
\end{proof}

\begin{proof}[Proof of Proposition \ref{prop:entropy_decreasing_mu}]
By Lemma \ref{lem:closed_form}, for each $\mu < 1/\lambda$ we have $\kappa := 1/\lambda - \mu > 0$ 
and $\hat{p}_{\lambda,\mu}(a;\m{q})(y) = \frac{\exp\{\beta(\mu) \varphi(y)\}}
         {\sum_{z\in Y}\exp\{\beta(\mu) \varphi(z)\}},$
where $\varphi(y) := \log q_a(y) - \lambda u(a,y)$, $\beta(\mu) := 1/(1-\lambda\mu)$, 
and $\varphi$ is nonconstant by the payoff-likelihood condition. Identify $\hat{p}_{\lambda,\mu}(a;\m{q})$ with the exponential family $p_t(y) \propto \exp\{t \varphi(y)\}$ 
at $t = \beta(\mu)$. By Lemma \ref{lem:entropy_exp_family}, the entropy $H(p_t)$ is strictly 
decreasing in $t > 0$. Since $\beta(\mu)$ is strictly increasing in $\mu$ (derivative: 
$\beta'(\mu) = \lambda/(1-\lambda\mu)^2 > 0$), the composition $\mu \mapsto H(\hat{p}_{\lambda,\mu}(a;\m{q}))$ 
is strictly decreasing on $(-\infty, 1/\lambda)$.
\end{proof}

\section{Appendix: Proofs from Section \ref{sec:dynamic}}\label{app:proofs_dynamic}

\subsection{Proof of Proposition \ref{prop:dynamic_selection}}

The proof is a direct adaptation of \citet[][Theorem 3]{lanzani25}, so we
 only verify the points at which our dynamic criterion differs from his. Fix a $\bar{\mu}$-optimal policy $\sigma$ and let $\alpha^\Lambda$ be a $\Lambda$-limit
frequency of $\sigma$. 

\par First, since $A$ is finite, all norms on $\Delta(A)$ are equivalent and our Definition
\ref{def:limit_freq}, stated in terms of $\|\cdot\|_1$, is equivalent to the notion of
$\Lambda$-limit frequency used in \citet{lanzani25}. Second, the stochastic objects governing learning are unchanged relative to
\citet{lanzani25}: beliefs are updated by Bayes rule over the same finite set $Q$, and the
misspecification concern still satisfies
$\lambda_t=\Lambda(h_t)=\frac{1}{ct}\operatorname{LLR}(h_t;Q).$
Accordingly, the parts of the proof of \citet[][Theorem 3]{lanzani25} that identify the
asymptotic misspecification concern and the asymptotic set of best-fitting models apply
verbatim here. In particular, if
$\tau(\alpha):=\frac{1}{c}\min_{\m{q}\in Q}D(\m{q};\alpha)$ where $D(\m{q};\alpha):=\sum_{a\in A}\alpha(a)R(p^\star_a\Vert q_a),$
then, along every history sequence considered in the proof of
\citet[][Theorem 3]{lanzani25}, the empirical frequencies determine the same limit concern
$\tau(\alpha)$ and the posterior concentrates on
$Q(\alpha)=\arg\min_{\m{q}\in Q}D(\m{q};\alpha).$

The last point to be  checked is whether the one-period best-reply
correspondence has non empty closed values and a compact graph. By Lemma \ref{lem:closed_form}, for each $a\in A$ the map
$(\lambda,\mu,\m{q})\longmapsto v_{\lambda,\mu}(a;\m{q})$
is continuous on $\{\kappa>0\}$, where $\kappa=\frac{1}{\lambda}-\mu$. Under Assumption
\ref{ass:kappa_pos}, along the dynamic path we always have $\mu=\bar{\mu}\lambda$ and
$\kappa=\frac{1}{\lambda}-\bar{\mu}\lambda>0.$
Since $A$ is finite, the desired properties hold, and the proof of
\citet[][Theorem 3]{lanzani25} applies verbatim with
$\operatorname{BR}^{\bar{\mu}}_\lambda$ in place of $\operatorname{BR}_\lambda$, while noting that,  $\forall\bar{\mu}\in\R$: $\operatorname{BR}^{\bar{\mu}}_0\equiv\operatorname{BR}_0$ is the SEU best reply. It follows that there
exists some $\eta^{\bar{\mu}}\in\Delta(Q)$ such that
$\text{\normalfont supp}(\eta^{\bar{\mu}})\subset Q(\alpha^\Lambda)$ and $\alpha^\Lambda\in\Delta\bigl(\operatorname{BR}^{\bar{\mu}}_{\tau(\alpha^\Lambda)}(\eta^{\bar{\mu}})\bigr).$
Define
$\tau^{\bar{\mu}}:=\tau(\alpha^\Lambda)
  =
  \frac{1}{c}\min_{\m{q}\in Q}D(\m{q};\alpha^\Lambda).$
Since
$\alpha^\Lambda\in
\Delta\bigl(\operatorname{BR}^{\bar{\mu}}_{\tau^{\bar{\mu}}}(\eta^{\bar{\mu}})\bigr),$
every action $a\in A$ with $\alpha^\Lambda(a)>0$ belongs to
$\operatorname{BR}^{\bar{\mu}}_{\tau^{\bar{\mu}}}(\eta^{\bar{\mu}}).$
Together with
$\text{\normalfont supp}(\eta^{\bar{\mu}})\subset Q(\alpha^\Lambda)$
and the definition of $\tau^{\bar{\mu}}$, this is exactly Definition
\ref{def:mixed_equilibrium_mu}. Therefore,
$(\alpha^\Lambda,\eta^{\bar{\mu}},\tau^{\bar{\mu}})$
is a mixed $c$-robust equilibrium for the $\bar{\mu}$-criterion. \qed

\subsection{Proof of Proposition \ref{prop:correct_spec_limit}}
Fix a $\bar{\mu}$-optimal policy $\sigma$ and let
$E^\Lambda:=\{\alpha_t\to\alpha^\Lambda\}.$
Fix an action $a\in A$ with $\alpha^\Lambda(a)>0$. We show that on $E^\Lambda$,
$V_{\lambda_t,\bar{\mu}\lambda_t}(a;\pi_t)\longrightarrow \E_{p^\star_a}[u(a,y)]$ $\mathbb{P}_\sigma\text{-a.s.}$

Let $\m{q}^\star\in Q$ be as in Assumption \ref{ass:correct_spec}. Then, 
$\forall\alpha\in\Delta(A)$,
$D(\m{q}^\star;\alpha)
  =
  \sum_{a'\in A}\alpha(a')R(p^\star_{a'}\Vert q^\star_{a'})
  =
  0.$
Hence,
$\min_{\m{q}\in Q}D(\m{q};\alpha^\Lambda)=0.$
On the event $E^\Lambda$, every action $a'\in A$ with $\alpha^\Lambda(a')>0$ is played
infinitely often, and the empirical distribution of outcomes following action $a'$ converges
to $p^\star_{a'}$ by the strong law of large numbers. Therefore, the hypotheses of
\citeauthor{lanzani25}'s (\citeyear[][Lemma 6]{lanzani25})  are satisfied on $E^\Lambda$, and
$$
  \lambda_t
  =
  \frac{1}{ct}\operatorname{LLR}(h_t;Q)
  \longrightarrow
  \frac{1}{c}\min_{\m{q}\in Q}D(\m{q};\alpha^\Lambda)
  =
  0
  \qquad
  \mathbb{P}_\sigma\text{-a.s. on }E^\Lambda.
$$
Since $\mu_t=\bar{\mu}\lambda_t$, it follows that
$\mu_t\longrightarrow 0$
   $\mathbb{P}_\sigma\text{-a.s. on }E^\Lambda.$

\par Next, since the prior $\pi_0$ has full support on $Q$, the same Bayesian
likelihood-ratio argument used in the proof of \citet[][Theorem 3]{lanzani25} implies that
on $E^\Lambda$ the posterior $\pi_t$ concentrates on $Q(\alpha^\Lambda)$. Since
$\m{q}^\star\in Q$ and $q^\star_a=p^\star_a$ for all $a\in A$, we have $D(\m{q}^\star;\alpha^\Lambda)
  =
  \sum_{a'\in A}\alpha^\Lambda(a')R(p^\star_{a'}\Vert q^\star_{a'})
  =
  0.$
Therefore, every $\m{q}\in Q(\alpha^\Lambda)$ satisfies
$D(\m{q};\alpha^\Lambda)=0$. Since $\alpha^\Lambda(a)>0$, it follows that $0
  =
  D(\m{q};\alpha^\Lambda)
  =
  \sum_{a'\in A}\alpha^\Lambda(a')R(p^\star_{a'}\Vert q_{a'}),$
and because every term in the sum is nonnegative, we must have
$R(p^\star_a\Vert q_a)=0$, hence $q_a=p^\star_a$. Thus, if we define
$Q^\star_a:=\{\m{q}\in Q: q_a=p^\star_a\},$
then
$\pi_t(Q^\star_a)\longrightarrow 1$ $\mathbb{P}_\sigma\text{-a.s. on }E^\Lambda.$

Fix $\m{q}\in Q^\star_a$. Since $\mu_t\ge 0$, we have
\begin{equation}\label{eq:squeeze_v_correct_spec}
  v_{\lambda_t,0}(a;\m{q})
  \le
  v_{\lambda_t,\mu_t}(a;\m{q})
  \le
  \E_{q_a}[u(a,y)] + \mu_t H(q_a),
\end{equation}
where the upper bound uses feasibility of $p_a=q_a$ in \eqref{eq:static_criterion}. For
$\mu=0$, we get
$v_{\lambda,0}(a;\m{q})
  =
  -\frac{1}{\lambda}\log\big(\sum_{y\in Y}q_a(y)e^{-\lambda u(a,y)}\big),$
so $v_{\lambda,0}(a;\m{q})\to \E_{q_a}[u(a,y)]$ as $\lambda\downarrow0$. Since
$q_a=p^\star_a$ and $\mu_t\to0$, \eqref{eq:squeeze_v_correct_spec} implies
$v_{\lambda_t,\mu_t}(a;\m{q})
  \longrightarrow
  \E_{p^\star_a}[u(a,y)]$ 
   $\mathbb{P}_\sigma\text{-a.s. on }E^\Lambda.$

Finally, because $Q$ is finite and $u$ is bounded, the family
$\{v_{\lambda_t,\mu_t}(a;\m{q}):t\ge1,\ \m{q}\in Q\}$ is uniformly bounded. Hence
$V_{\lambda_t,\mu_t}(a;\pi_t)
  =
  \sum_{\m{q}\in Q^\star_a}\pi_t(\m{q})v_{\lambda_t,\mu_t}(a;\m{q})
  +
  \sum_{q\notin Q^\star_a}\pi_t(\m{q})v_{\lambda_t,\mu_t}(a;\m{q}).$
The second term converges to $0$ $\mathbb{P}_\sigma$-a.s. on $E^\Lambda$ because
$\pi_t(Q^\star_a)\to1$. The first term converges to $\E_{p^\star_a}[u(a,y)]$
$\mathbb{P}_\sigma$-a.s. on $E^\Lambda$ because every $\m{q}\in Q^\star_a$ has
$q_a=p^\star_a$. Therefore, we conclude
$V_{\lambda_t,\bar{\mu}\lambda_t}(a;\pi_t)
  \longrightarrow
  \E_{p^\star_a}[u(a,y)]$ $\mathbb{P}_\sigma\text{-a.s. on }E^\Lambda.$ \qed

\subsection{Proof of Lemma \ref{lem:Delta_derivatives}}

Fix $(\lambda,\bar{\mu},\pi)$ such that $\kappa := 1/\lambda - \bar{\mu}\lambda > 0$. Set
$\mu = \bar{\mu}\lambda$. For each action $a \in \{r,s\}$ and model $\m{q} \in Q$,
Lemma \ref{lem:closed_form} establishes that $v_{\lambda,\mu}(a;\m{q})$ is differentiable in
$(\lambda,\mu)$. The aggregate worst-case value is
$V_{\lambda,\mu}(a;\pi) = \sum_{\m{q} \in Q} v_{\lambda,\mu}(a;\m{q}) \pi(\m{q}).$
Since $Q$ is finite, we may differentiate term-by-term. By the envelope theorem
(Lemma \ref{lem:closed_form}, eq. \eqref{eq:envelope_mu}):
$$
  \frac{\partial}{\partial\mu} v_{\lambda,\mu}(a;\m{q}) = H(\hat{p}_{\lambda,\mu}(a;\m{q})),
  \qquad
  \frac{\partial}{\partial\lambda} v_{\lambda,\mu}(a;\m{q})
  = -\frac{1}{\lambda^2} R(\hat{p}_{\lambda,\mu}(a;\m{q}) \Vert q_a).
$$

\paragraph{Part 1: Derivative with respect to $\bar{\mu}$.}
Using the chain rule with $\mu = \bar{\mu}\lambda$:
$$
  \frac{\partial}{\partial\bar{\mu}} v_{\lambda,\bar{\mu}\lambda}(a;\m{q})
  = \frac{\partial v}{\partial\mu} \cdot \frac{\partial\mu}{\partial\bar{\mu}}
  = H(\hat{p}_{\lambda,\bar{\mu}\lambda}(a;\m{q})) \cdot \lambda.
$$
Summing over models: $\frac{\partial}{\partial\bar{\mu}} V_{\lambda,\bar{\mu}\lambda}(a;\pi)
  = \lambda \sum_{\m{q} \in Q} \pi(\m{q}) H(\hat{p}_{\lambda,\bar{\mu}\lambda}(a;\m{q}))
  = \lambda H_a(\lambda,\bar{\mu},\pi).$
Therefore, $\frac{\partial}{\partial\bar{\mu}} \Delta(\lambda,\bar{\mu},\pi)
  = \frac{\partial}{\partial\bar{\mu}} V_{\lambda,\bar{\mu}\lambda}(r;\pi)
    - \frac{\partial}{\partial\bar{\mu}} V_{\lambda,\bar{\mu}\lambda}(s;\pi)
  = \lambda \big( H_r(\lambda,\bar{\mu},\pi) - H_s(\lambda,\bar{\mu},\pi) \big).$

\paragraph{Part 2: Sign on switching surface.}
On the switching surface $S = \{(\lambda,\bar{\mu},\pi) : \Delta(\lambda,\bar{\mu},\pi) = 0\}$,
the sign of $\partial\Delta/\partial\bar{\mu}$ equals the sign of
$H_r(\lambda,\bar{\mu},\pi) - H_s(\lambda,\bar{\mu},\pi)$ (since $\lambda > 0$).
\qed

\subsection{Proof of Theorem \ref{thm:mu_kills_cycles} and Corollary \ref{cor:nonvacuity_theorem1}}

\begin{proof}[Proof of Theorem \ref{thm:mu_kills_cycles}]
Define the domain $\mathcal{D} := [\underline{\lambda},\bar{\lambda}] \times \Delta(Q)$.
Let $(\bar{\mu}_0,H^\ast)$ be as in Assumption \ref{ass:uniform_entropy_gap} and set
$\hat{m} := \underline{\lambda} H^\ast > 0$ and
  $K_0 := \sup_{(\lambda,\pi)\in\mathcal{D}} \max\{\Delta(\lambda,\bar{\mu}_0,\pi),0\}.$
Since $\Delta$ is continuous and $\mathcal{D}$ is compact, $K_0<\infty$. Define
$\varepsilon
  :=
  \frac{1}{2}\big(\frac{1}{\bar{\lambda}^2}-\bar{\mu}_0-\frac{K_0}{\hat{m}}\big)_+$ and $
  \bar{\mu}^\ast
  :=
  \bar{\mu}_0+\frac{K_0}{\hat{m}}+\varepsilon,$
where $(x)_+ := \max\{x,0\}$. Note that if $\varepsilon>0$, then $\bar{\mu}^\ast < 1/\bar{\lambda}^2$ and
$\bar{\mu}^\ast > \bar{\mu}_0 + K_0/\hat{m}$; if $\varepsilon=0$, then
$\bar{\mu}^\ast \ge 1/\bar{\lambda}^2$, so the interval $[\bar{\mu}^{\ast},1/\bar{\lambda}^{2})$ becomes empty and hence the theorem would hold
vacuously. We may therefore restrict attention to $\bar{\mu}\in[\bar{\mu}^*, 1/\bar{\lambda}^2)$.

\paragraph{Step 1: Equilibrium concerns lie in $[\underline{\lambda},\bar{\lambda}]$.}
Let $(\alpha^{\bar{\mu}},\eta^{\bar{\mu}},\tau^{\bar{\mu}})$ be any mixed $c$-robust equilibrium
for the $\bar{\mu}$-criterion with $\alpha^{\bar{\mu}}(r)>0$.
By Assumption \ref{ass:U}, $\tau^{\bar{\mu}}\in[\underline{\lambda},\bar{\lambda}]$.
\paragraph{Step 2: Uniform strict negativity of $\Delta$ for large $\bar{\mu}$.}
Fix any $(\lambda,\pi)\in\mathcal{D}$ and define $f(\bar{\mu}) := \Delta(\lambda,\bar{\mu},\pi)$ 
on the admissible range $\{\bar{\mu}:\bar{\mu}\lambda^2<1\}$. 
By Lemma \ref{lem:closed_form} and Lemma \ref{lem:Delta_derivatives}, 
$\Delta(\lambda,\bar{\mu},\pi)$ is continuously differentiable in $\bar{\mu}$ 
whenever $\bar{\mu}\lambda^2<1$, so $f$ is $C^1$ on its domain. We claim that for every $(\lambda,\pi)\in\mathcal{D}$ and every admissible
$\bar{\mu}>\bar{\mu}_0+K_0/\hat{m}$, we have $f(\bar{\mu})<0$.

\emph{Case 1: $f(\bar{\mu}_0)\le 0$.}
Suppose, toward a contradiction, that there exists an admissible $\bar{\mu}_1>\bar{\mu}_0$
with $f(\bar{\mu}_1)\ge 0$.

If $f(\bar{\mu}_1)>0$, then by continuity the set
$Z:=\{\bar{\mu}\in[\bar{\mu}_0,\bar{\mu}_1]: f(\bar{\mu})=0\}$
is nonempty. Let $\hat{\mu}:=\max Z$. Then, $f(\bar{\mu})>0$ for every
$\bar{\mu}\in(\hat{\mu},\bar{\mu}_1]$. Hence, Assumption \ref{ass:uniform_entropy_gap}
and Lemma \ref{lem:Delta_derivatives} imply that
$f'(\bar{\mu})\le -\hat{m}$ for all $\bar{\mu}\in(\hat{\mu},\bar{\mu}_1]$.
Integrating from $\hat{\mu}$ to $\bar{\mu}_1$ gives
$f(\bar{\mu}_1)-f(\hat{\mu})
  =
  \int_{\hat{\mu}}^{\bar{\mu}_1} f'(t)\hspace{0.02in}dt
  \le
  -\hat{m}(\bar{\mu}_1-\hat{\mu})<0.$
Since $f(\hat{\mu})=0$, this yields $f(\bar{\mu}_1)<0$, a contradiction.

If instead $f(\bar{\mu}_1)=0$, then Assumption \ref{ass:uniform_entropy_gap}
and Lemma \ref{lem:Delta_derivatives} imply that
$f'(\bar{\mu}_1)\le -\hat{m}<0.$
Since $f$ is $C^1$ on the admissible range, there exists $\delta>0$ small enough such that
$[\bar{\mu}_1-\delta,\bar{\mu}_1]$ is admissible and
$f'(t)\le -\hat{m}/2$ $\text{for all } t\in[\bar{\mu}_1-\delta,\bar{\mu}_1].$
Therefore,
$$
  f(\bar{\mu}_1)-f(\bar{\mu}_1-\delta)
  =
  \int_{\bar{\mu}_1-\delta}^{\bar{\mu}_1} f'(t)\hspace{0.02in}dt
  \le
  -(\hat{m}/2)\delta<0,
$$
so $f(\bar{\mu}_1-\delta)>0$. Since $f(\bar{\mu}_0)\le 0$, continuity implies that the set
$\widetilde{Z}:=\{\bar{\mu}\in[\bar{\mu}_0,\bar{\mu}_1-\delta]: f(\bar{\mu})=0\}$
is nonempty. Let $\widetilde{\mu}:=\max \widetilde{Z}$. Then $f(\bar{\mu})>0$ for every
$\bar{\mu}\in(\widetilde{\mu},\bar{\mu}_1-\delta]$. Hence, Assumption
\ref{ass:uniform_entropy_gap} and Lemma \ref{lem:Delta_derivatives} imply that
$f'(\bar{\mu})\le -\hat{m}$ for all $\bar{\mu}\in(\widetilde{\mu},\bar{\mu}_1-\delta]$.
Integrating from $\widetilde{\mu}$ to $\bar{\mu}_1-\delta$ gives
$f(\bar{\mu}_1-\delta)-f(\widetilde{\mu})
  =
  \int_{\widetilde{\mu}}^{\bar{\mu}_1-\delta} f'(t)\hspace{0.02in}dt
  \le
  -\hat{m}(\bar{\mu}_1-\delta-\widetilde{\mu})<0,$ 
contradicting $f(\widetilde{\mu})=0$ and $f(\bar{\mu}_1-\delta)>0$.
Therefore, if $f(\bar{\mu}_0)\le 0$, then $f(\bar{\mu})<0$ for all admissible
$\bar{\mu}>\bar{\mu}_0$.

\emph{Case 2: $f(\bar{\mu}_0)>0$.}
Define $\bar{\mu}_1:=\bar{\mu}_0+\frac{f(\bar{\mu}_0)}{\hat{m}},$
so $\bar{\mu}_1\le \bar{\mu}_0+K_0/\hat{m}$. We claim that, whenever $\bar{\mu}_1$ is admissible,
$f(\bar{\mu}_1)\le 0$. Suppose not, so $f(\bar{\mu}_1)>0$. If there exists $\tilde{\mu}\in[\bar{\mu}_0,\bar{\mu}_1]$ with $f(\tilde{\mu})=0$, let
$\hat{\mu}:=\max\{\bar{\mu}\in[\tilde{\mu},\bar{\mu}_1]:f(\bar{\mu})=0\}.$
Then, $f(\bar{\mu})>0$ for every $\bar{\mu}\in(\hat{\mu},\bar{\mu}_1]$, and therefore
$f'(\bar{\mu})\le -\hat{m}$ for all $\bar{\mu}\in(\hat{\mu},\bar{\mu}_1].$
Integrating from $\hat{\mu}$ to $\bar{\mu}_1$ gives
$$
  f(\bar{\mu}_1)-f(\hat{\mu})
  =
  \int_{\hat{\mu}}^{\bar{\mu}_1}f'(t)\hspace{0.02in}dt
  \le
  -\hat{m}(\bar{\mu}_1-\hat{\mu})<0,
$$
contradicting $f(\hat{\mu})=0$ and $f(\bar{\mu}_1)>0$. Hence, there is no zero in
$[\bar{\mu}_0,\bar{\mu}_1]$, so $f(\bar{\mu})>0$ throughout that interval. Therefore, Assumption \ref{ass:uniform_entropy_gap} applies on all of
$[\bar{\mu}_0,\bar{\mu}_1]$, giving
$f'(\bar{\mu})\le -\hat{m}$
  for all $\bar{\mu}\in[\bar{\mu}_0,\bar{\mu}_1].$
Integrating yields
$$
  f(\bar{\mu}_1)
  =
  f(\bar{\mu}_0)+\int_{\bar{\mu}_0}^{\bar{\mu}_1}f'(t)\hspace{0.02in}dt
  \le
  f(\bar{\mu}_0)-\hat{m}(\bar{\mu}_1-\bar{\mu}_0)
  =
  0,
$$
contradicting $f(\bar{\mu}_1)>0$. Thus $f(\bar{\mu}_1)\le0$. By Case 1,
$f(\bar{\mu})<0$ for all admissible $\bar{\mu}>\bar{\mu}_1$.

Combining the two cases, for every $(\lambda,\pi)\in\mathcal{D}$ we have
$\Delta(\lambda,\bar{\mu},\pi)<0$
  for all admissible $\bar{\mu}>\bar{\mu}_0+\frac{K_0}{\hat{m}}.$
In particular, for every admissible $\bar{\mu}\ge \bar{\mu}^\ast$ with
$\bar{\mu}\bar{\lambda}^2<1$, we have
$\Delta(\lambda,\bar{\mu},\pi)<0$
  for all $(\lambda,\pi)\in\mathcal{D}.$

\paragraph{Step 3: Elimination of mixed equilibria.}
Take any admissible $\bar{\mu}\ge\bar{\mu}^\ast$ with $\bar{\mu}\bar{\lambda}^2<1$ and suppose
$(\alpha^{\bar{\mu}},\eta^{\bar{\mu}},\tau^{\bar{\mu}})$ is a mixed $c$-robust equilibrium.
If $\alpha^{\bar{\mu}}(r)>0$, then by Step 1 we have
$\tau^{\bar{\mu}}\in[\underline{\lambda},\bar{\lambda}]$, hence
$(\tau^{\bar{\mu}},\eta^{\bar{\mu}})\in\mathcal{D}$.
By Step 2, $\Delta(\tau^{\bar{\mu}},\bar{\mu},\eta^{\bar{\mu}})<0,$
i.e.,
$V_{\tau^{\bar{\mu}},\bar{\mu}\tau^{\bar{\mu}}}(r;\eta^{\bar{\mu}})
<
V_{\tau^{\bar{\mu}},\bar{\mu}\tau^{\bar{\mu}}}(s;\eta^{\bar{\mu}})$.
By Definition \ref{def:mixed_equilibrium_mu}(ii), any $a$ with $\alpha^{\bar{\mu}}(a)>0$ must be a
best reply. Since $r$ yields strictly lower value than $s$, we have
$r\notin \operatorname{BR}_{\tau^{\bar{\mu}}}^{\bar{\mu}}(\eta^{\bar{\mu}})$, so
$\alpha^{\bar{\mu}}(r)=0$ in every equilibrium. Thus, $\alpha^{\bar{\mu}}(s)=1$, proving part (i).

\paragraph{Step 4: Proof of part (ii).}
By Proposition \ref{prop:dynamic_selection}, every $\Lambda$-limit frequency $\alpha^\Lambda$ of a
$\bar{\mu}$-optimal policy corresponds to a mixed $c$-robust equilibrium. By Step 3, every such
equilibrium satisfies $\alpha^{\bar{\mu}}(s)=1$. Therefore, every $\Lambda$-limit frequency satisfies
$\alpha^\Lambda(s)=1$, so every $\Lambda$-limit frequency is pure and equals the safe action.
Hence, the mixed $c$-robust-equilibrium $\lambda$-driven cycles between $r$ and $s$ are eliminated,
which proves part (ii).
\end{proof}

\begin{proof}[Proof of Corollary \ref{cor:nonvacuity_theorem1}]
The proof of Theorem \ref{thm:mu_kills_cycles} shows that  the conclusion there holds for every
admissible $\bar{\mu}>\bar{\mu}_0+K_0/(\underline{\lambda}H^\ast)$. Assumption
\ref{ass:initial_risky_advantage} implies
$\bar{\mu}_0+\frac{K_0}{\underline{\lambda}H^\ast}<\frac{1}{\bar{\lambda}^2}.$
Hence, we can choose
$\bar{\mu}^\ast\in\big(\bar{\mu}_0+\frac{K_0}{\underline{\lambda}H^\ast},
\frac{1}{\bar{\lambda}^2}\big),$
which implies that $[\bar{\mu}^\ast,1/\bar{\lambda}^2)$ is nonempty.
\end{proof}

\subsection{Proof of Proposition \ref{prop:multi_thresholds}}

Fix $\lambda \in [\underline{\lambda},\bar{\lambda}]$ and $\bar{\mu}$ satisfying $\bar{\mu}\lambda^2<1$. Using the single-model values $v_{\lambda,\bar{\mu}\lambda}(a;\m{q})$ from
\eqref{eq:static_criterion}, define the arm-specific slope $
  \Gamma_i(\lambda,\bar{\mu})
  :=
  \big[v_{\lambda,\bar{\mu}\lambda}(r_i;\m{q}^H)-v_{\lambda,\bar{\mu}\lambda}(s;\m{q}^H)\big]
  -\big[v_{\lambda,\bar{\mu}\lambda}(r_i;\m{q}^L)-v_{\lambda,\bar{\mu}\lambda}(s;\m{q}^L)\big].$
Since $V_{\lambda,\bar{\mu}\lambda}(a;\pi)$ is affine in $\pi$ by
\eqref{eq:posterior_value}, we then have
\begin{equation}\label{eq:Delta_i_theta_derivative}
  \frac{\partial}{\partial\theta}\Delta_i(\lambda,\bar{\mu},\theta)
  = \Gamma_i(\lambda,\bar{\mu})
  \quad\text{for all }\theta\in[0,1].
\end{equation}

\paragraph{Step 1: $\Delta_i(\lambda,\bar{\mu},\theta)$ is strictly increasing in $\theta$.}
By eq. \eqref{eq:Delta_i_theta_derivative}:
$\frac{\partial}{\partial\theta} \Delta_i(\lambda,\bar{\mu},\theta) = \Gamma_i(\lambda,\bar{\mu}).$
By Assumption \ref{ass:monotone_theta_multi}, $\Gamma_i(\lambda,\bar{\mu}) > 0$, so
$\theta \mapsto \Delta_i(\lambda,\bar{\mu},\theta)$ is strictly increasing.

\paragraph{Step 2: Existence and uniqueness of the threshold.}
By assumption $\Delta_i(\lambda,\bar{\mu},0)<0<\Delta_i(\lambda,\bar{\mu},1)$, so $\Delta_i(\lambda,\bar{\mu},0) < 0$ and
$\Delta_i(\lambda,\bar{\mu},1) > 0$. By Step 1 and the intermediate value theorem, there
exists a unique $\theta_i^\ast(\bar{\mu}) \in (0,1)$ with
$\Delta_i(\lambda,\bar{\mu},\theta_i^\ast(\bar{\mu})) = 0$. Monotonicity gives:
$\Delta_i(\lambda,\bar{\mu},\theta) \ge 0
  \Longleftrightarrow
  \theta \ge \theta_i^\ast(\bar{\mu}).$

\paragraph{Step 3: Comparative statics on the domain of existence.}
 Fix $\bar\mu_2>\bar\mu_1$
such that
$\Delta_i(\lambda,\bar\mu_j,0)<0<\Delta_i(\lambda,\bar\mu_j,1),$ for $j=1,2.$ By Assumption \ref{ass:uniform_entropy_gap_multi} and Lemma \ref{lem:Delta_derivatives}:
$\frac{\partial}{\partial\bar{\mu}} \Delta_i(\lambda,\bar{\mu},\theta)
  = \lambda(H_{r_i} - H_s) \le -\lambda H_i^\ast < 0.$
Thus, $\bar{\mu} \mapsto \Delta_i(\lambda,\bar{\mu},\theta)$ is strictly decreasing for
each fixed $\theta$. For $\bar{\mu}_2 > \bar{\mu}_1$ satisfying $\bar{\mu}_2\lambda^2<1$, if $\theta$ satisfies
$\Delta_i(\lambda,\bar{\mu}_2,\theta) \ge 0$, then since $\Delta_i$ is decreasing in
$\bar{\mu}$: $\Delta_i(\lambda,\bar{\mu}_1,\theta) > \Delta_i(\lambda,\bar{\mu}_2,\theta) \ge 0.$
In particular, $\theta \ge \theta_i^\ast(\bar{\mu}_2)$ implies
$\theta > \theta_i^\ast(\bar{\mu}_1)$, so
$\theta_i^\ast(\bar{\mu}_2) > \theta_i^\ast(\bar{\mu}_1)$.
\qed

\subsection{Proof of Lemma \ref{lem:pairwise_Delta}}

The proof is identical to Lemma \ref{lem:Delta_derivatives}, applied to arbitrary actions
$i,j \in A$ instead of $r,s$. Setting $\mu = \bar{\mu}\lambda$ and using the envelope
theorem:
$$
  \frac{\partial}{\partial\bar{\mu}} \Delta_{ij}(\lambda,\bar{\mu},\pi)
  = \frac{\partial}{\partial\bar{\mu}} \big[ V_{\lambda,\bar{\mu}\lambda}(i;\pi)
    - V_{\lambda,\bar{\mu}\lambda}(j;\pi) \big]
  = \lambda \big( H_i(\lambda,\bar{\mu},\pi) - H_j(\lambda,\bar{\mu},\pi) \big).\qed
$$

\subsection{Proof of Proposition \ref{prop:all_risky_elimination}}

\paragraph{Step 1: Bound on the derivative.}
By Assumption \ref{ass:pairwise_entropy}(i) and Lemma \ref{lem:pairwise_Delta}, for all
$(\lambda,\bar{\mu},\pi) \in [\underline{\lambda},\bar{\lambda}] \times [\bar{\mu}_0, 1/\bar{\lambda}^2) \times \Delta(Q)$: $\frac{\partial}{\partial\bar{\mu}} \Delta_{ij}(\lambda,\bar{\mu},\pi)
  = \lambda(H_i - H_j) \le -\lambda H_{ij}^\ast \le -\underline{\lambda} H_{ij}^\ast.$

\paragraph{Step 2: Integration.}
For $\bar{\mu} \in [\bar{\mu}_0, 1/\bar{\lambda}^2)$ and
$(\lambda,\pi) \in [\underline{\lambda},\bar{\lambda}] \times \Delta(Q)$:
$$
  \Delta_{ij}(\lambda,\bar{\mu},\pi)
  = \Delta_{ij}(\lambda,\bar{\mu}_0,\pi)
    + \int_{\bar{\mu}_0}^{\bar{\mu}} \frac{\partial\Delta_{ij}}{\partial t}\, dt
  \le K_{ij} - \underline{\lambda} H_{ij}^\ast (\bar{\mu} - \bar{\mu}_0),
$$
where $K_{ij}$ is defined in \eqref{eq:Mij_def}.

\paragraph{Step 3: Finding the threshold.}
Set $\Delta_{ij}(\lambda,\bar{\mu},\pi) < 0$, which requires:
$\bar{\mu} > \bar{\mu}_0 + \frac{K_{ij}}{\underline{\lambda} H_{ij}^\ast}.$
Define
$\bar{\mu}^\ast := \bar{\mu}_0 + K_{ij}/(\underline{\lambda} H_{ij}^\ast)$. By Assumption \ref{ass:pairwise_entropy}(ii):
$K_{ij} < \underline{\lambda} H_{ij}^\ast (1/\bar{\lambda}^2 - \bar{\mu}_0),$
so $\bar{\mu}^\ast = \bar{\mu}_0 + \frac{K_{ij}}{\underline{\lambda} H_{ij}^\ast}
  < \bar{\mu}_0 + (1/\bar{\lambda}^2 - \bar{\mu}_0) = 1/\bar{\lambda}^2.$
Thus, $\bar{\mu}^\ast \in [\bar{\mu}_0, 1/\bar{\lambda}^2)$.

\paragraph{Step 4: Conclusion.}
For all $\bar{\mu} \in [\bar{\mu}^\ast, 1/\bar{\lambda}^2)$,
$\lambda \in [\underline{\lambda},\bar{\lambda}]$, and $\pi \in \Delta(Q)$:
$\Delta_{ij}(\lambda,\bar{\mu},\pi) < 0,$ i.e.,
  $V_{\lambda,\bar{\mu}\lambda}(j;\pi) > V_{\lambda,\bar{\mu}\lambda}(i;\pi).$
Hence, $i \notin \operatorname{BR}^{\bar{\mu}}_\lambda(\pi)$ for all such parameters.
\qed

\section{Appendix: proofs from Section  \ref{sec:axiomatization}}

\begin{lemma}
\label{prop:corner}
If $\lambda>0$ and $\mu\geq1/\lambda$, then for every $q\in\Delta({Y})$ and every act $f$, 
\begin{equation}
\label{eq:corner}
 v_{\lambda,\mu}(f;q)
 =\min_{y\in\text{\normalfont supp}(q)}\Big\{u(f(y)) +\frac{1}{\lambda}\log\frac{1}{q(y)}\Big\}.
\end{equation}
\end{lemma}
The proof of Lemma \ref{prop:corner} is standard, so it is in Online Appendix \ref{ola:omittedproofs}.

\begin{proof}[Proof of Theorem \ref{prop:arc_representation}.1] Fix $q\in Q$ and an act $f\in X^{Y}$. If $p\not\ll q$, then $R(p\Vert q)=+\infty$,
so such $p$ cannot attain the minimum in \eqref{eq:axiom_v}. Hence
\begin{equation}\label{eq:restrict_dom_aux}
v_{\lambda,\mu}(f;q)
=\min_{\substack{p\in\Delta(Y)\\ p\ll q}}
\left\{
\E_p[u(f)]+\frac{1}{\lambda}R(p\Vert q)+\mu H(p)
\right\}.
\end{equation}

Fix $p\in\Delta(Y)$ with $p\ll q$. Then,
\begin{align*}
\frac{1}{\lambda}R(p\Vert q)+\mu H(p)
&=\frac{1}{\lambda}\sum_{y\in Y}p(y)\big(\log p(y)-\log q(y)\big)-\mu\sum_{y\in Y}p(y)\log p(y)\\
&=\Big(\tfrac{1}{\lambda}-\mu\Big)\sum_{y\in Y}p(y)\log p(y)-\frac{1}{\lambda}\sum_{y\in Y}p(y)\log q(y)\\
&=\kappa\sum_{y\in Y}p(y)\log p(y)-\frac{1}{\lambda}\sum_{y\in Y}p(y)\log q(y),
\end{align*}
where $\kappa=\frac{1}{\lambda}-\mu>0$. Now define $Z(q)$ and $\tilde q$ by \eqref{eq:axiom_Zq}. For every $y$ with $q(y)>0$, $\log \tilde q(y)=\beta\log q(y)-\log Z(q).$
Since $p\ll q$, taking expectations under $p$ yields
$\sum_{y\in Y}p(y)\log \tilde q(y)
=\beta\sum_{y\in Y}p(y)\log q(y)-\log Z(q),$
and therefore
\begin{align*}
R(p\Vert \tilde q)
=\sum_{y\in Y}p(y)\log\frac{p(y)}{\tilde q(y)}
&=\sum_{y\in Y}p(y)\log p(y)-\sum_{y\in Y}p(y)\log \tilde q(y)\\
&=\sum_{y\in Y}p(y)\log p(y)-\beta\sum_{y\in Y}p(y)\log q(y)+\log Z(q).
\end{align*}
Multiplying by $\kappa$ gives
$\kappa R(p\Vert \tilde q)
=\kappa\sum_{y\in Y}p(y)\log p(y)-\kappa\beta\sum_{y\in Y}p(y)\log q(y)+\kappa\log Z(q).$
By $\beta=\frac{1}{1-\lambda\mu}$ and $\kappa=\frac{1}{\lambda}-\mu=\frac{1-\lambda\mu}{\lambda}$ we have
$\kappa\beta=\frac{1}{\lambda}$. Hence,
$\frac{1}{\lambda}R(p\Vert q)+\mu H(p)
=
\kappa R(p\Vert \tilde q)-\kappa\log Z(q).$
Substituting into \eqref{eq:restrict_dom_aux} yields
$$
v_{\lambda,\mu}(f;q)
=
-\kappa\log Z(q)
+\min_{\substack{p\in\Delta(Y)\\ p\ll q}}
\left\{
\E_p[u(f)]+\kappa R(p\Vert \tilde q)
\right\}.
$$
Since $\tilde q$ is supported on $\{y:q(y)>0\}$, any $p\not\ll q$ also satisfies $p\not\ll \tilde q$ and so
$R(p\Vert \tilde q)=+\infty$; therefore the restriction $p\ll q$ is redundant:
$$
\min_{\substack{p\in\Delta(Y)\\ p\ll q}}
\left\{
\E_p[u(f)]+\kappa R(p\Vert \tilde q)
\right\}
=
\min_{p\in\Delta(Y)}
\left\{
\E_p[u(f)]+\kappa R(p\Vert \tilde q)
\right\}.
$$
Multiply by $\pi(q)$ and sum over $q\in Q$ to obtain
$$
V_{\lambda,\mu}(f;\pi)
=
-\kappa\sum_{q\in Q}\pi(q)\log Z(q)
+\sum_{q\in Q}\pi(q)\min_{p\in\Delta(Y)}\left\{\E_p[u(f)]+\kappa R(p\Vert \tilde q)\right\}.
$$
Define the act-independent constant $\widetilde{C}(\lambda,\mu,\pi):=-\kappa\sum_{q\in Q}\pi(q)\log Z(q)\in\R,$
which gives \eqref{eq:arc_form}. The ARC formulation follows by setting $\lambda_{\mathrm{ARC}}:=1/\kappa$.\end{proof}

\par
\noindent The proof of Theorem \ref{prop:arc_representation}.2 is standard, so it is relegated to  Online Appendix \ref{ola:omittedproofs}.

\begin{proof}[Proof of Theorem \ref{prop:arc_representation}.3] The limit of $V_{\lambda,\mu}(\cdot;\pi)$ as $\lambda\rightarrow0$ is $\underset{\lambda\rightarrow0}{\text{\normalfont lim}}\hspace{0.02in}V_{\lambda,\mu}(\cdot;\pi)=\sum_{q\in Q}\pi(q)\big(\E_{q}[u(\cdot)]+\mu H(q)\big)$, which is a subjective expected utility preference with belief $\sum_{q\in Q}q\pi(q)$ because the additive term $\mu \sum_{q\in Q}H(q)\pi(q)$ is a constant that is independent of acts. When $Q=\{q\}$, $\pi$ is a Dirac measure on $q$, so $\underset{\lambda\rightarrow0}{\text{\normalfont lim}}\hspace{0.02in}V_{\lambda,\mu}(\cdot;\pi)=\E_{q}[u(\cdot)]+\mu H(q)=\underset{\lambda\rightarrow0}{\text{\normalfont lim}}\hspace{0.02in}v_{\lambda,\mu}(\cdot;q)$, which is the subjective version of \citeauthor{mononen25}'s (\citeyear{mononen25}) entropy-modified expected utility preference. 
\end{proof}

\setstretch{1}

\bibliography{ref}

\begin{thebibliography}{}

\bibitem[Augenblick and Rabin, 2021]{rabin21}
Augenblick, N. and Rabin, M. (2021).
\newblock Belief movement, uncertainty reduction, and rational updating.
\newblock {\em The Quarterly Journal of Economics}, 136(2):933--985.

\bibitem[Barseghyan et~al., 2013]{barseghyan2013nature}
Barseghyan, L., Molinari, F., O'Donoghue, T., and Teitelbaum, J.~C. (2013).
\newblock The nature of risk preferences: Evidence from insurance choices.
\newblock {\em American economic review}, 103(6):2499--2529.

\bibitem[Beck and Sch{\"o}gl, 1993]{esc93}
Beck, C. and Sch{\"o}gl, F. (1993).
\newblock Thermodynamics of chaotic systems.
\newblock {\em Thermodynamics of Chaotic Systems: an Introduction}.

\bibitem[Benjamin, 2019]{ben19}
Benjamin, D.~J. (2019).
\newblock Errors in probabilistic reasoning and judgment biases.
\newblock {\em Handbook of Behavioral Economics: Applications and Foundations 1}, 2:69--186.

\bibitem[Bernheim and Sprenger, 2020]{bernheim20}
Bernheim, B.~D. and Sprenger, C. (2020).
\newblock On the empirical validity of cumulative prospect theory: Experimental evidence of rank-independent probability weighting.
\newblock {\em Econometrica}, 88(4):1363--1409.

\bibitem[Blumer et~al., 1987]{blumer87}
Blumer, A., Ehrenfeucht, A., Haussler, D., and Warmuth, M.~K. (1987).
\newblock Occam's razor.
\newblock {\em Information processing letters}, 24(6):377--380.

\bibitem[Box, 1976]{box76}
Box, G.~E. (1976).
\newblock Science and statistics.
\newblock {\em Journal of the American Statistical Association}, 71(356):791--799.

\bibitem[Caggiano et~al., 2017]{shock17}
Caggiano, G., Castelnuovo, E., and Figueres, J.~M. (2017).
\newblock Economic policy uncertainty and unemployment in the united states: A nonlinear approach.
\newblock {\em Economics Letters}, 151:31--34.

\bibitem[Caggiano et~al., 2014]{shock14}
Caggiano, G., Castelnuovo, E., and Groshenny, N. (2014).
\newblock Uncertainty shocks and unemployment dynamics in us recessions.
\newblock {\em Journal of Monetary Economics}, 67:78--92.

\bibitem[Cerreia-Vioglio et~al., 2025]{miss25}
Cerreia-Vioglio, S., Hansen, L.~P., Maccheroni, F., and Marinacci, M. (2025).
\newblock Making decisions under model misspecification.
\newblock {\em Review of Economic Studies}, page rdaf046.

\bibitem[Chamberlain, 2020]{chamb20}
Chamberlain, G. (2020).
\newblock Robust decision theory and econometrics.
\newblock {\em Annual Review of Economics}, 12(1):239--271.

\bibitem[Chetty et~al., 2009]{chetty09}
Chetty, R., Looney, A., and Kroft, K. (2009).
\newblock Salience and taxation: Theory and evidence.
\newblock {\em American economic review}, 99(4):1145--1177.

\bibitem[de~Clippel et~al., 2025]{ortoleva25}
de~Clippel, G., Moscariello, P., Ortoleva, P., and Rozen, K. (2025).
\newblock Caution in the face of complexity.
\newblock Technical report, Mimeo.

\bibitem[Eden et~al., 2026]{eden26}
Eden, M., Samuelson, L., and Steiner, J. (2026).
\newblock The importance of unlikely events.

\bibitem[Enke and Graeber, 2023]{enke23}
Enke, B. and Graeber, T. (2023).
\newblock Cognitive uncertainty.
\newblock {\em The Quarterly Journal of Economics}, 138(4):2021--2067.

\bibitem[Enke et~al., 2025]{enke25}
Enke, B., Graeber, T., and Oprea, R. (2025).
\newblock Complexity and time.
\newblock {\em Journal of the European Economic Association}, 23(5):1838--1867.

\bibitem[Fiebig et~al., 2010]{fiebig10}
Fiebig, D.~G., Keane, M.~P., Louviere, J., and Wasi, N. (2010).
\newblock The generalized multinomial logit model: Accounting for scale and coefficient heterogeneity.
\newblock {\em Marketing Science}, 29(3):393--421.

\bibitem[Fudenberg et~al., 2015]{fudenberg2015stochastic}
Fudenberg, D., Iijima, R., and Strzalecki, T. (2015).
\newblock Stochastic choice and revealed perturbed utility.
\newblock {\em Econometrica}, 83(6):2371--2409.

\bibitem[Gabaix, 2025]{gabaix25}
Gabaix, X. (2025).
\newblock A theory of complexity aversion.
\newblock {\em Available at SSRN}.

\bibitem[Gilboa and Schmeidler, 1989]{gilboa89}
Gilboa, I. and Schmeidler, D. (1989).
\newblock Maxmin expected utility with non-unique prior.
\newblock {\em Journal of mathematical economics}, 18(2):141--153.

\bibitem[Gr{\"u}nwald and Roos, 2019]{gr19}
Gr{\"u}nwald, P. and Roos, T. (2019).
\newblock Minimum description length revisited.
\newblock {\em International journal of mathematics for industry}, 11(01):1930001.

\bibitem[Handel and Schwartzstein, 2018]{handel18}
Handel, B. and Schwartzstein, J. (2018).
\newblock Frictions or mental gaps: what's behind the information we (don't) use and when do we care?
\newblock {\em Journal of Economic Perspectives}, 32(1):155--178.

\bibitem[Hansen and Sargent, 2001]{hansen01}
Hansen, L.~P. and Sargent, T.~J. (2001).
\newblock Robust control and model uncertainty.
\newblock {\em American Economic Review}, 91(2):60--66.

\bibitem[Hansen and Sargent, 2008]{hansen08}
Hansen, L.~P. and Sargent, T.~J. (2008).
\newblock {\em Robustness}.
\newblock Princeton University Press.

\bibitem[Lanzani, 2025a]{lanzani25}
Lanzani, G. (2025a).
\newblock Dynamic concern for misspecification.
\newblock {\em Econometrica}, 93(4):1333--1370.

\bibitem[Lanzani, 2025b]{lanzani25w}
Lanzani, G. (2025b).
\newblock Supplement to “dynamic concern for misspecification”.
\newblock {\em Econometrica}.

\bibitem[Lipman, 1995]{lipman95}
Lipman, B.~L. (1995).
\newblock Information processing and bounded rationality: a survey.
\newblock {\em Canadian Journal of Economics}, pages 42--67.

\bibitem[Maccheroni et~al., 2006]{mmr06}
Maccheroni, F., Marinacci, M., and Rustichini, A. (2006).
\newblock Ambiguity aversion, robustness, and the variational representation of preferences.
\newblock {\em Econometrica}, 74(6):1447--1498.

\bibitem[Mat{\v{e}}jka and McKay, 2015]{matvejka15}
Mat{\v{e}}jka, F. and McKay, A. (2015).
\newblock Rational inattention to discrete choices: A new foundation for the multinomial logit model.
\newblock {\em American Economic Review}, 105(1):272--298.

\bibitem[Mononen, 2025]{mononen25}
Mononen, L. (2025).
\newblock On preference for simplicity and probability weighting.
\newblock {\em working paper}.

\bibitem[Mullainathan et~al., 2008]{mullainathan2008coarse}
Mullainathan, S., Schwartzstein, J., and Shleifer, A. (2008).
\newblock Coarse thinking and persuasion.
\newblock {\em The Quarterly journal of economics}, 123(2):577--619.

\bibitem[Oprea, 2020]{oprea20}
Oprea, R. (2020).
\newblock What makes a rule complex?
\newblock {\em American economic review}, 110(12):3913--3951.

\bibitem[Puri, 2025]{puri25}
Puri, I. (2025).
\newblock Simplicity and risk.
\newblock {\em The Journal of Finance}, 80(2):1029--1080.

\bibitem[Rissanen, 1978]{mdl78}
Rissanen, J. (1978).
\newblock Modeling by shortest data description.
\newblock {\em Automatica}, 14(5):465--471.

\bibitem[Robson et~al., 2023]{growth22}
Robson, A., Samuelson, L., and Steiner, J. (2023).
\newblock Decision theory and stochastic growth.
\newblock {\em American Economic Review: Insights}, 5(3):357--376.

\bibitem[Shiller, 2017]{shiller17}
Shiller, R.~J. (2017).
\newblock Narrative economics.
\newblock {\em American Economic Review}, 107(4):967--1004.

\bibitem[Sims, 2003]{sims03}
Sims, C.~A. (2003).
\newblock Implications of rational inattention.
\newblock {\em Journal of monetary Economics}, 50(3):665--690.

\bibitem[Strzalecki, 2011]{strz11}
Strzalecki, T. (2011).
\newblock Axiomatic foundations of multiplier preferences.
\newblock {\em Econometrica}, 79(1):47--73.

\bibitem[Strzalecki, 2024]{tomasz24}
Strzalecki, T. (2024).
\newblock Variational bayes and non-bayesian updating.
\newblock {\em arXiv preprint arXiv:2405.08796}.

\bibitem[Sunstein, 2002]{sunstein02}
Sunstein, C.~R. (2002).
\newblock Probability neglect: Emotions, worst cases, and law.
\newblock {\em The Yale Law Journal}, 112:61.

\bibitem[Sunstein, 2003]{sunstein03}
Sunstein, C.~R. (2003).
\newblock Terrorism and probability neglect.
\newblock {\em Journal of Risk and Uncertainty}, 26(2-3):121--136.

\bibitem[Swait and Louviere, 1993]{swait93}
Swait, J. and Louviere, J. (1993).
\newblock The role of the scale parameter in the estimation and comparison of multinomial logit models.
\newblock {\em Journal of Marketing Research}, 30(3):305--314.

\end{thebibliography}
\bibliographystyle{apalike}
\newpage \renewcommand{\thepage}{OA-\arabic{page}} \setcounter{page}{1}

\section{For Online Publication}
\onehalfspacing
\subsection{General complexity functionals}
\label{app:general_complexity}

This appendix shows that most of our qualitative results extend to when the entropy function is generalized to the additive perturbations axiomatized by \citet{fudenberg2015stochastic}.  
Suppose the agent's static objective function is 
 as
 \begin{equation}
 \label{eq:genM_static}
   v^W_{\lambda,\mu}(a;\m{q})
   :=
   \min_{p_a\in\Delta(Y)}
  \left\{
    \E_{p_a}[u(a,y)]
     + \frac{1}{\lambda}R(p_a\Vert q_a)
     + \mu \sum_{y\in Y} W(p_a(y)),
   \right\}
    \end{equation}
    where $W$ is continuous on $[0,1]$, twice continuously differentiable on $(0,1]$, strictly concave, $W(0)=0$, and
$C_W:=\sup_{x\in(0,1]}-xW''(x)<\infty.$\footnote{The additive structure can further be relaxed, but the analysis would  become less transparent because the Hessian of the objective would no longer be diagonal, so off-diagonal terms---which do not have clear economic interpretations---would need to be accounted for.}
The conditions on $W$ imply that more diffuse beliefs are more complex, and that the curvature of the complexity term is not strong enough to overturn the strict convexity generated by the KL term. Entropy corresponds to $W(x)=-x\log x$, for which $C_W=1$. Let $M(\ell):=\sum_{y\in Y} W(\ell(y))$, for every $\ell\in\Delta(Y)$.

\begin{proposition}[Static properties under a general complexity index]
\label{prop:genM_static}
Suppose Assumption \ref{ass:full_support}  holds. Fix $(a,\m{q})\in A\times Q, \lambda>0$, and $\mu\ge 0$ with $\mu<1/(\lambda C_W)$. Then:
\begin{enumerate}
  \item[\textnormal{(i)}] The minimization problem in \eqref{eq:genM_static} has a unique solution, denoted $\hat p^W_{\lambda,\mu}(a;\m{q})$.
  \item[\textnormal{(ii)}] The maps $(\lambda,\mu,\m{q})\mapsto v^W_{\lambda,\mu}(a;\m{q})$ and $(\lambda,\mu,\m{q})\mapsto \hat p^W_{\lambda,\mu}(a;\m{q})$ are continuous on
  $$
    \big\{(\lambda,\mu,\m{q}): \lambda>0,\ \mu<1/(\lambda C_W)\big\}.
  $$
  \item[\textnormal{(iii)}] \emph{For every $\mu \in (0,1/(\lambda C_W))$,}
$\frac{\partial}{\partial \mu}v^W_{\lambda,\mu}(a;\m{q})
=
M\left(\hat p^W_{\lambda,\mu}(a;\m{q})\right).$
\emph{At $\mu=0$, the right derivative exists and satisfies}
$\frac{\partial^+}{\partial \mu}v^W_{\lambda,\mu}(a;\m{q})\Big|_{\mu=0}
=
M\left(\hat p^W_{\lambda,0}(a;\m{q})\right).$
\item[\textnormal{(iv)}] If $0\le \mu_1<\mu_2<1/(\lambda C_W)$, then $M\big(\hat p^W_{\lambda,\mu_2}(a;\m{q})\big)
    \le
    M\big(\hat p^W_{\lambda,\mu_1}(a;\m{q})\big),$
  with strict inequality whenever
  $\hat p^W_{\lambda,\mu_2}(a;\m{q})\neq \hat p^W_{\lambda,\mu_1}(a;\m{q})$.
\end{enumerate}
\end{proposition}

Proposition \ref{prop:genM_static} identifies the main static driver of our qualitative results. The exact Shannon formula is not essential; what matters is that the complexity term enters linearly in the objective and ranks diffuse distortions above concentrated ones.  What is lost relative to entropy is tractability, not the economic mechanism. Notice also that whenever the minimizer is interior, its first-order conditions are
\begin{equation}
\label{eq:genM_foc}
  u(a,y)
  + \frac{1}{\lambda}\Big(1+\log\frac{\hat p^W_{\lambda,\mu}(a;\m{q})(y)}{q_a(y)}\Big)
  + \mu W'\big(\hat p^W_{\lambda,\mu}(a;\m{q})(y)\big)
  + \xi
  =0
  \qquad (y\in Y)
\end{equation}
for some multiplier $\xi\in\R$. For $W(x)=-x\log x$, \eqref{eq:genM_foc} reduces to the Gibbs closed-form formula in Lemma \ref{lem:closed_form}. For a generic $W$, \eqref{eq:genM_foc} is implicit, so the closed form and the representation results are lost, but our qualitative results are preserved.

The remaining theoretical results in the paper use only this more general structure. In particular:
(i) Proposition \ref{prop:dynamic_selection} only needs continuity of the one-period value and best-reply correspondence, which Proposition \ref{prop:genM_static} provides;
(ii) Proposition \ref{prop:correct_spec_limit} only needs boundedness of $M$ on the simplex and the fact that $\mu_t=\bar{\mu}\lambda_t\to0$ under correct specification; and
(iii) Theorem \ref{thm:mu_kills_cycles}, Proposition \ref{prop:multi_thresholds}, Lemma \ref{lem:pairwise_Delta}, Proposition \ref{prop:all_risky_elimination}, and Proposition \ref{thm:welfare_improvement} use entropy only through the envelope identities and the corresponding entropy-gap assumptions. For concreteness, we now state and prove the analog of Theorem \ref{thm:mu_kills_cycles}.

 For any action $a\in A$, $\pi\in\Delta(Q)$, and admissible pair $(\lambda,\bar\mu)$ with $\bar\mu\lambda^2<1/C_W$, define
$M_a(\lambda,\bar\mu,\pi)
  :=
  \sum_{\m{q}\in Q}\pi(\m{q})
  M\big(\hat p^W_{\lambda,\bar\mu\lambda}(a;\m{q})\big).$
If $A=\{r,s\}$, define the risky-safe value gap by
$\Delta^W(\lambda,\bar\mu,\pi)
  :=
  V^W_{\lambda,\bar\mu\lambda}(r;\pi)
  -
  V^W_{\lambda,\bar\mu\lambda}(s;\pi).$

\begin{assumption}[Uniform complexity gap]
\label{ass:genM_gap}
Let $[\underline\lambda,\bar\lambda]$ be as in Assumption \ref{ass:U}. There exist constants $\bar\mu_0\ge 0$ and $M^\ast>0$ such that for all
$(\lambda,\pi)\in[\underline\lambda,\bar\lambda]\times\Delta(Q)$ and all admissible
$\bar\mu\ge\bar\mu_0$ satisfying $\bar\mu\lambda^2<1/C_W$,
$\Delta^W(\lambda,\bar\mu,\pi)\ge 0
  \Longrightarrow
  M_s(\lambda,\bar\mu,\pi)-M_r(\lambda,\bar\mu,\pi)\ge M^\ast.$
\end{assumption}

Assumption \ref{ass:genM_gap} is the analog of Assumption \ref{ass:uniform_entropy_gap} with $M$ in place of $H$. It says that whenever the risky arm is still competitive, the safe arm can only be supported by a uniformly more diffuse and therefore more complex worst-case narrative.

\begin{lemma}
\label{lem:genM_Delta_derivative}
Suppose Assumptions \ref{ass:full_support} and \ref{ass:genM_gap} hold. Then, for any two actions $i,j\in A$, any posterior $\pi\in\Delta(Q)$, and any admissible $(\lambda,\bar\mu)$ with $\bar\mu\lambda^2<1/C_W$,
$$
  \frac{\partial}{\partial\bar\mu}
  \Big(
    V^W_{\lambda,\bar\mu\lambda}(i;\pi)
    -
    V^W_{\lambda,\bar\mu\lambda}(j;\pi)
  \Big)
  =
  \lambda\big(M_i(\lambda,\bar\mu,\pi)-M_j(\lambda,\bar\mu,\pi)\big).
$$
\end{lemma}

The best-reply correspondence is now 
$\operatorname{BR}^{\bar{\mu},W}_{\lambda_t}(\pi_t)
  :=
  \underset{a\in A}{\arg\max} \hspace{0.02in}V^{W}_{\lambda_t,\bar{\mu}\lambda_t}(a;\pi_t)$, for every $t$, and  the associated optimal policy and $c$-robust equilibrium based on $W$ are defined as in Section \ref{sec:dynamic} by substituting   $\operatorname{BR}^{\bar{\mu}}_{\lambda_t}$ for $\operatorname{BR}^{\bar{\mu},W}_{\lambda_t}$. For brevity, we also interpret Assumption \ref{ass:U} as being based on $W$.

\begin{theorem}[General complexity indices eliminate $\lambda$-cycles]
\label{thm:genM_kills_cycles}
Consider the dynamic model with the one-period criterion \eqref{eq:genM_static} and the same normalization rule $\lambda_t=\Lambda(h_t)$ as in \eqref{eq:rho_def}. Suppose $A=\{r,s\}$ and Assumptions \ref{ass:U} and \ref{ass:genM_gap} hold. Then, there exists $\bar\mu^\ast>0$ such that for every admissible $\bar\mu\ge \bar\mu^\ast$ with $\bar\mu\bar\lambda^2<1/C_W$ the following holds:
\begin{enumerate}
  \item[\textnormal{(i)}] Every mixed $c$-robust equilibrium based on $W$ satisfies $\alpha(s)=1$.
  \item[\textnormal{(ii)}] Consequently, every $\Lambda$-limit frequency of any optimal policy based on $W$ is pure and equals the safe action.
\end{enumerate}
\end{theorem}

The same logic yields the other  qualitative results in the main text. Proposition \ref{prop:multi_thresholds} and Proposition \ref{prop:all_risky_elimination} use entropy only through the pairwise derivative identity in Lemma \ref{lem:pairwise_Delta}; Lemma \ref{lem:genM_Delta_derivative} gives the corresponding identity for $W$, so those proofs go through after replacing the relevant entropy-gap assumptions by their $W$-analogues. Likewise, Proposition \ref{thm:welfare_improvement} uses only the eventual selection of the safe arm, so it continues to hold once Theorem \ref{thm:genM_kills_cycles} replaces Theorem \ref{thm:mu_kills_cycles}.

The economic content of this extension is therefore simple. The paper does not need Shannon entropy in order to generate conservatism, selection, or the elimination of $\lambda$-cycles. It only needs a complexity index that assigns a larger value to diffuse worst-case narratives than to concentrated ones, enters linearly in the objective, and is regular enough that the KL term still pins down a unique minimizer. Thus, Shannon entropy is the tractable benchmark inside this broader class because it turns \eqref{eq:genM_foc} into a closed-form Gibbs distortion and therefore delivers the cleaner representation and application formulas in the main text.

\subsubsection{Proofs}

\begin{proof}[Proof of Proposition \ref{prop:genM_static}]
Let $F_{\lambda,\mu}(p_a;a,\m{q})
  :=
  \sum_{y\in Y} u(a,y)p_a(y)
  + \frac{1}{\lambda}R(p_a\Vert q_a)
  + \mu \sum_{y\in Y}W(p_a(y)).$
For $x\in[0,1]$, define $g_{\lambda,\mu}(x):=(1/\lambda)x\log x+\mu W(x)$, with the convention $0\log 0:=0$. For every $x\in(0,1]$, $g''_{\lambda,\mu}(x)
  =
  \frac{1}{\lambda x}+\mu W''(x)
  \ge
  \frac{1-\lambda\mu C_W}{\lambda x}
  >0.$
Hence, $g_{\lambda,\mu}$ is strictly convex on $[0,1]$. Since
$F_{\lambda,\mu}(p_a;a,\m{q})
  =
  \sum_{y\in Y} g_{\lambda,\mu}(p_a(y))
  + \sum_{y\in Y}\big(u(a,y)-\frac{1}{\lambda}\log q_a(y)\big)p_{a}(y),$
then $F_{\lambda,\mu}(\cdot;a,\m{q})$ is strictly convex on the compact convex set $\Delta(Y)$. Therefore, it has a unique minimizer $\hat p^W_{\lambda,\mu}(a;\m{q})$, which proves (i).

For (ii), the objective $F_{\lambda,\mu}(p_a;a,\m{q})$ is jointly continuous in $(p_a,\lambda,\mu,\m{q})$ on the admissible domain, and the feasible set $\Delta(Y)$ is compact and independent of parameters. Berge's maximum theorem yields continuity of the value function, and uniqueness of the minimizer upgrades upper hemicontinuity of the argmin correspondence to continuity of $(\lambda,\mu,\m{q})\mapsto \hat p^W_{\lambda,\mu}(a;\m{q})$.

For (iii), fix admissible $(\lambda,\mu,\m{q})$ and write $p_{\mu,\m{q}}:=\hat p^W_{\lambda,\mu}(a;\m{q})$. For $h>0$ small enough that $\mu+h<1/(\lambda C_W)$, optimality of $p_{\mu+h}$ for $F_{\lambda,\mu+h}$ and of $p_\mu$ for $F_{\lambda,\mu}$ gives
$$
  F_{\lambda,\mu+h}(p_{\mu+h};a,\m{q})
  \le
  F_{\lambda,\mu+h}(p_\mu;a,\m{q})
  =
  F_{\lambda,\mu}(p_\mu;a,\m{q})+hM(p_\mu),
$$
and $F_{\lambda,\mu}(p_\mu;a,\m{q})
  \le
  F_{\lambda,\mu}(p_{\mu+h};a,\m{q})
  =
  F_{\lambda,\mu+h}(p_{\mu+h};a,\m{q})-hM(p_{\mu+h}).$
Subtracting shows
$M(p_{\mu+h})
  \le
  \frac{v^W_{\lambda,\mu+h}(a;\m{q})-v^W_{\lambda,\mu}(a;\m{q})}{h}
  \le
  M(p_\mu)$, where we recall that $M(\ell)=\sum_{y\in Y}W(\ell(y))$ for all $\ell\in \Delta(Y)$.
Letting $h \downarrow 0$ and using continuity of $\mu \mapsto p_\mu$ and of $M$ yields $\frac{\partial^+}{\partial \mu}v^W_{\lambda,\mu}(a;\m{q})
=
M\big(\hat p^W_{\lambda,\mu}(a;\m{q})\big).$
If $\mu \in (0,1/(\lambda C_W))$, the same argument with $h<0$ and $\mu+h\ge 0$ gives the
left derivative, so the two-sided derivative exists and equals
$M\big(\hat p^W_{\lambda,\mu}(a;\m{q})\big)$.

For (iv), let $p_i:=\hat p^W_{\lambda,\mu_i}(a;\m{q})$ for $i=1,2$. Optimality implies
$F_{\lambda,\mu_1}(p_1;a,\m{q})\le F_{\lambda,\mu_1}(p_2;a,\m{q})$ and $F_{\lambda,\mu_2}(p_2;a,\m{q})\le F_{\lambda,\mu_2}(p_1;a,\m{q}).$
Adding the two inequalities and cancelling the common terms yields $(\mu_2-\mu_1)\big(M(p_2)-M(p_1)\big)\le 0,$
so $M(p_2)\le M(p_1)$. If equality held with $p_1\neq p_2$, then both optimality inequalities above would have to bind, so both $p_1$ and $p_2$ would minimize the strictly convex function $F_{\lambda,\mu_1}(\cdot;a,\m{q})$, a contradiction. Hence equality is possible only when $p_1=p_2$.
\end{proof}

\begin{proof}[Proof of Lemma \ref{lem:genM_Delta_derivative}]
Fix $a\in A$. By Proposition \ref{prop:genM_static}, $\frac{\partial}{\partial\mu} v^W_{\lambda,\mu}(a;\m{q})
  = M\big(\hat p^W_{\lambda,\mu}(a;\m{q})\big).$
Setting $\mu=\bar\mu\lambda$ and using the chain rule gives
$\frac{\partial}{\partial\bar\mu} v^W_{\lambda,\bar\mu\lambda}(a;\m{q})
  =
  \lambda M\big(\hat p^W_{\lambda,\bar\mu\lambda}(a;\m{q})\big).$
Let $V^W_{\lambda,\mu}(a;\pi):=\sum_{\m{q}\in Q}\pi(\m{q}) v^W_{\lambda,\mu}(a;\m{q})$, so after summing the derivative over $\m{q}\in Q$ with weights $\pi(\m{q})$, we get
$\frac{\partial}{\partial\bar\mu}V^W_{\lambda,\bar\mu\lambda}(a;\pi)
  =
  \lambda M_a(\lambda,\bar\mu,\pi),$
and subtracting the identities for $i$ and $j$ proves the claim.
\end{proof}

\begin{proof}[Proof of Theorem \ref{thm:genM_kills_cycles}]
Let $\mathcal D:=[\underline\lambda,\bar\lambda]\times\Delta(Q),$ $\hat m:=\underline\lambda M^\ast>0,$
and define the constant
$K_0:=\sup_{(\lambda,\pi)\in\mathcal D}\max\{\Delta^W(\lambda,\bar\mu_0,\pi),0\}.$
Since $\Delta^W$ is continuous by Proposition \ref{prop:genM_static} and $\mathcal D$ is compact, $K_0<\infty$. Set
$\varepsilon
  :=
  \frac12\big(\frac{1}{C_W\bar\lambda^2}-\bar\mu_0-\frac{K_0}{\hat m}\big)_+$ and $\bar\mu^\ast
  :=
  \bar\mu_0+\frac{K_0}{\hat m}+\varepsilon.$
If $\varepsilon=0$, then there is no admissible $\bar\mu$ satisfying both $\bar\mu\ge\bar\mu^\ast$ and $\bar\mu\bar\lambda^2<1/C_W$, so the theorem holds vacuously as Theorem \ref{thm:mu_kills_cycles}. Hence, it suffices to consider the case $\varepsilon>0$, in which $\bar\mu^\ast<1/(C_W\bar\lambda^2)$.

Fix $(\lambda,\pi)\in\mathcal D$ and define $f(\bar\mu):=\Delta^W(\lambda,\bar\mu,\pi)$
on the admissible interval $\{\bar\mu:\bar\mu\lambda^2<1/C_W\}$. By Proposition \ref{prop:genM_static} and Lemma \ref{lem:genM_Delta_derivative}, $f$ is continuous on the admissible interval and continuously differentiable at every interior point $\bar\mu>0$. If $\bar\mu_0=0$, all integrations below should be read over $[\varepsilon,\bar\mu_1]$ for arbitrary $\varepsilon>0$, and then $\varepsilon\downarrow0$; continuity of $f$ and the definition
$K_0:=\sup_{(\lambda,\pi)\in D}\max\{\Delta^W(\lambda,\bar\mu_0,\pi),0\}$
give the same bound. If $\bar\mu_0>0$, the argument is unchanged. Assumption \ref{ass:genM_gap} and Lemma \ref{lem:genM_Delta_derivative} imply that whenever $\bar\mu\ge\bar\mu_0$ and $f(\bar\mu)\ge 0$,
\begin{equation}
\label{eq:genM_negative_slope}
  f'(\bar\mu)
  =
  \lambda\big(M_r(\lambda,\bar\mu,\pi)-M_s(\lambda,\bar\mu,\pi)\big)
  \le
  -\underline\lambda M^\ast
  =
  -\hat m.
\end{equation}
We claim that for every admissible $\bar\mu>\bar\mu_0+K_0/\hat m$, one has $f(\bar\mu)<0$.

\par Suppose first that $f(\bar\mu_0)\le 0$. We claim that then
$f(\bar\mu)<0$
  for every admissible $\bar\mu>\bar\mu_0.$
Suppose not. Then there exists an admissible $\bar\mu_1>\bar\mu_0$ such that
$f(\bar\mu_1)\ge 0$.

If $f(\bar\mu_1)>0$, then by continuity of $f$, the set
$Z:=\{\bar\mu\in[\bar\mu_0,\bar\mu_1]:f(\bar\mu)=0\}$
is nonempty. Let $\hat\mu:=\max Z$. Then $f(\bar\mu)>0$ for every
$\bar\mu\in(\hat\mu,\bar\mu_1].$
Hence, by the derivative bound established above,
$f'(\bar\mu)\le -\hat m$
  for all $\bar\mu\in(\hat\mu,\bar\mu_1].$
Integrating from $\hat\mu$ to $\bar\mu_1$ gives
$f(\bar\mu_1)-f(\hat\mu)
  =
  \int_{\hat\mu}^{\bar\mu_1}f'(t)\hspace{0.02in}dt
  \le
  -\hat m(\bar\mu_1-\hat\mu)<0.$
Since $f(\hat\mu)=0$, this yields $f(\bar\mu_1)<0$, a contradiction.

If instead $f(\bar\mu_1)=0$, then \eqref{eq:genM_negative_slope} implies
$f'(\bar\mu_1)\le -\hat m<0.$
Because $f$ is $C^1$ on the admissible interval, there exists $\delta>0$ such that
$[\bar\mu_1-\delta,\bar\mu_1]$ is admissible and
$f'(t)\le -\hat m/2$ $\text{for all }t\in[\bar\mu_1-\delta,\bar\mu_1].$
Therefore,
$f(\bar\mu_1)-f(\bar\mu_1-\delta)
  =
  \int_{\bar\mu_1-\delta}^{\bar\mu_1}f'(t)\hspace{0.02in}dt
  \le
  -(\hat m/2)\delta<0,$
so $f(\bar\mu_1-\delta)>0$. Since $f(\bar\mu_0)\le 0$, continuity implies that the set
$Z:=\{\bar\mu\in[\bar\mu_0,\bar\mu_1-\delta]:f(\bar\mu)=0\}$
is nonempty. Let $\hat\mu:=\max Z$. Then $f(\bar\mu)>0$ for every
$\bar\mu\in(\hat\mu,\bar\mu_1-\delta].$
Hence, by the derivative bound established above,
$f'(\bar\mu)\le -\hat m$
  for all $\bar\mu\in(\hat\mu,\bar\mu_1-\delta].$
Integrating from $\hat\mu$ to $\bar\mu_1-\delta$ gives
$f(\bar\mu_1-\delta)-f(\hat\mu)
  =
  \int_{\hat\mu}^{\bar\mu_1-\delta}f'(t)\hspace{0.02in}dt
  \le
  -\hat m(\bar\mu_1-\delta-\hat\mu)<0.$
Since $f(\hat\mu)=0$, this yields $f(\bar\mu_1-\delta)<0$, contradicting
$f(\bar\mu_1-\delta)>0$. Therefore, $f(\bar\mu)<0$
  for every admissible $\bar\mu>\bar\mu_0$
whenever $f(\bar\mu_0)\le 0$.

\par Next suppose that $f(\bar\mu_0)>0$, and define
$\bar\mu_1:=\bar\mu_0+\frac{f(\bar\mu_0)}{\hat m}.$
If $\bar\mu_1$ is not admissible, then there is no admissible $\bar\mu>\bar\mu_1$, so there is nothing to prove beyond the admissible range. Suppose therefore that $\bar\mu_1$ is admissible. We claim that
$f(\bar\mu_1)\le 0.$
Suppose not, so that $f(\bar\mu_1)>0$. If there exists $\tilde\mu\in[\bar\mu_0,\bar\mu_1]$ such that $f(\tilde\mu)=0$, let
$\hat\mu:=\max\{\bar\mu\in[\tilde\mu,\bar\mu_1]:f(\bar\mu)=0\}.$
Then $f(\bar\mu)>0$ for every $\bar\mu\in(\hat\mu,\bar\mu_1]$, hence
$f'(\bar\mu)\le -\hat m$
  for all $\bar\mu\in(\hat\mu,\bar\mu_1].$
Integrating from $\hat\mu$ to $\bar\mu_1$ yields
$f(\bar\mu_1)-f(\hat\mu)
  =
  \int_{\hat\mu}^{\bar\mu_1}f'(t)\hspace{0.02in}dt
  \le
  -\hat m(\bar\mu_1-\hat\mu)<0,$
contradicting $f(\hat\mu)=0$ and $f(\bar\mu_1)>0$.

Hence, there is no zero of $f$ in $[\bar\mu_0,\bar\mu_1]$. Since $f(\bar\mu_0)>0$ and
$f(\bar\mu_1)>0$, continuity implies that
$f(\bar\mu)>0$
  for all $\bar\mu\in[\bar\mu_0,\bar\mu_1].$
Therefore, the derivative bound applies throughout that interval:
$f'(\bar\mu)\le -\hat m$
  for all $\bar\mu\in[\bar\mu_0,\bar\mu_1].$
Integrating from $\bar\mu_0$ to $\bar\mu_1$ gives
$$
  f(\bar\mu_1)
  =
  f(\bar\mu_0)+\int_{\bar\mu_0}^{\bar\mu_1}f'(t)\hspace{0.02in}dt
  \le
  f(\bar\mu_0)-\hat m(\bar\mu_1-\bar\mu_0)
  =
  0,
$$
a contradiction. Thus, $f(\bar\mu_1)\le 0$. Since
$\bar\mu_1
  =
  \bar\mu_0+\frac{f(\bar\mu_0)}{\hat m}
  \le
  \bar\mu_0+\frac{K_0}{\hat m},$
the first part of the argument, applied from the starting point $\bar\mu_1$, implies that
$f(\bar\mu)<0$
  for all admissible $\bar\mu>\bar\mu_0+\frac{K_0}{\hat m}.$ We have thus proved that
\begin{equation}
\label{eq:genM_uniform_negativity}
  \Delta^W(\lambda,\bar\mu,\pi)<0
  \qquad\text{for all }(\lambda,\pi)\in\mathcal D
  \text{ and all admissible }\bar\mu\ge\bar\mu^\ast.
\end{equation}

To prove part (i), fix admissible $\bar\mu\ge\bar\mu^\ast$ and let $(\alpha,\eta,\tau)$ be a mixed $c$-robust equilibrium for the dynamic criterion based on $M$. If $\alpha(r)>0$, Assumption \ref{ass:U} implies $\tau\in[\underline\lambda,\bar\lambda]$. Since $(\alpha,\eta,\tau)$ is a mixed $c$-robust equilibrium, any action in the support of $\alpha$ must be a best reply at $(\tau,\eta)$, so in particular $r$ must be a best reply. But \eqref{eq:genM_uniform_negativity} gives
$\Delta^W(\tau,\bar\mu,\eta)<0,$
which means
$V^W_{\tau,\bar\mu\tau}(r;\eta)
  <
  V^W_{\tau,\bar\mu\tau}(s;\eta).$
Thus, $r$ is not a best reply, a contradiction. Hence, $\alpha(r)=0$, so $\alpha(s)=1$.

For part (ii), the proof of Proposition \ref{prop:dynamic_selection} uses only continuity of the one-period value and of the best-reply correspondence. Proposition \ref{prop:genM_static} provides exactly those properties for the dynamic criterion based on $M$, so the same argument applies verbatim: every $\Lambda$-limit frequency is the action component of a mixed $c$-robust equilibrium for the criterion based on $M$. Part (i) then implies that every such limit frequency is pure safe.
\end{proof}

\subsection{ Taxation}\label{sec:lanzani25_taxation}
In this application, the safe action makes tax uncertainty payoff-irrelevant, so its worst-case belief remains diffuse across tax states. By contrast, positive labor loads linearly on the tax-rate shock, so the worst-case belief concentrates on the highest-tax realization. This yields a uniform entropy gap.

We start from \citeauthor{lanzani25}'s (\citeyear{lanzani25}) Example 1, restrict attention to a two-action subproblem, and replace Lanzani's continuous tax shock with a symmetric finite-support version to fit our framework.

\subsubsection{Setup and primitives}

The agent chooses a target gross income $a\in A$ at cost
$C(a)$. Taxes are paid at an average tax rate $\tau(a)+\varepsilon$, where $\tau$ is strictly
increasing because of progressivity, and the payoff is
$u(a,t)=a-t-C(a).$
The agent's misspecified structured model is the random-coefficient specification
$t=(\theta+\varepsilon)a,$
under which the marginal and average tax rates are equal. We assume 
$A=\{s,r\}=\{0,\bar a\},$ for $\bar a>0.$
Thus, $s$ is the safe action of zero labor and $r$ is the risky action of positive labor. Let the
shock take the symmetric three-point form
$\varepsilon\in\{-\zeta,0,\zeta\},$ for $\zeta>0,$
with probability $1/3$ on each realization. Let $\Theta$ be any
finite subset of Lanzani's parametric model class and let
$Q=\{q^\theta:\theta\in\Theta\}.\footnote{Since $\Theta$ is finite, the relevant set of average-tax-rate realizations does not have full support. However, for each $\theta\in\Theta$, the model $q^\theta$ assigns probability $1/3$ to $m_L^\theta,m_M^\theta,m_H^\theta$ and zero elsewhere, so the closed-form distortion formula is applied on the support of $q^\theta$.}$
For each $\theta\in\Theta$, the structured model $q^\theta$ induces the three average-tax-rate
realizations
$$
m_L^\theta:=\theta-\zeta,\qquad m_M^\theta:=\theta,\qquad m_H^\theta:=\theta+\zeta,
$$
each with probability $1/3$. Since the tax bill is $t=ma$, the payoff function is 
$u(a,m)=a(1-m)-C(a).$
Notice that this is only a reparameterization of outcomes from tax bills to average tax rates. The true tax
schedule remains $\tau(a)+\varepsilon$, as in Lanzani's application, but Assumption
\ref{ass:uniform_entropy_gap} concerns only the static entropy ranking under the structured
models $q^\theta$, so the true DGP does not enter the proof below.
\subsubsection{Result}
For any posterior $\pi\in\Delta(Q)$ and any admissible $(\lambda,\bar\mu)$, define
$$
H_a(\lambda,\bar\mu,\pi)
:=
\sum_{\theta\in\Theta}\pi(q^\theta)H \big(\hat p_{\lambda,\bar\mu\lambda}(a;q^\theta)\big),
\qquad (a\in\{s,r\}),
$$
where $\hat p_{\lambda,\bar\mu\lambda}(a;q^\theta)$ is the worst-case distortion from
Lemma \ref{lem:closed_form}.\footnote{Because $R(p\Vert q^\theta)=+\infty$ whenever $p\not\ll q^\theta$, the inner minimization can be restricted to distributions supported on $\text{supp}(q^\theta)$. Hence the closed-form expression from Lemma \ref{lem:closed_form} applies on that support.}

\begin{proposition}\label{prop:lanzani25_tax_entropy_gap}
For any interval $[\underline\lambda,\bar\lambda]\subset(0,\infty)$, the taxation specialization
above satisfies Assumption \ref{ass:uniform_entropy_gap}. In particular, it holds with
$\bar\mu_0=0$ and
$$
H^\ast
:=
\log 3-
H \Big(
\frac{e^{-t_0}}{1+2\cosh t_0},
\frac{1}{1+2\cosh t_0},
\frac{e^{t_0}}{1+2\cosh t_0}
\Big)
>0,
\qquad
(t_0:=\bar a\zeta\underline\lambda).
$$
\end{proposition}

Proposition \ref{prop:lanzani25_tax_entropy_gap} gives an economic interpretation of
Assumption \ref{ass:uniform_entropy_gap} in this application. Here zero labor is safe because tax uncertainty is payoff-irrelevant when labor supply is zero, so the
worst-case narrative remains diffuse. By contrast, positive labor loads linearly on the tax-rate
shock, so the worst-case narrative concentrates on the highest-tax realization. Thus in this
environment, Assumption \ref{ass:uniform_entropy_gap} is simply the
statement that the risky action admits a simpler downside narrative than the safe one. This is
precisely the force emphasized by Lanzani's Example 1: higher effort is more exposed to
uncertainty in the marginal tax rate.

\begin{proof}[Proof of Proposition \ref{prop:lanzani25_tax_entropy_gap}]
Fix any $(\lambda,\pi)\in[\underline\lambda,\bar\lambda]\times\Delta(Q)$ and any admissible
$\bar\mu\ge 0$ with $\bar\mu\lambda^2<1$. Set
$\mu:=\bar\mu\lambda,$ $\kappa:=\frac{1}{\lambda}-\mu=\frac{1-\bar\mu\lambda^2}{\lambda}>0.$ We first compute the entropy term model by model. Fix any $\theta\in\Theta$.

\medskip

\noindent\emph{Safe action.}
Since $s=0$, payoff satisfies
$u(s,m)=u(0,m)=-C(0)$
for every $m\in\{m_L^\theta,m_M^\theta,m_H^\theta\}$. Hence, by Lemma
\ref{lem:closed_form}, the payoff tilt cancels and
$\hat p_{\lambda,\mu}(s;q^\theta)
=
\left(\frac13,\frac13,\frac13\right).$
Therefore, \begin{equation}\label{eq:lanzani25_tax_Hs}
H \big(\hat p_{\lambda,\mu}(s;q^\theta)\big)=\log 3.
\end{equation}

\medskip

\noindent\emph{Risky action.}
Let $r=\bar a$. By Lemma \ref{lem:closed_form},
$\hat p_{\lambda,\mu}(r;q^\theta)(m)
\propto
e^{-u(r,m)/\kappa}q^\theta(m)^\beta,$ $\beta:=\frac{1}{1-\lambda\mu}>0.$
Since $q^\theta$ is uniform on $\{m_L^\theta,m_M^\theta,m_H^\theta\}$, the factor
$q^\theta(m)^\beta$ is constant across states. Also,
$u(r,m)=\bar a(1-m)-C(\bar a).$
Hence,
$\hat p_{\lambda,\mu}(r;q^\theta)(m)\propto e^{\bar a m/\kappa}.$
Substituting
$$
m_L^\theta=\theta-\zeta,\qquad m_M^\theta=\theta,\qquad m_H^\theta=\theta+\zeta,
$$
we get $\hat p_{\lambda,\mu}(r;q^\theta)
\propto
\big(e^{\bar a(\theta-\zeta)/\kappa},
e^{\bar a\theta/\kappa},
e^{\bar a(\theta+\zeta)/\kappa}\big).$
The common factor $e^{\bar a\theta/\kappa}$ cancels, so the distortion is independent of
$\theta$:
\begin{equation}\label{eq:lanzani25_tax_risky_distortion}
\hat p_{\lambda,\mu}(r;q^\theta)
=
\left(
\frac{e^{-t}}{1+2\cosh t},
\frac{1}{1+2\cosh t},
\frac{e^{t}}{1+2\cosh t}
\right),
\qquad
(t:=\frac{\bar a\zeta}{\kappa}>0).
\end{equation}

Define $h(t):=
H \big(
\frac{e^{-t}}{1+2\cosh t},
\frac{1}{1+2\cosh t},
\frac{e^{t}}{1+2\cosh t}
\big),$ for $ t\ge 0.$
Let $X\in\{-1,0,1\}$ have distribution proportional to $e^{tx}$, denoted $G_t$. Then, $h(t)=\log(1+2\cosh t)-t\E_{G_t}[X].$
Differentiating,
$h'(t)
=
\frac{2\sinh t}{1+2\cosh t}
-\E_{G_t}[X]
-t\frac{d}{dt}\E_{G_t}[X]
=
-t\text{Var}_{G_t}(X)
<
0$ $\text{for all }t>0.$
Thus, $h$ is strictly decreasing on $(0,\infty)$. Moreover,
$\frac{1}{\kappa}
=
\frac{\lambda}{1-\bar\mu\lambda^2}
\ge \lambda
\ge \underline\lambda,$
so
$t=\frac{\bar a\zeta}{\kappa}\ge \bar a\zeta\underline\lambda=t_0>0.$
Since $h$ is strictly decreasing,
\begin{equation}\label{eq:lanzani25_tax_Hr_bound}
H \big(\hat p_{\lambda,\mu}(r;q^\theta)\big)
=
h(t)
\le h(t_0)
<
h(0)
=
\log 3.
\end{equation}

Since \eqref{eq:lanzani25_tax_Hs} and \eqref{eq:lanzani25_tax_Hr_bound} hold for every $\theta\in\Theta$ and the right-hand sides are independent of $\theta$, averaging with respect to $\pi$ yields
$H_s(\lambda,\bar\mu,\pi)=\log 3$ and $
H_r(\lambda,\bar\mu,\pi)\le h(t_0),$ so $H_s(\lambda,\bar\mu,\pi)-H_r(\lambda,\bar\mu,\pi)
\ge
\log 3-h(t_0)
=H^\ast>0$ for every $(\lambda,\pi)\in[\underline\lambda,\bar\lambda]\times\Delta(Q)$ and every admissible $\bar\mu\ge 0$. Hence the implication in Assumption \ref{ass:uniform_entropy_gap} holds a fortiori, so  Assumption \ref{ass:uniform_entropy_gap} holds with $\bar\mu_0=0$ and the constant $H^\ast$ defined above.
\end{proof}

\subsection{Portfolio choice}\label{sec:chamb20_entropy_gap}

This appendix verifies Assumption \ref{ass:uniform_entropy_gap} in a simple two-plan version of the portfolio environment in \citet[][Sections 2.1 and 2.3]{chamb20}. The key point is that the safe plan delivers the same payoff in every return history, so its worst-case distortion remains relatively diffuse, while the risky plan has a unique disaster history, so its worst-case distortion can concentrate there. This produces a uniform entropy gap.

There are two dates, $t=0,1$. The gross risk-free return is $r_f>0$, and the risky gross return in each period is either $R^H$ or $R^L$, where $0<R^L<r_f<R^H$. A return history is $y=(R_0,R_1)\in Y$, where $Y=\{HH,HL,LH,LL\}$ with $HH:=(R^H,R^H)$, $HL:=(R^H,R^L)$, $LH:=(R^L,R^H)$, and $LL:=(R^L,R^L)$. Initial wealth is $w_0>0$, and utility over terminal wealth is CRRA: $U(w)=\frac{w^{1-\gamma}}{1-\gamma}$ if $\gamma\neq 1$ and $U(w)=\log w$ if $\gamma=1$, with $\gamma>0$.

As in \citet{chamb20}, the reference model is a single distribution $q$ over $Y$. We restrict attention to the symmetric exchangeable case in which $q(HH)=q(LL)=a_0$ and $q(HL)=q(LH)=b_0$ for some $a_0>b_0>0$ with $2a_0+2b_0=1$. This is the symmetric version of Chamberlain's exchangeable benchmark and captures the idea that the reference model has no directional bias between up and down moves, while still making the two persistent histories $HH$ and $LL$ more likely than the mixed histories $HL$ and $LH$.

The action set is $A=\{s,r\}$. Action $s$ is the safe plan that invests only in the risk-free asset in both periods, and action $r$ is the risky plan that invests only in the risky asset in both periods. Terminal wealth is therefore $w_s(y)=w_0r_f^2$ for every $y\in Y$, while $w_r(HH)=w_0(R^H)^2$, $w_r(HL)=w_r(LH)=w_0R^HR^L$, and $w_r(LL)=w_0(R^L)^2$. We identify payoffs with $u(a,y):=U(w_a(y))$. Since $R^H>R^L$ and $U$ is strictly increasing, the risky plan has a unique worst history, namely $LL$, and its smallest utility gap above that history is
$
\delta:=u(r,HL)-u(r,LL)>0.
$

For any admissible pair $(\lambda,\bar\mu)$ with $\lambda>0$ and $\bar\mu\lambda^2<1$, write $\mu=\bar\mu\lambda$, $\kappa=\frac{1}{\lambda}-\mu=\frac{1-\bar\mu\lambda^2}{\lambda}>0$, and $\beta=\frac{1}{1-\lambda\mu}=\frac{1}{1-\bar\mu\lambda^2}>0$. Since $Q=\{\m{q}\}$ is a singleton, posterior beliefs play no role, and we write $H_a(\lambda,\bar\mu):=H(\hat p_{\lambda,\bar\mu\lambda}(a;\m{q}))$ for the entropy of the worst-case distortion for action $a\in\{s,r\}$.

To state the result, fix an interval $[\underline\lambda,\bar\lambda]\subset(0,\infty)$. Choose any $\bar p\in(1/4,1)$ such that $\bar H(\bar p)<\log 2$, where $\bar H(p):=-p\log p-(1-p)\log\!\big(\frac{1-p}{3}\big)$. Such a choice is always possible because $\bar H$ is continuous and $\bar H(p)\to 0$ as $p\uparrow 1$. Next define $\bar\kappa:=\frac{\delta}{\log(\frac{3\bar p}{1-\bar p})}$, set
$$
\bar\mu_0:=\max\Big\{0,\frac{1/\underline\lambda-\bar\kappa}{\underline\lambda}\Big\},
\qquad
H^\ast:=\log 2-\bar H(\bar p)>0,
$$
and assume $\bar\mu_0\bar\lambda^2<1$.

\begin{proposition}\label{prop:chamb20_entropy_gap}
In the portfolio environment above, Assumption \ref{ass:uniform_entropy_gap} holds on $[\underline\lambda,\bar\lambda]$ with the constants $\bar\mu_0$ and $H^\ast$ defined above.
\end{proposition}

Proposition \ref{prop:chamb20_entropy_gap} gives a transparent interpretation of Assumption \ref{ass:uniform_entropy_gap}. The safe plan has constant payoffs across histories, so even after the power transformation induced by complexity aversion its worst-case belief cannot collapse onto a single history; at most it can concentrate on one of the two symmetric pairs $\{HH,LL\}$ or $\{HL,LH\}$. By contrast, the risky plan has a unique disaster history $LL$, and once $\bar\mu$ is large enough its worst-case distortion assigns a uniformly large probability to that history. The entropy of the risky distortion is therefore uniformly below that of the safe distortion.

\begin{proof}[Proof of Proposition \ref{prop:chamb20_entropy_gap}]
Fix any $\lambda\in[\underline\lambda,\bar\lambda]$ and any admissible $\bar\mu\ge \bar\mu_0$ with $\bar\mu\lambda^2<1$. Set $\mu=\bar\mu\lambda$, $\kappa=\frac{1}{\lambda}-\mu>0$, and $\beta=\frac{1}{1-\lambda\mu}>0$.

For the safe plan, $u(s,y)=U(w_0r_f^2)$ is constant in $y$, so Lemma \ref{lem:closed_form} gives $\hat p_{\lambda,\mu}(s;q)(y)\propto q(y)^\beta$. Because $q(HH)=q(LL)$ and $q(HL)=q(LH)$, the distorted belief has the form $\hat p_{\lambda,\mu}(s;q)=(a,b,b,a)$ for some $a,b\ge 0$ with $2a+2b=1$. The set of such distributions is the line segment joining $(1/2,0,0,1/2)$ and $(0,1/2,1/2,0)$. Since Shannon entropy is concave, its minimum over this segment is attained at an endpoint, and each endpoint has entropy $\log 2$. Hence $H_s(\lambda,\bar\mu)\ge \log 2$ for every admissible $(\lambda,\bar\mu)$.

For the risky plan, Lemma \ref{lem:closed_form} gives $\hat p_{\lambda,\mu}(r;q)(y)\propto e^{-u(r,y)/\kappa}q(y)^\beta$. Let $w(y):=e^{-u(r,y)/\kappa}q(y)^\beta$. Since $LL$ is the unique worst history and $\delta=u(r,HL)-u(r,LL)>0$, we have $u(r,y)\ge u(r,LL)+\delta$ for every $y\neq LL$. Also, $q(y)\le q(LL)$ for every $y\neq LL$ because $q(HH)=q(LL)=a_0>b_0=q(HL)=q(LH)$. Therefore, for every $y\neq LL$,
$$
\frac{w(y)}{w(LL)}
=
e^{-(u(r,y)-u(r,LL))/\kappa}\Big(\frac{q(y)}{q(LL)}\Big)^\beta
\le
e^{-\delta/\kappa}.
$$
Summing over the three histories other than $LL$ yields $\sum_{y\neq LL}w(y)\le 3e^{-\delta/\kappa}w(LL)$.
Therefore,
$\hat p_{\lambda,\mu}(r;q)(LL)
=\frac{w(LL)}{w(LL)+\sum_{y\neq LL} w(y)}
\ge \frac{1}{1+3e^{-\delta/\kappa}}.$

Because $(\lambda,\bar\mu)\mapsto \kappa(\lambda,\bar\mu)=\frac1\lambda-\bar\mu\lambda$
is decreasing in both arguments, we have
$\kappa(\lambda,\bar\mu)\le \kappa(\underline\lambda,\bar\mu_0).$
Moreover, by the definition of $\bar\mu_0$,
$\bar\mu_0=\max\big\{0,\frac{1/\underline\lambda-\bar\kappa}{\underline\lambda}\big\},$
so
$\kappa(\underline\lambda,\bar\mu_0)
=\frac1{\underline\lambda}-\bar\mu_0\underline\lambda
\le \bar\kappa.$
Hence,
$\kappa(\lambda,\bar\mu)\le \bar\kappa.$
Since $x\mapsto e^{-\delta/x}$ is increasing on $(0,\infty)$, it follows that $e^{-\delta/\kappa}\le e^{-\delta/\bar\kappa}.$
Therefore,
$\hat p_{\lambda,\mu}(r;q)(LL)\ge \frac{1}{1+3e^{-\delta/\bar\kappa}}=\bar p,$
where the equality is the definition of $\bar\kappa$.

Now let $p\in\Delta(Y)$ satisfy $p(LL)\ge \bar p$. Since $x\mapsto -x\log x$ is concave,
holding fixed the total mass on $Y\setminus\{LL\}$, entropy is maximized when that
remaining mass is split equally across the three states $HH,HL,LH$. Therefore,
$$
H(p)\le -p(LL)\log p(LL)-(1-p(LL))\log\left(\frac{1-p(LL)}{3}\right)
=: \bar H(p(LL)).
$$
Moreover,
$\bar H'(x)=\log\left(\frac{1-x}{3x}\right)<0$ $\text{for all }x>1/4.$
Since $\bar p\in(1/4,1)$ and $p(LL)\ge \bar p$, we obtain
$H(p)\le \bar H(p(LL))\le \bar H(\bar p).$
Applying this to $p=\hat p_{\lambda,\mu}(r;q)$ yields
$H_r(\lambda,\bar\mu)=H(\hat p_{\lambda,\mu}(r;q))\le \bar H(\bar p).$ Combining this with the lower bound obtained earlier gives
$H_s(\lambda,\bar\mu)-H_r(\lambda,\bar\mu)
\ge \log 2-\bar H(\bar p)=H^\ast>0.$
Hence, Assumption \ref{ass:uniform_entropy_gap} holds on $[\underline\lambda,\bar\lambda]$ with constants
$\bar\mu_0$ and $H^\ast$.
\end{proof}

\subsection{Welfare Comparison}\label{sec:welfare}

This appendix complements Section \ref{sec:wel} by isolating a class of environments in which increasing complexity aversion is \textit{welfare-improving}: it
raises the agent's long-run realized payoff under the true DGP, even though it makes the agent more conservative over time. The key
mechanism is that when the true DGP actually favors the safe arm, eliminating safe-risky
cycles prevents the agent from over-experimenting with the risky action. 
\par Formally, we distinguish the agent's robust objective (which depends on
$\bar{\mu}$) from an ex-post payoff criterion evaluated under the true DGP. For any mixed
action $\alpha\in\Delta(A)$ define its objective expected payoff by
$U^\star(\alpha)
  :=
  \sum_{a\in A}\alpha(a) \bar u^\star(a)$ and $\bar u^\star(a)
  :=
  \sum_{y\in Y} p^\star_a(y)u(a,y),$
where $\m{p}^\star=(p_a^\star)_{a\in A}\in\Delta(Y)^A$ is the true DGP. Notice that the criterion above depends only on outcomes in the sense that it does \emph{not} attach any direct welfare value to model simplicity. Proposition \ref{thm:welfare_improvement} is a welfare comparison, whose
assumptions are discussed below:
\begin{itemize}
\item[(i)] pins down the ex-post ranking of actions under $\m{p}^\star$:
the safe arm $s$ is strictly better than the risky arm $r$ in objective expected
utility. Thus, any positive long-run weight on $r$ is an ex-post payoff loss.

\item[(ii)] formalizes the baseline dynamic pathology in the $\bar{\mu}=0$
model: even though $s$ is ex-post optimal by (i), there exists a $0$-optimal policy
whose $\Lambda$-limit frequency $\alpha^0$ assigns positive long-run probability to
both $r$ and $s$. This captures robustness-driven endogenous safe-risky cycles, which imply $\alpha^0(r)>0$.

\item[(iii)] is the high-$\bar{\mu}$ stabilization condition: the uniform entropy
gap (Assumption \ref{ass:uniform_entropy_gap})  ensures that for all $\bar{\mu}\in[\bar{\mu}^*,1/\bar{\lambda}^2)$,
the mixed $c$-robust equilibrium selects the safe arm in the long run,
i.e., $\alpha^{\bar{\mu}}(s)=1$. In other words, sufficiently strong complexity aversion
eliminates the safe-risky cycles of (ii).
\end{itemize}
\medskip

\begin{proposition}
\label{thm:welfare_improvement}
Suppose $A=\{r,s\}$ and assume:
\begin{enumerate}
  \item[\textnormal{(i)}] $\bar u^\star(s)>\bar u^\star(r)$;
  \item[\textnormal{(ii)}] For $\bar{\mu}=0$, there exists a $0$-optimal policy whose
  $\Lambda$-limit frequency $\alpha^0$ satisfies $\alpha^0(r)\in(0,1)$ and $\alpha^0(s)>0$;
  \item[\textnormal{(iii)}] Assumptions \ref{ass:U} and  \ref{ass:uniform_entropy_gap}  hold.
\end{enumerate}
Then, for any $\bar{\mu}\in[\bar{\mu}^*,1/\bar{\lambda}^2)$ there exists a $\bar{\mu}$-optimal policy with
$\Lambda$-limit frequency $\alpha^{\bar{\mu}}$ such that $U^\star(\alpha^{\bar{\mu}}) > U^\star(\alpha^0).$
\end{proposition}

  Under the true DGP, the safe arm dominates the risky
one, but the agent does not know this ex ante and entertains a misspecified model. When
$\bar{\mu}=0$, (ii) implies that their misspecification concerns alone can drive them to cycle between safe and
risky in the long run. Then, (iii) implies that complexity
aversion cuts off these cycles by selecting the safe arm in the long run; the agent becomes
more conservative, but in this environment conservatism is beneficial: it prevents them from
over-reacting to good runs of the risky arm that are only weakly supported by the data.
Thus, even though $\bar{\mu}>0$ makes the agent more pessimistic about worst-case payoffs
in each period, it can raise their \emph{actual} long-run payoff by improving the stability
of their behavior.

\subsubsection{Persistent complexity aversion}
\label{app:persistent_complexity}

We now study the welfare implications of allowing complexity aversion to persist in the long run ($\mu_t\rightarrow\mu_{\infty}>0$) even when misspecification concerns vanish ($\lambda_t\rightarrow0$) under correct specification. Proposition \ref{prop:persistent_mu_bad} shows this  can be \textit{welfare-reducing}. 

\begin{lemma}
\label{lem:persistent_mu_limit}
Suppose Assumptions \ref{ass:full_support} and \ref{ass:correct_spec} hold. Let
$\sigma$ be a policy such that, at each history $h_t$,
$\sigma_t(h_t)\in \arg\max_{a\in A} V_{\lambda_t,\mu_t}(a;\pi_t),$
where $\mu_t\ge 0$, $\mu_t<1/\lambda_t$ eventually, and
$\mu_t \rightarrow \mu_{\infty} \in [0,\infty)$ a.s. Let $\alpha^\Lambda$ be any $\Lambda$-limit frequency of $\sigma$. Then, on the event
$\{\alpha_t\to \alpha^\Lambda\}$, for every
$a\in A$ with $\alpha^\Lambda(a)>0$,
$V_{\lambda_t,\mu_t}(a;\pi_t)
  \longrightarrow
  \bar u^\star(a)+\mu_{\infty} H(p_a^\star)$ $\text{a.s.}$
\end{lemma}

The next result shows that if complexity aversion does not vanish after
misspecification concerns vanish, then it can be strictly welfare-reducing under
correct specification.

\begin{proposition}
\label{prop:persistent_mu_bad}
Fix any $\mu_{\infty}>0$. There exists a correctly specified environment such that
$\bar u^\star(s)>\bar u^\star(r),$
but every policy that is optimal for $V_{\lambda_t,\mu_t}(\cdot;\pi_t)$ and satisfies
$\mu_t\to\mu_{\infty}$ has $\alpha^\Lambda(r)=1$ for every $\Lambda$-limit frequency $\alpha^\Lambda$. Thus, $U^\star(\alpha^\Lambda)=\bar u^\star(r)<\bar u^\star(s).$
\end{proposition}

Proposition \ref{prop:persistent_mu_bad} identifies the pathology ruled out by our
specification $\mu_t=\bar{\mu}\lambda_t$. Under correct specification, misspecification
concerns vanish on path. If $\mu_t$ does not vanish with them, then the agent does not
converge to Bayesian expected utility. Instead, by Lemma
\ref{lem:persistent_mu_limit}, the limiting criterion is
$\bar u^\star(a)+\mu_{\infty} H(p_a^\star),$
so simplicity continues to distort choice even after misspecification concerns have
disappeared, which then causes welfare losses. The restriction $\mu_t=\bar{\mu}\lambda_t$ therefore removes this residual
wedge to prevent potential welfare losses.

\subsubsection{Proof of Proposition \ref{thm:welfare_improvement}}
\paragraph{Step 1: Structure of the $\bar{\mu}=0$ equilibrium.}
By assumption (ii), there exists a $0$-optimal policy with $\Lambda$-limit frequency
$\alpha^0$ satisfying $\alpha^0(r) \in (0,1)$ and $\alpha^0(s) > 0$. This means the agent
cycles between safe and risky actions in the long run under $\bar{\mu}=0$.

\paragraph{Step 2: Structure of the $\bar{\mu} \ge \bar{\mu}^\ast$ equilibrium.}
By assumption (iii) and Theorem \ref{thm:mu_kills_cycles}, for all
$\bar{\mu}\in[\bar{\mu}^*,1/\bar{\lambda}^2)$, every mixed $c$-robust equilibrium has
$\alpha^{\bar{\mu}}(s) = 1$. By Proposition \ref{prop:dynamic_selection}, every
$\Lambda$-limit frequency of a $\bar{\mu}$-optimal policy satisfies
$\alpha^{\bar{\mu}}(s) = 1$.

\paragraph{Step 3: Comparison of long-run payoffs.}
The objective expected payoff of a mixed action $\alpha$ is:
$U^\star(\alpha) = \alpha(r) \bar{u}^\star(r) + \alpha(s) \bar{u}^\star(s).$ For $\alpha^0$ with $\alpha^0(r) > 0$:
$$
  U^\star(\alpha^0)
  = \alpha^0(r) \bar{u}^\star(r) + (1 - \alpha^0(r)) \bar{u}^\star(s)
  = \bar{u}^\star(s) - \alpha^0(r) (\bar{u}^\star(s) - \bar{u}^\star(r)).
$$
By assumption (i), $\bar{u}^\star(s) > \bar{u}^\star(r)$, so
$\bar{u}^\star(s) - \bar{u}^\star(r) > 0$. Since $\alpha^0(r) > 0$:
$U^\star(\alpha^0) < \bar{u}^\star(s).$ For $\alpha^{\bar{\mu}}$ with $\alpha^{\bar{\mu}}(s) = 1$:
$U^\star(\alpha^{\bar{\mu}}) = \bar{u}^\star(s).$ Therefore,
$U^\star(\alpha^{\bar{\mu}}) = \bar{u}^\star(s) > U^\star(\alpha^0).$
\qed
\subsubsection{Proof of Lemma \ref{lem:persistent_mu_limit}}
The proof is nearly identical to that of Proposition \ref{prop:correct_spec_limit}. Under
correct specification, the same argument gives $\lambda_t\to 0$ a.s. and posterior
concentration on models that coincide with $\m{p}^\star$ on played actions. For every such
model $q$ and every played action $a$, Theorem
\ref{prop:arc_representation}.3 implies
$v_{\lambda_t,\mu_t}(a;\m{q})
  \longrightarrow
  \E_{p_a^\star}[u(a,y)] + \mu_{\infty} H(p_a^\star),$ 
because $q_a=p_a^\star$, $\lambda_t\to 0$, and $\mu_t\to\mu_{\infty}$. Averaging over the
posterior yields the claim.
\qed

\subsubsection{Proof of Proposition \ref{prop:persistent_mu_bad}}
Fix $\mu_{\infty}>0$. Let $A=\{r,s\}$, $Y=\{g,b\}$, $Q=\{\m{q}^\star\}$ and $q^\star_a=p^\star_a$ for all $a$, so the agent is
correctly specified by construction. Pick any $\delta\in(0,1/2)$ and then choose
$\varepsilon>0$ such that $0<\varepsilon<\mu_{\infty}\big(\log 2
  +(1-\delta)\log(1-\delta)+\delta\log\delta\big).$
Define payoffs and the true DGP by
$$
  u(s,g)=u(s,b)=1,
  \qquad
  p^\star_s(g)=1-\delta,\quad p^\star_s(b)=\delta,
$$
and
$$
  u(r,g)=1,\qquad u(r,b)=1-2\varepsilon,
  \qquad
  p^\star_r(g)=p^\star_r(b)=\frac12.
$$
Since $Q=\{\m{q}^\star\}$ and $q^\star_a=p^\star_a$ for both $a\in\{r,s\}$, Assumption
\ref{ass:correct_spec} holds. The true expected payoffs are
$\bar u^\star(s)=1$ and $\bar u^\star(r)=\frac12\cdot 1+\frac12(1-2\varepsilon)=1-\varepsilon,$
so $\bar u^\star(s)>\bar u^\star(r).$
The corresponding entropies are
$H(p^\star_r)=\log 2$ and $H(p^\star_s)=-(1-\delta)\log(1-\delta)-\delta\log\delta.$
Hence, by the choice of $\varepsilon$,
$\big(\bar u^\star(r)+\mu_{\infty} H(p^\star_r)\big)
  -\big(\bar u^\star(s)+\mu_{\infty} H(p^\star_s)\big)
  =
  -\varepsilon+\mu_{\infty}\big(H(p^\star_r)-H(p^\star_s)\big)>0.$

Now, let $\sigma$ be any policy that is optimal for $V_{\lambda_t,\mu_t}(\cdot;\pi_t)$ and satisfies $\mu_t\to\mu_\infty$. Since $Q=\{\m{q}^\star\}$ is a singleton, $\pi_t=\delta_{q^\star}$ for every $t$, so
$V_{\lambda_t,\mu_t}(a;\pi_t)=v_{\lambda_t,\mu_t}(a;\m{q}^\star)$ for each $a\in\{r,s\}$. Moreover, under correct specification, Lemma \ref{lem:persistent_mu_limit} gives $\lambda_t\to0$ a.s. Since $q_a^\star=p_a^\star$ for both $a\in\{r,s\}$, Theorem \ref{prop:arc_representation}.3 applies action by action and yields
$V_{\lambda_t,\mu_t}(a;\pi_t)\to \bar u^\star(a)+\mu_\infty H(p_a^\star)$
a.s. for each $a\in\{r,s\}$. Hence, $V_{\lambda_t,\mu_t}(r;\pi_t)-V_{\lambda_t,\mu_t}(s;\pi_t)
\to
\bar u^\star(r)+\mu_\infty H(p_r^\star)-\bar u^\star(s)-\mu_\infty H(p_s^\star)>0$
a.s.
Therefore, there exists an a.s. finite $T$ such that for all $t\ge T$,
$V_{\lambda_t,\mu_t}(r;\pi_t)>V_{\lambda_t,\mu_t}(s;\pi_t).$
Thus, $r$ is the unique best reply from date $T$ onward, so every optimal policy chooses $r$ at every date $t\ge T$. It follows that $\alpha_t(r)\to1$, hence every $\Lambda$-limit frequency satisfies $\alpha^\Lambda(r)=1$. Consequently,
$U^\star(\alpha^\Lambda)=\bar u^\star(r)<\bar u^\star(s).$
\qed

\subsection{Other connections to \citet{gabaix25}}\label{app:gabaix}
This appendix shows how our framework connects to other examples in \citet{gabaix25}.
\paragraph{Pricing: Gabaix's Section 4.}
The construction of Section \ref{subsec:gabaix} applies to Gabaix's pricing application, where the choice variable is a scalar coefficient $b$ measuring the sensitivity of price to marginal cost. Gabaix's welfare criterion here  takes the form 
$V_P^{G}(b)=\widetilde V_P(b)-\Gamma_P |b|,$
where $\widetilde V_P$ is the smooth allocative part of the pricing problem and the first-order term comes from the dread cost of misperceiving a variable price schedule. Applying the same binary-narrative construction with $x=b$ and
$
\gamma=\Gamma_P$ yields
$
V_P^{FM}(b)=\frac{1}{\lambda}\log 2+\widetilde V_P(b)-\Gamma_P |b|,
$
which coincides with Gabaix's pricing objective up to the same additive constant. Here again our model reproduces the same reduced-form pricing problem and therefore the same threshold logic for rigid pricing.

\paragraph{Taxes: Gabaix's Section 5.}
The tax application is slightly different. The relevant reduced-form object is not a single absolute-value penalty, but an additive penalty across tax features. A convenient way to represent the problem is
$$
V^G_T(z)=\widetilde V_T(z)-\sum_{k=1}^K \Gamma_k |z_k|,
$$
where $z=(z_1,\dots,z_K)$ is a finite-dimensional representation of the tax schedule, $\widetilde V_T$ is the smooth Ramsey part of the objective, and the coefficients $\Gamma_k\geq 0$ summarize the strength of first-order complexity aversion along each coordinate. A single binary narrative is not enough to reproduce this objective exactly, because it generates one kinked term, not a sum of coordinatewise kinks. However, the same reduced-form objective is obtained from a separable extension of our construction. Specifically, for each coordinate $k$, introduce a symmetric binary narrative $Y_k=\{0,1\}$ with associated payoff
$f_k(0)=0$ and
$f_k(1)=-\Gamma_k |z_k|.$ In the high-simplicity regime,
$v_{\lambda,\mu}(f_k;q_k)=\frac{1}{\lambda}\log 2-\Gamma_k |z_k|.$
Summing these feature-specific terms gives
$$
V^{FM}_T(z):=\widetilde V_T(z)+\sum_{k=1}^K v_{\lambda,\mu}(f_k;q_k)
= \frac{K}{\lambda}\log 2 + \widetilde V_T(z)-\sum_{k=1}^K \Gamma_k |z_k|.
$$
Thus, our model also reproduces Gabaix's reduced-form tax objective, up to an additive constant, for any fixed finite-dimensional representation of the tax schedule. In particular, any implication that depends only on that objective---compression of tax differentials, thresholding of small deviations to zero, and uniform taxation for sufficiently strong complexity concerns---is inherited by our model. What our framework does \emph{not} reproduce  by itself is Gabaix's additional analysis of endogenous tax categories and the closed-form formula for the optimal number of tax buckets, which relies on further assumptions on cognition costs and on the distribution of elasticities.

\subsection{Omitted Proofs}\label{ola:omittedproofs}
\subsubsection{Proof of Theorem \ref{prop:arc_representation}}

\paragraph{Intramodel Sure-Thing Principle \citep[][Axiom 3]{lanzani25w}.}
Fix a finite state space $\Omega$, a model $\rho\in\Delta(\Omega)$, and let $X^\Omega$ denote the set of acts from $\Omega$ to the consequence space $X$. For $W\subseteq\Omega$ and acts $f,h\in X^\Omega$, let $fWh\in X^\Omega$ denote the act defined by
$$
(fWh)(\omega)=
\begin{cases}
f(\omega), & \omega\in W,\\
h(\omega), & \omega\in \Omega\setminus W.
\end{cases}
$$
If $\succeq_\rho$ denotes the preference conditional on the model being $\rho$, then Axiom 3 requires that
$$
fWh \succeq_\rho gWh \Longrightarrow fWh' \succeq_\rho gWh'
\qquad
\forall W\subseteq \Omega,\ \forall f,g,h,h'\in X^\Omega.
$$
In our Anscombe-Aumann restatement, when $Q=\{\rho\}$ is a singleton, we identify $\Omega=Y$ and define $\succeq_\rho$ by
$f\succeq_\rho g
\iff
v_{\lambda,\mu}(f;\rho)\ge v_{\lambda,\mu}(g;\rho).$

\begin{lemma}\label{lem:axiom3_counterexample}
Suppose $\lambda>0$, $\mu\ge 1/\lambda$, and $|Y|\ge 2$. Then, the high-simplicity criterion in Theorem \ref{prop:arc_representation}.2 violates \citet[][Axiom 3]{lanzani25w}. Thus, it is not an ARC preference.
\end{lemma}

\begin{proof}[Proof of Lemma \ref{lem:axiom3_counterexample}]
It suffices to show a violation in a singleton-model environment. Let
$Q=\{\rho\}$, $\Omega=Y=\{1,2\}$, and $\rho(1)=\rho(2)=1/2$. Since $\rho$ has full support, $\text{supp}(\rho)=\Omega$. By Lemma \ref{prop:corner},
for every act $f\in X^\Omega$,
$$
v_{\lambda,\mu}(f;\rho)
=
\min_{\omega\in\Omega}
\left\{
u(f(\omega))+\frac{1}{\lambda}\log\frac{1}{\rho(\omega)}
\right\}
=
\frac{1}{\lambda}\log 2+\min_{\omega\in\Omega}u(f(\omega)).
$$
Thus, conditional on $\rho$, the preference $\succsim_\rho$ ranks acts by their worst utility
level. Since $u$ is affine and nonconstant and $X$ is convex, there exist consequences
$x_0,x_1,x_2\in X$ such that
$u(x_0)<u(x_1)<u(x_2).$
Let $W=\{1\}$ and define acts $f,g,h,h'\in X^\Omega$ by
$$
f(1)=x_1,\qquad g(1)=x_2,\qquad h(\omega)=x_0\ \forall \omega\in\Omega,
\qquad h'(\omega)=x_2\ \forall \omega\in\Omega.
$$
Then,
$fWh=(x_1,x_0)$ and $gWh=(x_2,x_0),$
so
\begin{align*}
v_{\lambda,\mu}(fWh;\rho)
&=
\frac{1}{\lambda}\log 2+\min\{u(x_1),u(x_0)\}
=
\frac{1}{\lambda}\log 2+u(x_0),\\
v_{\lambda,\mu}(gWh;\rho)
&=
\frac{1}{\lambda}\log 2+\min\{u(x_2),u(x_0)\}
=
\frac{1}{\lambda}\log 2+u(x_0).
\end{align*}
Hence,
$fWh \sim_\rho gWh,$
so $fWh \succsim_\rho gWh$. Now replace $h$ by $h'$. We obtain
$fWh'=(x_1,x_2)$ and $gWh'=(x_2,x_2),$
so
$v_{\lambda,\mu}(fWh';\rho)
=
\frac{1}{\lambda}\log 2+\min\{u(x_1),u(x_2)\}
=
\frac{1}{\lambda}\log 2+u(x_1),$
but
$v_{\lambda,\mu}(gWh';\rho)
=
\frac{1}{\lambda}\log 2+\min\{u(x_2),u(x_2)\}
=
\frac{1}{\lambda}\log 2+u(x_2).$
Thus,
$gWh' \succ_\rho fWh',$
and hence $fWh' \not\succsim_\rho gWh'$. We found $W\subseteq\Omega$ and acts $f,g,h,h'$ such that
$fWh \succsim_\rho gWh$ but $fWh' \not\succsim_\rho gWh',$
which violates Axiom 3. Since Axiom 3 is necessary for ARC preferences  \citep[][Theorem 3]{lanzani25w}, the high-simplicity criterion cannot be ARC.
\end{proof}
\begin{proof}[Proof of Lemma \ref{prop:corner}]
Fix $q\in\Delta({Y})$ and an act $f$. If $p\not\ll q$ then ${R}(p\|q)=+\infty$ and such $p$ cannot be optimal. Hence, restrict attention to $p\in\Delta({Y})$ with $\text{supp}(p)\subseteq \text{supp}(q)$. Then,
\begin{align}
\E_p[u(f)] +\frac{1}{\lambda}{R}(p\|q)+\mu{H}(p)
&=\sum_{y\in{Y}} p(y)u(f(y))
 +\Big(\frac{1}{\lambda}-\mu\Big)\sum_{y\in{Y}} p(y)\log p(y)
 -\frac{1}{\lambda}\sum_{y\in{Y}} p(y)\log q(y)\nonumber\\
&=\sum_{y\in{Y}} p(y)\Big(u(f(y)) -\frac{1}{\lambda}\log q(y)\Big)
 +\kappa\sum_{y\in{Y}} p(y)\log p(y).
\label{eq:expand}
\end{align}

\medskip
\noindent\emph{Step 1.}
Since $p(y)\in[0,1]$ for each $y$, one has $\log p(y)\le 0$ and hence $\sum_{y} p(y)\log p(y)\le 0$.
Since $\kappa=\frac{1}{\lambda}-\mu\leq0$, it follows that
\begin{equation}
\label{eq:nonneg}
\kappa\sum_{y\in{Y}} p(y)\log p(y)\ge 0.
\end{equation}
Moreover, $\sum_y p(y)\log p(y)=0$ holds if and only if $p(y)\in\{0,1\}$ for all $y$, i.e., if and only if $p$ is a Dirac measure.

\medskip
\noindent\emph{Step 2: minimizer is a Dirac measure.}
Let $a_q^f(y):=u(f(y)) -\frac{1}{\lambda}\log q(y)\in\R\cup\{+\infty\},$
where $a_q^f(y)=+\infty$ when $q(y)=0$ (consistent with ${R}(\delta_y\|q)=+\infty$).
For every feasible $p$ we have
\begin{equation}
\label{eq:expectation-lb}
\sum_{y\in{Y}} p(y)a_q^f(y)\ge \min_{y\in\text{supp}(q)} a_q^f(y),
\end{equation}
with equality if $p=\delta_{y^*}$ for any minimizer $y^*$ of $a_q^f$ on $\text{supp}(q)$.

\medskip
\noindent\emph{Step 3: combine all steps.}
Combining \eqref{eq:expand}-\eqref{eq:expectation-lb}, for every feasible $p$,
$$
\E_p[u(f)] +\frac{1}{\lambda}{R}(p\|q)+\mu{H}(p)
\ge \sum_{y} p(y)a_q^f(y)
\ge \min_{y\in\text{supp}(q)} a_q^f(y).
$$
Choose $p=\delta_{y^*}$ where $y^*$ minimizes $a_q^f$ over $\text{supp}(q)$. Then, $\sum_y p(y)\log p(y)=0$ so the first inequality holds with equality, and \eqref{eq:expectation-lb} holds with equality as well. Hence, the minimum over $p$ equals $\min_{y\in\text{supp}(q)} a_q^f(y)$, which coincides with \eqref{eq:corner}.
\end{proof}

\begin{proof}[Proof of Theorem \ref{prop:arc_representation}.2] When $\lambda>0$ and $\mu\geq1/\lambda$, the preference is not ARC (Lemma \ref{lem:axiom3_counterexample}). We  now aim to prove that there exist a constant $C\in\R$ independent of acts and a grounded, convex, lower semicontinuous cost function $\hat c:\Delta({Y})\to[0,+\infty]$ such that for all $f$,
\begin{equation}
\label{eq:variational-rep}
V_{\lambda,\mu}(f;\pi)= C + \min_{p\in\Delta({Y})}\Big\{\E_p[u(f)] + \hat c(p)\Big\}.
\end{equation}
Let $n:=|{Y}|$ and identify utility acts with vectors in $\R^n$.
For each act $f$ define its utility vector $\xi_f\in\R^n$ by $\xi_f(y):=u(f(y))$.
Define a functional $I$ on the convex set $\mathcal{D}:=u({X})^{{Y}}\subset\R^n$ by
\begin{equation}
\label{eq:I}
I(\xi):=\sum_{q\in Q}\pi(q)\min_{y\in\text{supp}(q)}\Big\{\xi(y) +\frac{1}{\lambda}\log\frac{1}{q(y)}\Big\}.
\end{equation}
 Then, by Lemma \ref{prop:corner}, $V_{\lambda,\mu}(f;\pi)=I(\xi_f)$ for all acts $f$.

\medskip
\noindent\emph{Step 1: basic properties of $I$.}
Let $b_q(y):=\frac{1}{\lambda}\log(1/q(y))$, then each map $\xi\mapsto \min_{y\in\text{supp}(q)}\{\xi(y)+b_q(y)\}$ is the pointwise minimum of finitely many affine functionals in $\xi$, hence is concave and continuous on $\R^n$.
A finite weighted sum of concave (respectively, continuous) functions is concave (respectively, continuous), so $I$ is concave and continuous on $\R^n$ and hence on $\mathcal{D}$.
Moreover, $I$ is monotone (if $\xi\le\eta$ coordinatewise then $I(\xi)\le I(\eta)$).

\medskip
\noindent\emph{Step 2: supergradients exist and are probabilities.}
Fix $\xi\in\R^n$. For each $q\in Q$, let
$A_q(\xi):=\text{arg min}_{y\in\text{supp}(q)}\{\xi(y)+b_q(y)\}\neq\emptyset.$
Define the (concave) function $m_q(\xi):=\min_{y\in\text{supp}(q)}\{\xi(y)+b_q(y)\}$.
Since $m_q$ is the minimum of finitely many affine maps, its superdifferential at $\xi$ contains the convex hull of the gradients of the active affine pieces, i.e.
\begin{equation}
\label{eq:subdiff-mq}
\operatorname{co}\{\delta_y: y\in A_q(\xi)\}\subseteq \partial m_q(\xi)\subseteq \Delta({Y}),
\end{equation}
where $\delta_y\in\Delta({Y})$ denotes the Dirac mass on $y$ and $\operatorname{co}$ denotes the convex hull.
Indeed, for any $y\in A_q(\xi)$ we have $m_q(\eta)\le \eta(y)+b_q(y)=m_q(\xi)+\delta_y\cdot(\eta-\xi)$ for all $\eta$, and taking convex combinations preserves the supergradient inequality.

Since $I=\sum_q \pi(q)m_q$, the sum of supergradients is a supergradient of the sum: if $p_q\in\partial m_q(\xi)$ for each $q$, then $\sum_{q\in Q}\pi(q) p_q\in \partial I(\xi).$
In particular, $\partial I(\xi)\neq\emptyset$, because each $\partial m_q(\xi)$ is nonempty by \eqref{eq:subdiff-mq}.
Moreover, every vector of the form $\sum_q \pi(q)p_q$ (with $p_q\in\partial m_q(\xi)$) belongs to $\Delta({Y})$ because $\Delta({Y})$ is convex and closed under finite convex combinations.

\medskip
\noindent\emph{Step 3: construct a convex cost function.}
Define $c:\Delta({Y})\to(-\infty,+\infty]$ by
\begin{equation}
\label{eq:conjugate}
 c(p):=\sup_{\eta\in\mathcal{D}}\big\{I(\eta)-p\cdot \eta\big\}.
\end{equation}
Since $p\mapsto -p\cdot \eta$ is affine for each fixed $\eta$ and a supremum of affine functions is convex and lower semicontinuous, $c$ is convex and lower semicontinuous.

\medskip
\noindent\emph{Step 4: the variational representation.}
First, for any $p\in\Delta({Y})$ and any $\xi\in\mathcal{D}$, the definition \eqref{eq:conjugate} implies
\begin{equation}
\label{eq:ineq1}
 c(p)\ge I(\xi)-p\cdot\xi\implies p\cdot\xi+c(p)\ge I(\xi).
\end{equation}
Taking the infimum over $p\in\Delta({Y})$ yields
\begin{equation}
\label{eq:inf-ge}
\inf_{p\in\Delta({Y})}\{p\cdot\xi+c(p)\}\ge I(\xi).
\end{equation}

For the reverse inequality, fix $\xi\in\mathcal{D}$ and choose any $p^\xi\in\partial I(\xi)$, which exists by Step 2.
By definition of superdifferential for a concave function, $p^\xi\in\partial I(\xi)$ means that
$I(\eta)\le I(\xi)+p^\xi\cdot(\eta-\xi)$ $\forall \eta\in\mathcal{D}.$
Rearranging gives $I(\eta)-p^\xi\cdot\eta\le I(\xi)-p^\xi\cdot\xi$ for all $\eta\in\mathcal{D}$.
Taking supremum over $\eta$ and using \eqref{eq:conjugate} yields
\begin{equation}
\label{eq:c-at-supergrad}
 c(p^\xi)=\sup_{\eta\in\mathcal{D}}\{I(\eta)-p^\xi\cdot\eta\}\le I(\xi)-p^\xi\cdot\xi.
\end{equation}
On the other hand, applying \eqref{eq:ineq1} at $p=p^\xi$ gives $c(p^\xi)\ge I(\xi)-p^\xi\cdot\xi$.
Therefore equality holds in \eqref{eq:c-at-supergrad} and hence
\begin{equation}
\label{eq:touching}
 I(\xi)=p^\xi\cdot\xi+c(p^\xi)\ge \inf_{p\in\Delta({Y})}\{p\cdot\xi+c(p)\}.
\end{equation}
Combining \eqref{eq:inf-ge} and \eqref{eq:touching}, we conclude
\begin{equation}
\label{eq:I-rep}
 I(\xi)=\inf_{p\in\Delta({Y})}\{p\cdot\xi+c(p)\}\qquad\forall \xi\in\mathcal{D}.
\end{equation}

\medskip
\noindent\emph{Step 5: ground the cost.}
By Step 2, $\partial I(\xi)$ is nonempty for every $\xi$; pick some $\xi_0\in\mathcal{D}$ and $p^{\xi_0}\in\partial I(\xi_0)$. Then by \eqref{eq:touching}, $c(p^{\xi_0})=I(\xi_0)-p^{\xi_0}\cdot\xi_0<+\infty$, so $c$ is proper on $\Delta({Y})$.
Let $m:=\inf_{p\in\Delta({Y})} c(p)\in\R$, where finiteness follows because $c$ is proper, lower semicontinuous, and $\Delta({Y})$ is compact.
Define the grounded cost $\hat c(p):=c(p)-m\ge 0$ and the constant $C:=m$.
Then, \eqref{eq:I-rep} implies for all $\xi\in\mathcal{D}$, $I(\xi)= C + \inf_{p\in\Delta({Y})}\{p\cdot\xi+\hat c(p)\}.$
Substituting $\xi=\xi_f$ yields \eqref{eq:variational-rep} with $\E_p[u(f)]=p\cdot\xi_f$.\end{proof}

\subsubsection{Proof of Proposition \ref{prop:cri_ri_logit}}
\begin{proof}[Proof of Proposition \ref{prop:cri_ri_logit}]
Let $\kappa:=\frac{1}{\lambda}-\mu>0$ and define the full saddle objective
$$
  \mathcal{F}(\psi,m)
  :=
  \sum_{\omega\in\Omega} m(\omega)\sum_{a\in A} v(a,\omega)\psi(a|\omega)
  +\frac{1}{\lambda}R(m\Vert g)
  +\mu H(m)
  - \xi \sum_{\omega\in\Omega} g(\omega)\sum_{a\in A}\psi(a|\omega)\log\frac{\psi(a|\omega)}{\bar{\psi}(a)},
$$
where $\bar{\psi}(a)=\sum_{\omega}g(\omega)\psi(a|\omega)$.

\smallskip
\noindent\textit{Step 1. Existence of a saddle point for the full problem.}
The sets $\Delta(A)^\Omega$ and $\Delta(\Omega)$ are compact and convex. For each fixed $m$,
the mapping $\psi\mapsto \mathcal{F}(\psi,m)$ is continuous and concave on $\Delta(A)^\Omega$:
the payoff term is linear in $\psi$, and the Shannon information term is concave with a minus sign.
Indeed, writing
$I(\psi;g):=\sum_{\omega} g(\omega) R\bigl(\psi(\cdot|\omega)\Vert\bar{\psi}\bigr)$,
we have
$\xi \sum_{\omega\in\Omega} g(\omega)\sum_{a\in A}\psi(a|\omega)\log\frac{\psi(a|\omega)}{\bar{\psi}(a)}
  = \xi I(\psi;g),$
and $I(\psi;g)$ is convex because $R(\cdot\Vert\cdot)$ is jointly convex and both
$\psi(\cdot|\omega)$ and $\bar{\psi}$ are affine in $\psi$.

For each fixed $\psi$, the mapping $m\mapsto \mathcal{F}(\psi,m)$ is continuous and convex on $\Delta(\Omega)$.
Using $H(m)=-\sum_{\omega}m(\omega)\log m(\omega)$ and $0\log0:=0$, we can rewrite
$$
  \frac{1}{\lambda}R(m\Vert g)+\mu H(m)
  =
  \kappa \sum_{\omega\in\Omega} m(\omega)\log m(\omega)
  -\frac{1}{\lambda}\sum_{\omega\in\Omega} m(\omega)\log g(\omega),
$$
and this is convex in $m$ because $\kappa>0$ and $x\mapsto x\log x$ is convex. Therefore, Sion's minimax theorem applies to the full objective $\mathcal{F}$, so there exists a saddle point
$(\psi^\star,m^\star)\in\Delta(A)^\Omega\times\Delta(\Omega)$ such that
$\mathcal F(\psi,m^\star)\le \mathcal F(\psi^\star,m^\star)\le \mathcal F(\psi^\star,m)$ for all $(\psi,m)\in\Delta(A)^\Omega\times\Delta(\Omega).$
In particular, $\psi^\star$ maximizes $\psi\mapsto \mathcal{F}(\psi,m^\star)$ over $\Delta(A)^\Omega$, and,
by Observation \ref{prop:cri_ri_mstar}, $m^\star$ minimizes $m\mapsto \mathcal F(\psi^\star,m)$ over $\Delta(\Omega)$.
Since $g$ has full support, \eqref{eq:cri_ri_mstar} implies $m^\star(\omega)>0$ for every $\omega\in\Omega$.

\smallskip
\noindent\textit{Step 2. Interior choice probabilities on the support of $\bar{\psi}^\star$.}
Fix $a\in A$ and assume $\bar{\psi}^\star(a)>0$. We claim that then $\psi^\star(a|\omega)>0$ for every $\omega$.
Suppose instead that $\psi^\star(a|\omega_0)=0$ for some $\omega_0$.
Choose any $b\in A$ with $\psi^\star(b|\omega_0)>0$ and consider, for $\varepsilon>0$ small,
a feasible perturbation $\psi^\varepsilon$ that increases $\psi(a|\omega_0)$ by $\varepsilon$ and decreases
$\psi(b|\omega_0)$ by $\varepsilon$, leaving all other components unchanged.
Since $m^\star(\omega_0)>0$ and payoffs are bounded, the change in the payoff term
in $\mathcal{F}(\cdot,m^\star)$ is $O(\varepsilon)$.
However, the Shannon term produces a strictly first-order gain: for $\varepsilon\downarrow0$,
$\frac{1}{\varepsilon}\big[
  -\xi g(\omega_0)\varepsilon \log\frac{\varepsilon}{\bar{\psi}^\star(a)+O(\varepsilon)}
  \big]
  = -\xi g(\omega_0)\log \varepsilon + O(1)\longrightarrow +\infty,$
so $\mathcal{F}(\psi^\varepsilon,m^\star)>\mathcal{F}(\psi^\star,m^\star)$ for $\varepsilon$ small enough,
contradicting the optimality of $\psi^\star$ given $m^\star$.
Hence $\psi^\star(a|\omega)>0$ for all $\omega$ whenever $\bar{\psi}^\star(a)>0$.
Under the maintained assumption that $\bar{\psi}^\star(a)>0$ for all $a\in A$, it follows that
$\psi^\star(a|\omega)>0$ for all $(a,\omega)$.

\smallskip
\noindent\textit{Step 3. First-order conditions in $\psi$ at the saddle point.}
Consider the concave problem
$$
  \max_{\psi\in\Delta(A)^\Omega}\mathcal{F}(\psi,m^\star).
$$
Form the Lagrangian
$\mathcal{L}(\psi,\zeta)
  :=
  \mathcal{F}(\psi,m^\star)
  + \sum_{\omega\in\Omega}\zeta(\omega)\big(\sum_{a\in A}\psi(a|\omega)-1\big),$
with multipliers $\zeta(\omega)\in\R$.
For each $(a,\omega)$, stationarity requires
\begin{equation}\label{eq:cri_ri_stationarity}
  0
  =
  \frac{\partial \mathcal{L}}{\partial \psi(a|\omega)}(\psi^\star,\zeta)
  =
  m^\star(\omega)v(a,\omega)
  - \xi\frac{\partial I(\psi;g)}{\partial\psi(a|\omega)}\Big|_{\psi=\psi^\star}
  + \zeta(\omega).
\end{equation}

We now compute $\partial I/\partial\psi(a|\omega)$ explicitly.
Write
$I(\psi;g)
  =
  \sum_{\omega} g(\omega)\sum_{a}\psi(a|\omega)\log\psi(a|\omega)
  -\sum_{a}\bar{\psi}(a)\log\bar{\psi}(a),$
where we recall $\bar{\psi}(a)=\sum_{\omega}g(\omega)\psi(a|\omega)$. Hence, for $\psi(a|\omega)>0$ and $\bar{\psi}(a)>0$,
$\frac{\partial}{\partial \psi(a|\omega)}
  \sum_{\omega'} g(\omega')\sum_{b}\psi(b|\omega')\log\psi(b|\omega')
  =
  g(\omega)\bigl(1+\log\psi(a|\omega)\bigr),$
and
$\frac{\partial}{\partial \psi(a|\omega)}\sum_{b}\bar{\psi}(b)\log\bar{\psi}(b)
  =
  (1+\log\bar{\psi}(a))\cdot \frac{\partial \bar{\psi}(a)}{\partial \psi(a|\omega)}
  =
  g(\omega)\bigl(1+\log\bar{\psi}(a)\bigr).$
Subtracting yields
\begin{equation}\label{eq:cri_ri_I_deriv}
  \frac{\partial I(\psi;g)}{\partial\psi(a|\omega)}
  =
  g(\omega)\log\frac{\psi(a|\omega)}{\bar{\psi}(a)}.
\end{equation}
Substituting \eqref{eq:cri_ri_I_deriv} into \eqref{eq:cri_ri_stationarity} gives
$\log\frac{\psi^\star(a|\omega)}{\bar{\psi}^\star(a)}
  =
  \frac{m^\star(\omega)}{\xi g(\omega)}v(a,\omega)
  + \tilde{\zeta}(\omega),$ where $
  \tilde{\zeta}(\omega):=\frac{\zeta(\omega)}{\xi g(\omega)}.$
Exponentiating and using $\sum_{a}\psi^\star(a|\omega)=1$ to normalize across $a$ yields
$$
  \psi^\star(a|\omega)
  =
  \frac{\bar{\psi}^\star(a)\exp \left\{\frac{m^\star(\omega)}{\xi g(\omega)}v(a,\omega)\right\}}
       {\sum_{b\in A}\bar{\psi}^\star(b)\exp \left\{\frac{m^\star(\omega)}{\xi g(\omega)}v(b,\omega)\right\}}
  =
  \frac{\bar{\psi}^\star(a)\exp \left\{\frac{v(a,\omega)}{\xi g(\omega)/m^\star(\omega)}\right\}}
       {\sum_{b\in A}\bar{\psi}^\star(b)\exp \left\{\frac{v(b,\omega)}{\xi g(\omega)/m^\star(\omega)}\right\}}.
$$
Lastly, by definition \eqref{eq:cri_ri_scale},
$\xi_\omega(\psi^\star)=\xi\,\frac{g(\omega)}{m^\star(\omega)},$
which gives \eqref{eq:cri_ri_logit}.
\end{proof}

\subsubsection{Proof of Proposition \ref{prop:home_bias}}

\begin{proof}[Proof of Proposition \ref{prop:home_bias}]
Fix $\lambda>0$ and $\mu\in[0,1/\lambda)$ and define $\kappa:=\frac{1}{\lambda}-\mu>0$ and $\beta:=\frac{1}{1-\lambda\mu}>0$ as in Lemma \ref{lem:closed_form}.
For $a\in\{d,f\}$ define the normalizing constant
$\mathcal Z_a
:=
\sum_{y\in Y}\exp\{-u(a,y)/\kappa\}\,q_a(y)^{\beta}.$
By Lemma \ref{lem:closed_form}, the unique minimizer of \eqref{eq:static_criterion} is $\hat p_{\lambda,\mu}(a;\m{q})(y)=\frac{\exp\{-u(a,y)/\kappa\}q_a(y)^\beta}{\mathcal Z_a}$ for $y\in Y$.
We first show the value identity
\begin{equation}\label{eq:homebias_value_identity}
v_{\lambda,\mu}(a;\m{q})=-\kappa\log \mathcal Z_a,\qquad a\in\{d,f\}.
\end{equation}
For any $p\in\Delta(Y)$ with $p\ll q_a$, expand \eqref{eq:static_criterion} using $H(p)=-\sum_y p(y)\log p(y)$:
$$
\sum_{y}u(a,y)p(y)+\frac{1}{\lambda}R(p\Vert q_a)+\mu H(p)
=
\sum_y u(a,y)p(y)
+\kappa\sum_y p(y)\log p(y)
-\frac{1}{\lambda}\sum_y p(y)\log q_a(y).
$$
Evaluating at $p=\hat p_{\lambda,\mu}(a;\m{q})$ and using $
\log \hat p_{\lambda,\mu}(a;\m{q})(y)=-\frac{u(a,y)}{\kappa}+\beta\log q_a(y)-\log \mathcal Z_a$ yields \begin{align*}
\kappa\sum_y \hat p_{\lambda,\mu}(a;\m{q})(y)\log \hat p_{\lambda,\mu}(a;\m{q})(y)
&=
-\sum_y \hat p_{\lambda,\mu}(a;\m{q})(y)u(a,y)
+\kappa\beta\sum_y \hat p_{\lambda,\mu}(a;\m{q})(y)\log q_a(y)
-\kappa\log \mathcal Z_a.
\end{align*}
Substituting into the objective at $\hat p_{\lambda,\mu}(a;\m{q})$ yields
$$
v_{\lambda,\mu}(a;\m{q})
=
\Big(\kappa\beta-\frac{1}{\lambda}\Big)\sum_y \hat p_{\lambda,\mu}(a;\m{q})(y)\log q_a(y)
-\kappa\log \mathcal Z_a.
$$
Since $\kappa=\frac{1-\lambda\mu}{\lambda}$ and $\beta=\frac{1}{1-\lambda\mu}$, we have $\kappa\beta=\frac{1}{\lambda}$, so the first term vanishes and \eqref{eq:homebias_value_identity} follows.

Under the binary payoff specification,
$$
\mathcal Z_a
=
e^{-1/\kappa}\,q_a(y^*)^{\beta}
+\sum_{i=1}^N q_a(y_i)^{\beta}
=
e^{-1/\kappa}(1-\delta)^{\beta}
+\sum_{i=1}^N q_a(y_i)^{\beta}.
$$
If $\mu=0$, then $\beta=1$ and $\sum_{i=1}^N q_d(y_i)=\sum_{i=1}^N q_f(y_i)=\delta$, hence $\mathcal Z_d=\mathcal Z_f$ and therefore
$v_{\lambda,0}(d;\m{q})=v_{\lambda,0}(f;\m{q})$ by \eqref{eq:homebias_value_identity}.

If $\mu\in(0,1/\lambda)$, then $\beta>1$ and $x\mapsto x^\beta$ is strictly convex on $(0,\infty)$.
Since $(q_f(y_1),\dots,q_f(y_N))$ is not constant and $\frac{1}{N}\sum_{i=1}^N q_f(y_i)=\delta/N$, Jensen's inequality implies
$\frac{1}{N}\sum_{i=1}^N q_f(y_i)^\beta
>
\big(\frac{1}{N}\sum_{i=1}^N q_f(y_i)\big)^\beta
=
\big(\frac{\delta}{N}\big)^\beta,$
 so
$\sum_{i=1}^N q_f(y_i)^\beta
>
N\big(\frac{\delta}{N}\big)^\beta
=
\sum_{i=1}^N q_d(y_i)^\beta.$
Hence, $\mathcal Z_f>\mathcal Z_d$, and \eqref{eq:homebias_value_identity} with $\kappa>0$ gives $v_{\lambda,\mu}(d;\m{q})=-\kappa\log \mathcal Z_d > -\kappa\log \mathcal Z_f = v_{\lambda,\mu}(f;\m{q}).$
\end{proof}

\subsubsection{Proof of Proposition \ref{prop:growth_misspec_fails}}
Given a model $p\in\Delta({Y})$, the growth-optimal mixed strategy problem in \citet[][eq. (7)]{growth22} is
\begin{equation}\label{eq:growth22_problem}
V_0(p):=\max_{\alpha\in\Delta(A)} \E_p\Big[\log\Big(\E_{\alpha}\big[e^{u(a,{y})}\big]\Big)\Big]
=\max_{\alpha\in\Delta(A)} \E_p\Big[\log\Big(\sum_{a\in A}\alpha(a)e^{u(a,{y})}\Big)\Big].
\end{equation}
Given an optimizer $\alpha_0^*(p)$ of \eqref{eq:growth22_problem}, define the sampled choice rule
\begin{equation}\label{eq:sampled_choice_rule}
q_0^*(a\mid {y}):=\frac{\alpha_0^*(p)(a)e^{u(a,{y})}}{\sum_{b\in A}\alpha_0^*(p)(b)e^{u(b,{y})}},
\qquad
q_0^*(a):=\sum_{{y}\in{Y}}p({y})q_0^*(a\mid {y}),
\end{equation}
and, whenever $q_0^*(a)>0$, the associated posterior
\begin{equation}\label{eq:posterior_dollar}
q_0^*({y}\mid a):=\frac{p({y})q_0^*(a\mid{y})}{q_0^*(a)}.
\end{equation}
The regularity condition used in \citet[][Proposition 4]{growth22} is stated below.

\begin{assumption}[Regularity Condition 2]\label{ass:rc2}
 \normalfont   The misspecified model $p'\in\Delta({Y})$ lies in the convex hull of $\{q_0^*(\cdot\mid a): a\in A, q_0^*(a)>0\}$.
\end{assumption}
\begin{lemma}\label{lem:growth_closed_form}
For every $\mu\in[0,1)$, $\alpha\in\Delta(A)$ and ${y}\in{Y}$,
\begin{equation}\label{eq:growth_closed_form}
G_\mu(\alpha,{y})
=
(1-\mu)\log\Big(\sum_{a\in A}\alpha(a)^{\frac{1}{1-\mu}}e^{\frac{u(a,{y})}{1-\mu}}\Big).
\end{equation}
In particular, $G_0(\alpha,{y})=\log\big(\sum_{a\in A}\alpha(a)e^{u(a,{y})}\big)$, so $V_0(\cdot)$  coincides with \eqref{eq:growth22_problem}.
\end{lemma}

\begin{proof}[Proof of Lemma \ref{lem:growth_closed_form}]
Fix $\mu\in[0,1)$ and ${y}\in{Y}$. Since $G_\mu(\alpha,{y})$ is minus a minimization problem, we may write
$G_\mu(\alpha,{y})
=
\max_{q\in\Delta(A)}
\big\{
-\E_q[-u(a,{y})]
-
R(q\Vert \alpha)
-
\mu H(q)
\big\}.$
Let $s:=\frac{1}{1-\mu}>1$ and define $v_{y}(a):=u(a,{y})+\log\alpha(a)$. Then,
$$
-\E_q[-u(a,{y})]
-
R(q\Vert \alpha)
-
\mu H(q)
=
\sum_a q(a)v_{y}(a)-(1-\mu)\sum_a q(a)\log q(a).
$$
Thus, for fixed $\alpha,{y}$, the problem is
$\max_{q\in\Delta(A)}\big\{\sum_a q(a)v_{y}(a)-(1-\mu)\sum_a q(a)\log q(a)\big\}.$
The objective function is strictly concave in $q$ (since $1-\mu>0$) and the Lagrangian first-order conditions yield the unique maximizer
$q^*(a)\propto \exp(v_{y}(a)/(1-\mu))=\alpha(a)^{s}e^{s u(a,{y})}$.
Substituting this optimizer back gives the standard log-sum-exp value
$$
G_\mu(\alpha,{y})=(1-\mu)\log\Big(\sum_a \exp \Big(\frac{v_{y}(a)}{1-\mu}\Big)\Big)
=(1-\mu)\log\Big(\sum_a \alpha(a)^{\frac{1}{1-\mu}}e^{\frac{u(a,{y})}{1-\mu}}\Big),
$$
which is \eqref{eq:growth_closed_form}. Setting $\mu=0$ yields $G_0(\alpha,{y})=\log(\sum_a \alpha(a)e^{u(a,{y})})$ and hence $V_0(\cdot)$ coincides with \eqref{eq:growth22_problem}.
\end{proof}

\begin{lemma}\label{lem:growth_contraction}
Fix $\mu\in(0,1)$, $\alpha\in\Delta(A)$ and ${y}\in{Y}$. Let $x_a:=\alpha(a)e^{u(a,{y})}$ for $a\in A$. Then
\begin{equation*}
G_\mu(\alpha,{y})=\log\Big(\big(\sum_{a\in A}x_a^{\frac{1}{1-\mu}}\big)^{1-\mu}\Big)
 \le 
\log\Big(\sum_{a\in A}x_a\Big)=G_0(\alpha,{y}),
\end{equation*}
with strict inequality whenever $x_a>0$ for at least two actions.
\end{lemma}

\begin{proof}[Proof of Lemma \ref{lem:growth_contraction}]
Let $s:=\frac{1}{1-\mu}>1$. By Lemma \ref{lem:growth_closed_form},
$G_\mu(\alpha,{y})=\frac{1}{s}\log(\sum_a x_a^s)=\log((\sum_a x_a^s)^{1/s})$ and $G_0(\alpha,{y})=\log(\sum_a x_a)$.
For nonnegative $(x_a)_a$, the $\ell^s$-norm is dominated by the $\ell^1$-norm:
$(\sum_a x_a^s)^{1/s}\le \sum_a x_a$. To verify this when at least two $x_a$ are positive, pick any two indices with $x_i,x_j>0$ and set $t:=x_j/x_i>0$. Then,
$\big(\frac{x_i^s+x_j^s}{x_i^s}\big)^{1/s}=(1+t^s)^{1/s}
 < 1+t=\frac{x_i+x_j}{x_i},$
where the strict inequality follows from $1+t^s<(1+t)^s$ for $t>0$ and $s>1$ since $f(t):=(1+t)^s-t^s$ satisfies $f'(t)=s[(1+t)^{s-1}-t^{s-1}]>0$ and $f(0)=1$. Multiplying by $x_i$ gives
$(x_i^s+x_j^s)^{1/s}<x_i+x_j.$
Repeatedly applying this two-variable inequality to the remaining nonnegative components
(equivalently, by induction on the number of positive components) yields
$\big(\sum_{a\in A} x_a^s\big)^{1/s}<\sum_{a\in A} x_a$
whenever at least two components are positive.
\end{proof}

\begin{proof}[Proof of Proposition \ref{prop:growth_misspec_fails}]
We construct an explicit two-state example satisfying Assumption \ref{ass:rc2} and show that $L_\mu(p,p')<L_0(p,p')$ for every $\mu\in(0,1)$.

\smallskip\noindent\emph{Step 1: primitives and Assumption \ref{ass:rc2}.}
Let ${Y}=\{{y}_1,{y}_2\}$ and $A=\{1,2\}$. Let the gross returns be $e^{u(1,{y}_1)}=3,$ $ e^{u(2,{y}_1)}=1,$ $
e^{u(1,{y}_2)}=1,$ $ e^{u(2,{y}_2)}=4,$
and let the true prior be $p({y}_1)=p({y}_2)=\frac12$.
For $\mu=0$, the benchmark objective \eqref{eq:growth22_problem} becomes, writing $\alpha:=\alpha(1)\in[0,1]$,
$$
\E_p\log\Big(\sum_{a}\alpha(a)e^{u(a,{y})}\Big)
=\frac12\log\big(1+2\alpha\big)+\frac12\log\big(4-3\alpha\big).
$$
This is strictly concave on $(0,1)$ since
$$
\frac{d^2}{d\alpha^2}\Big[\tfrac12\log(1+2\alpha)+\tfrac12\log(4-3\alpha)\Big]
=-\frac{2}{(1+2\alpha)^2}-\frac{9}{2(4-3\alpha)^2}<0,
$$
hence it has a unique maximizer $\alpha_0^*(p)(1)=\frac{5}{12}\in(0,1)$.
Using \eqref{eq:sampled_choice_rule}-\eqref{eq:posterior_dollar} and $p({y}_1)=p({y}_2)=\frac12$,
$q_0^*(1\mid {y}_1)=\frac{\frac{5}{12}\cdot 3}{\frac{5}{12}\cdot 3+\frac{7}{12}\cdot 1}=\frac{15}{22}$ and $q_0^*(1\mid {y}_2)=\frac{\frac{5}{12}\cdot 1}{\frac{5}{12}\cdot 1+\frac{7}{12}\cdot 4}=\frac{5}{33}.$
Therefore, by Bayes rule \eqref{eq:posterior_dollar},
$q_0^*({y}_1\mid 1)
=\frac{\frac12\cdot \frac{15}{22}}{\frac12\cdot \frac{15}{22}+\frac12\cdot \frac{5}{33}}
=\frac{\frac{15}{22}}{\frac{15}{22}+\frac{5}{33}}
=\frac{9}{11}.$
Define the misspecified prior $p'\in\Delta({Y})$ by $p'({y}_1)=\frac{9}{11}$ and $p'({y}_2)=\frac{2}{11}$. Then, $p'=q_0^*(\cdot\mid 1)$ is one of the posteriors induced by the benchmark optimizer under $p$, hence $p'$ lies in the convex hull of $\{q_0^*(\cdot\mid a):q_0^*(a)>0\}$. Thus, Assumption \ref{ass:rc2} holds.

\smallskip\noindent\emph{Step 2: benchmark loss equals Kullback-Leibler divergence.}
By Lemma \ref{lem:growth_closed_form}, $V_0(\cdot)$ coincides with the benchmark problem in \eqref{eq:growth22_problem} or \citet[][eq. (7)]{growth22}. Since Assumption \ref{ass:rc2} holds for $(p,p')$ by Step 1, \citet[][Proposition 4]{growth22} applies and yields
\begin{equation}\label{eq:bench_KL}
L_0(p,p')=R(p\Vert p')=\frac12\log\frac{1/2}{9/11}+\frac12\log\frac{1/2}{2/11}=\frac12\log\frac{121}{72}.
\end{equation}

\smallskip\noindent\emph{Step 3: for $\mu>0$, loss is strictly smaller than benchmark.}
Fix $\mu\in(0,1)$. By Lemma \ref{lem:growth_contraction}, for every mixed strategy $\alpha$ and both states ${y}_1,{y}_2$,
$G_\mu(\alpha,{y})\le G_0(\alpha,{y})$, with equality only for degenerate mixed strategies (since all returns are strictly positive).
Taking expectations under $p$ and maximizing over $\alpha$ gives
$V_\mu(p)=\max_{\alpha}\E_p[G_\mu(\alpha,{y})]
 \le 
\max_{\alpha}\E_p[G_0(\alpha,{y})]
=V_0(p).$
 We claim that this inequality is strict. Toward contradiction, suppose $V_\mu(p)=V_0(p)$ held. Let $\widehat\alpha\in\arg\max_{\alpha}\E_p[G_\mu(\alpha,{y})]$ be an optimizer at $\mu$. Then,
$V_0(p)
=V_\mu(p)
=\E_p[G_\mu(\widehat\alpha,{y})]
\le \E_p[G_0(\widehat\alpha,{y})]
\le V_0(p),$
so all inequalities bind. Binding of the first inequality forces $G_\mu(\widehat\alpha,{y})=G_0(\widehat\alpha,{y})$ for both ${y}_1$ and ${y}_2$, which, by Lemma \ref{lem:growth_contraction} and strict positivity of returns for both actions in both states, implies that $\widehat\alpha$ is degenerate. However, for any degenerate mixed strategy $\alpha=\delta_a$,
Lemma \ref{lem:growth_closed_form} implies $G_\mu(\delta_a,{y})=u(a,{y})=G_0(\delta_a,{y})$, so $\E_p[G_0(\delta_1,{y})]=\tfrac12\log 3$ and 
$\E_p[G_0(\delta_2,{y})]=\tfrac12\log 4.$
In contrast, the feasible diversified mixed strategy $\alpha(1)=\alpha(2)=\frac12$ delivers
$$
\E_p[G_0(\alpha,{y})]
=\frac12\log\Big(\frac{3+1}{2}\Big)+\frac12\log\Big(\frac{1+4}{2}\Big)
=\frac12\log 2+\frac12\log\frac52
=\frac12\log 5
>\frac12\log 4,
$$
so no degenerate mixed strategy can maximize $V_0(p)$. This contradicts the conclusion that $\widehat\alpha$ must be degenerate. Hence, $V_\mu(p)<V_0(p)$. Next, we show that the optimal mixed strategy under $p'$ is degenerate for all $\mu\in(0,1)$. Since $p'({y}_1)=\frac{9}{11}$ is sufficiently tilted toward ${y}_1$, the benchmark optimizer at $\mu=0$ is $\alpha_0^*(p')=\delta_1$. To see this, given $p'=(9/11,2/11)$, the objective is $\frac{9}{11}\log(1+2\alpha)+\frac{2}{11}\log(4-3\alpha)$, which is increasing in $\alpha\in[0,1]$ and strictly increasing for $\alpha<1$, so $\alpha=1$ is the unique maximizer. Now fix any $\mu\in(0,1)$. By Lemma \ref{lem:growth_contraction},
$\E_{p'}[G_\mu(\alpha,{y})]\le \E_{p'}[G_0(\alpha,{y})]$ for all $\alpha$, while $\E_{p'}[G_\mu(\delta_1,{y})]=\E_{p'}[G_0(\delta_1,{y})]$ because $\delta_1$ is degenerate. Therefore, $\E_{p'}[G_\mu(\alpha,{y})]
\le \E_{p'}[G_0(\alpha,{y})]
\le \E_{p'}[G_0(\delta_1,{y})]
=\E_{p'}[G_\mu(\delta_1,{y})],$
so $\alpha_\mu^*(p')=\delta_1$ is optimal for every $\mu\in(0,1)$. Combining, for every $\mu\in(0,1)$, we have
$$
L_\mu(p,p')
=\E_p[G_\mu(\alpha_\mu^*(p),{y})]-\E_p[G_\mu(\delta_1,{y})]
=V_\mu(p)-\tfrac12\log 3
 < 
V_0(p)-\tfrac12\log 3
=L_0(p,p').
$$
Using \eqref{eq:bench_KL}, this implies $L_\mu(p,p')<R(p\Vert p')$ for every $\mu\in(0,1)$. 
\end{proof}

\end{document}